\documentclass[reprint,amsmath,amssymb,aps,prx,floatfix]{revtex4-2}

\usepackage{amsthm}
\usepackage{graphicx}
\usepackage{xcolor}
\usepackage[bottom]{footmisc}
\usepackage{booktabs}
\usepackage{url}
\usepackage{pgfplots}
\usepgfplotslibrary{groupplots}
\pgfplotsset{compat=1.18}

\newtheorem{definition}{Definition}
\newtheorem{theorem}{Theorem}
\newtheorem{lemma}{Lemma}
\newtheorem{proposition}{Proposition}
\newtheorem{remark}{Remark}
\newtheorem{corollary}{Corollary}
\begin{document}

\title{Multiscale Schmidt-Spectrum Bounds for High-Dimensional Entanglement: Geometric Measures and Schmidt-Number Witnesses}

\author{Liang Xiong}
\email{Contact author: xiongliang199@163.com}

\author{Zhixiang Jin}
\email{jzxjinzhixiang@126.com}

\author{Yanling Wang}
\email{ylwang@dgut.edu.com}

 \author{Wei Chen}
 \email{2016025@dgut.edu.cn}

\author{Hong Tao}
 \email{taohong@dgut.edu.cn}

\affiliation{School of Computer Science and Technology, Dongguan University of
Technology, Dongguan, China}

\author{Nung-sing Sze}
\thanks{Corresponding author: raymond.sze@polyu.edu.hk}
\affiliation{Department of Applied Mathematics, The Hong Kong Polytechnic
University, Hung Hom, Hong Kong, China}

\date{\today}

\begin{abstract}
High-dimensional bipartite entanglement depends on how probability is
distributed across the Schmidt spectrum, whereas a single reference-state
fidelity resolves only one spectral scale. We develop multiscale
Schmidt-spectrum bounds that connect geometric measures with quantitative
Schmidt-number witnesses. For any partition of a pure-state Schmidt spectrum,
several nested Vidal tails determine block masses. We derive the sharp upper
boundary of the associated normalized nuclear-norm coordinate and show that
equality holds if and only if the spectrum is uniform within each block.
Refining the partition gives a monotone hierarchy of tighter bounds whenever
the added tail data distinguish unequal block means. We also solve a relaxed
weighted multiscale optimization globally: one scalar parameter specifies its
unique full-support optimizer and explicit value on the nontrivial branch.
Using the established single-tail fidelity--resource curve as a baseline, we
obtain exact-fidelity equality refinements, convex-roof lower bounds for mixed
states, and quantitative calibrations of Schmidt-number witnesses. A
higher-tail relation further bounds convex-roof extended negativity in terms of
Vidal tails and identifies the pure-state equality spectra. Leakage-aware and
joint-confidence formulations state how these bounds can be used with
incomplete data. Multistep tails require block-resolved or independently
certified spectral information; they are not determined by one projector
expectation.
\end{abstract}

\maketitle

\section{Introduction}

High-dimensional entanglement offers more than a larger state space. A single
carrier can encode several logical values, and the additional levels can improve
the information rate, noise tolerance, or resource efficiency of communication,
cryptographic, metrological, and processing tasks
\cite{Horodecki2009,quantuminformation,Bennett1996,Bennett1996a,
PhysRevA.88.032309,PhysRevA.71.044305,lanyon2009simplifying,
PhysRevLett.96.010401,PRXQuantum.4.040305,Luo2019Teleportation}. These gains
depend on genuinely high-dimensional entanglement rather than on a large local
Hilbert space alone. The relevant question is therefore quantitative: how much
Schmidt weight is supported beyond each low-rank sector?

Photonic experiments now address this question in orbital-angular-momentum,
path, pixel, time-bin, and frequency-bin encodings. Representative developments
include spatial-mode and multipath sources, multiphoton states, integrated
circuits, quantum memories, and metropolitan-scale time-frequency links
\cite{Mair2001,dada2011experimental,malik2016multi,Hu2020Multipath,
Valencia2020Pixel,Kues2017,Wang2018Science,Ding2016Memory,
Islam2017TimeBin,Chang2026TimeFrequency,erhard2020advances}. The local dimension
in such systems may be large, but complete tomography becomes expensive and
mode-dependent losses complicate interpretation. Certification methods based on
few bases, randomized measurements, calibrated witnesses, or device-independent
data address different parts of this problem
\cite{bavaresco2018measurements,Guo2018GMLE,PRXQuantum.4.020324,
LiHuberFriis2025,PhysRevLett.111.030501,PhysRevX.4.011011,
Zhang2019SelfTesting,WuLiZhu2026}. Their output is often a fidelity, a witness
value, or a dimension threshold rather than a complete Schmidt spectrum.

For a pure bipartite state, the ordered Schmidt spectrum records strictly more
information than its rank. Its tail sums govern stochastic single-copy
conversion and form a complete family of pure-state entanglement monotones
\cite{VidalSingleCopy,VidalMonotones,VidalJonathanNielsen2000}. Mixed states are
stratified by the Schmidt number, the least Schmidt rank required in a
pure-state decomposition. For
$m=\dim\mathcal H_B\leq\dim\mathcal H_A$, the corresponding compact convex sets
obey $\mathcal S_1\subset\cdots\subset\mathcal S_m$
\cite{PhysRevA.61.040301,PhysRevA.63.050301,
PhysRevLett.132.220203}. The convex-roof Vidal tail
$E_r=\sum_{i=r}^{m}\mu_i$, with
$\mu_1\geq\cdots\geq\mu_m$, vanishes exactly on $\mathcal S_{r-1}$. It also
coincides with the geometric overlap deficit from pure states of Schmidt rank at
most $r-1$. This places the Vidal hierarchy within the geometric-measure
framework initiated for pure and mixed entanglement
\cite{GEM1,GEM2,GEM3,UyanikTurgut2010,Blasone2008,
Demianowicz2021,ZhuZhangZeng2023,GirardGour2017}.

Convex roofs are generally difficult to evaluate. Symmetry reduction,
Legendre-transform methods, quantitative witnesses, and numerical convex-roof
procedures convert partial observables into bounds, but their tightness depends
on the state family and the information model
\cite{VollbrechtWerner2001,EisertBrandaoAudenaert2007,
GuhneReimpellWerner2007,GuhneReimpellWerner2008,
TothMoroderGuhne2015,tomography,tomography2}. A fidelity with one pure reference
is an especially useful scalar input. For convex roofs generated by pure free
states, $1-E_r(\rho)$ equals the maximum fidelity between $\rho$ and
$\mathcal S_{r-1}$ \cite{StreltsovKampermannBruss2010,Regula2018}; the
Bures-angle triangle inequality then gives the standard geometric lower bound.

The closest direct prior result is Proposition~3 of Wu, Li, and Zhu
\cite{WuLiZhu2026}. Their resource-constrained fidelity
$G_r(E)=\max_{\rho:E_r(\rho)\leq E}F_\phi(\rho)$ and the positive branch of our
fixed-fidelity value
$\mathcal R_{r,\phi}(F)=\inf_{\rho:F_\phi(\rho)=F}E_r(\rho)$ parametrize the same
single-tail boundary in inverse variables. We therefore treat that boundary,
including its proportional two-block optimizer, as established input. This
equivalence is confined to the single-scale problem. Our main questions use
several nested tails and concern a sharp blockwise spectral boundary, partition
refinement, and a relaxed weighted optimization; these inputs and conclusions
are not supplied by inversion of the one-tail curve.

One tail cannot specify how the remaining Schmidt weight is distributed. The
same is true of a normalized nuclear norm or the largest reference coefficient.
Several nested tails retain a coarse-grained spectral profile: their differences
are the masses of consecutive Schmidt blocks. Majorization and R\'enyi-entropy
methods provide the general language for extremal probability vectors
\cite{Bosyk2017,MarshallOlkinArnold2011,Sason2018}. Against this background, to
our knowledge, no previous result combines arbitrary nested Vidal tails with a
sharp normalized nuclear-norm boundary, an if-and-only-if blockwise equality
condition, and a partition-refinement criterion.

Our first main result states that, at fixed multistep tail profile, the maximal
normalized nuclear norm is obtained precisely by spectra that are uniform within
each chosen block. The boundary is sharp for every admissible profile. Supplying
an additional tail refines the partition and cannot weaken the bound; strict
tightening occurs exactly when the separated subblocks have unequal mean
probabilities. Singleton blocks recover the full spectrum. This hierarchy makes
the information gain explicit rather than attributing it to a single fidelity.

Nonnegative weights scalarize the tail vector. The resulting problem is a linear
cost minimized over a Bhattacharyya--Hellinger superlevel set of block masses.
Divergence-constrained linear optimization has a broader literature
\cite{BenTalEtAl2013}. We give a global analytic solution of the relaxed
optimization with weighted Vidal-tail costs. On the open nontrivial
branch, a unique scalar root determines the unique full-support optimizer and
its value. A global Cauchy--Schwarz certificate, rather than stationarity alone,
proves optimality. The exact fixed-fidelity convex roof of a general weighted
profile remains open.

The single-scale baseline also supports two equality refinements. A continuous
bounded-Schmidt-rank family realizes every exact fidelity on the zero branch,
whereas positive-branch equality fixes the ordered Schmidt spectrum. This is a
spectrum-level statement: degeneracies may permit Schmidt-basis rotations, and
it neither fixes a state vector nor classifies mixed-state equality cases.
Convexity additionally gives a Fenchel representation adapted to
incomplete-information inference.

These spectral statements have two quantitative consequences. First, the
projector support coefficient
$A_r=\sum_{i=1}^{r-1}\nu_i$ is both the Schmidt-number-witness threshold and the
activation point of the Vidal-tail bound. Hence a violation
$F_\phi(\rho)>A_r$ provides a resource value, not only a dimension label.
Bounded-Schmidt-rank support functions and rank-one $S(k)$ identities are
established \cite{GUHNE2009,norm1,johnston2012norms,PRA23NW,XiongSze2026}; our
contribution is their quantitative calibration by certified tail bounds. The
resulting nonprojector construction can be conservative and does not replace an
exact $S(k)$-norm calculation.

Second, the Vidal hierarchy can be compared with other entanglement measures.
Negativity and its convex-roof extension are established measures, and their
relations to concurrence and Schmidt-number bounds are well studied
\cite{Negativity02,CREN03,EltschkaSiewert2015,Regula2018,pra16zhang,
Concurrence00,Concurrence01,Concurrence04,Concurrence05,Wootters1998}. The
$k=2$ geometric-measure--CREN relation is known. For $k\geq3$, we derive an
explicit Vidal-tail upper envelope for CREN together with its analytic inverse
and the if-and-only-if pure-state equality spectrum. The result concerns
this combined higher-tail statement, not CREN itself or previously known
negativity bounds on Schmidt number.

The operational scope remains conditional. One reference-projector expectation
does not determine a multistep tail vector. Those inputs require block-resolved
measurements, an independently certified source model, or other spectral data.
For mixed states, fixed-subspace populations are not convex-roof tails. A global
projector may also require many local correlators. We therefore formulate
leakage-aware and joint-confidence bounds as an inference interface and describe
published photonic summaries only as screening inputs, not as raw-count
experimental validation.

Section II fixes notation. Section III states the attributed single-scale
baseline and then develops the multistep and weighted Schmidt-spectrum results,
including leakage and statistical consequences. Section IV converts certified
Vidal-tail values into quantitative Schmidt-number-witness coefficients.
Section V derives the negativity relation and records a nonoptimal concurrence
comparison. Section VI presents multiscale illustrations, mixed-state brackets,
and experimental interfaces. The appendices contain the longer proofs.

\section{Preliminaries}

Throughout, $\mathbb N_0$ and $\mathbb N_+$ denote the nonnegative and positive
integers. Let $\mathcal H_n$ be an $n$-dimensional complex Hilbert space,
$\mathcal L(\mathcal H_n)$ its operator space, and
$\mathcal D(\mathcal H_n)$ the positive semidefinite trace-one operators on
$\mathcal H_n$. Bipartite systems are finite dimensional and take the form
$\mathcal H_n\otimes\mathcal H_m$ with $m\leq n$. We write $X\succeq0$ for a
positive semidefinite operator and $\mathbb I$ for the relevant identity.

A state $\rho\in\mathcal D(\mathcal H_n\otimes\mathcal H_m)$ is separable if
it admits a decomposition
\[
  \rho=\sum_i p_i
  |a_i\rangle\langle a_i|\otimes|b_i\rangle\langle b_i|,
  \qquad
  p_i\geq0,\quad \sum_i p_i=1,
\]
with normalized local vectors $|a_i\rangle\in\mathcal H_n$ and
$|b_i\rangle\in\mathcal H_m$; otherwise, it is entangled. If
\[
\rho=\sum_{i,j=1}^{n}\sum_{k,l=1}^{m}
\rho_{ijkl}|i\rangle\langle j|\otimes|k\rangle\langle l|,
\]
its partial transpose with respect to $B$ is
\[
\rho^{T_B}
=\sum_{i,j=1}^{n}\sum_{k,l=1}^{m}
\rho_{ijkl}|i\rangle\langle j|\otimes|l\rangle\langle k|.
\]
Every separable state has a positive partial transpose
\cite{PPT1,PPT2}; hence $\rho^{T_B}\nsucceq0$ certifies entanglement.

We use the squared Uhlmann fidelity
\[
F(\rho,\sigma)
:=\left[\operatorname{Tr}
\sqrt{\sqrt{\rho}\sigma\sqrt{\rho}}\right]^2.
\]
For a pure state $\rho=|\psi\rangle\langle\psi|$,
$F(\rho,\sigma)=\langle\psi|\sigma|\psi\rangle$, and for two pure states this
reduces to $|\langle\psi|\phi\rangle|^2$. All geometric overlaps use this
convention.

\subsection{Schmidt number and its witnesses}

Every normalized pure state admits a Schmidt decomposition
\[
\begin{aligned}
|\psi\rangle&=\sum_{i=1}^{m}s_i|a_i b_i\rangle,\\
s_1&\geq\cdots\geq s_m\geq0,
\qquad \sum_{i=1}^{m}s_i^2=1.
\end{aligned}
\]
The local families are orthonormal; $s_i$ are the Schmidt amplitudes and
$\mu_i:=s_i^2$ the Schmidt probabilities, zero-padded to length $m$. The Schmidt rank
$\operatorname{SR}(|\psi\rangle)$ is the number of nonzero Schmidt
amplitudes.

For $\rho\in\mathcal D(\mathcal H_n\otimes\mathcal H_m)$, let
$\mathfrak E(\rho)$ denote the set of finite normalized pure-state ensembles
$\{p_i,|\psi_i\rangle\}$ satisfying
\[
\rho=\sum_i p_i|\psi_i\rangle\langle\psi_i|,
\qquad p_i\geq0,\qquad\sum_i p_i=1.
\]
The Schmidt number is
\[
\operatorname{SN}(\rho)
:=\min_{\{p_i,|\psi_i\rangle\}\in\mathfrak E(\rho)}
\max_{i:p_i>0}\operatorname{SR}(|\psi_i\rangle).
\]
$\operatorname{SN}(\rho)\leq k$ precisely when $\rho$ has a decomposition
whose nonzero components all have Schmidt rank at most $k$.

We distinguish the normalized pure-state set from its mixed-state convex hull:
\begin{widetext}
\[
\begin{aligned}
\mathcal V_k
&:=
\bigl\{
|\psi\rangle:
\langle\psi|\psi\rangle=1,
\ \operatorname{SR}(|\psi\rangle)\leq k
\bigr\},\\
\mathcal S_k
&:=
\bigl\{
\rho\in\mathcal D(\mathcal H_n\otimes\mathcal H_m):
\operatorname{SN}(\rho)\leq k
\bigr\}
=
\operatorname{conv}
\bigl\{
|\psi\rangle\langle\psi|:
|\psi\rangle\in\mathcal V_k
\bigr\}.
\end{aligned}
\]
\end{widetext}

The sets $\mathcal S_k$ are compact and convex, with
$\mathcal S_{k-1}\subseteq\mathcal S_k$. For the maximally entangled state
\[
|\Phi_m^+\rangle=\frac{1}{\sqrt m}\sum_{i=1}^{m}|ii\rangle,
\]
every $\rho_k\in\mathcal S_k$ satisfies
\[
\langle\Phi_m^+|\rho_k|\Phi_m^+\rangle\leq\frac{k}{m}.
\]
This threshold follows from the maximal squared overlap of
$|\Phi_m^+\rangle$ with a Schmidt-rank-$k$ pure state
\cite{PhysRevA.61.040301,PhysRevA.59.4206}.

We use a threshold-indexed witness convention throughout.

\begin{definition}[$k$-Schmidt-number witness]
Let $2\leq k\leq m$. A Hermitian operator $\mathcal{SW}_k$ is a
$k$-Schmidt-number witness if
\[
\operatorname{Tr}(\mathcal{SW}_k\sigma)\geq0
\qquad\text{for every }\sigma\in\mathcal S_{k-1},
\]
and if there exists a density operator $\rho\notin\mathcal S_{k-1}$ such that
\[
\operatorname{Tr}(\mathcal{SW}_k\rho)<0.
\]
Every negative expectation value certifies
$\operatorname{SN}(\rho)\geq k$. The detected state need not belong to
$\mathcal S_k$.
\end{definition}

Because $\mathcal S_1$ is the set of separable states, a
$2$-Schmidt-number witness is an entanglement witness. For a Hermitian operator
$X$, define the bounded-Schmidt-rank support function
\[
h_k(X):=\max_{|v\rangle\in\mathcal V_k}\langle v|X|v\rangle
=\max_{\sigma\in\mathcal S_k}\operatorname{Tr}(X\sigma).
\]
By definition of $h_{k-1}$, the operator $h_{k-1}(X)\mathbb I-X$ is nonnegative
on $\mathcal S_{k-1}$. It becomes a $k$-Schmidt-number witness precisely when
it is not positive semidefinite, or equivalently when
$\lambda_{\max}(X)>h_{k-1}(X)$.

\subsection{$k$-positive maps and $S(k)$-norms}

Let
$\Phi:\mathcal L(\mathcal H_n)\rightarrow\mathcal L(\mathcal H_m)$ be linear.
Its Hilbert--Schmidt adjoint is determined by
\[
\operatorname{Tr}\!\left[\Phi(X)^\dagger Y\right]
=\operatorname{Tr}\!\left[X^\dagger\Phi^\dagger(Y)\right].
\]
The map is Hermiticity preserving if it maps Hermitian operators to Hermitian
operators, positive if $\Phi(X)\succeq0$ whenever $X\succeq0$, and $k$-positive
if $\operatorname{id}_k\otimes\Phi$ is positive. Complete
positivity requires $k$-positivity for every $k\in\mathbb N_+$. In finite
dimensions, $n$-positivity already suffices.

With
$|\Phi_n^+\rangle=n^{-1/2}\sum_{i=1}^{n}|e_i e_i\rangle$, define the Choi
operator
\[
C_\Phi
:=(\operatorname{id}_n\otimes\Phi)
\bigl(|\Phi_n^+\rangle\langle\Phi_n^+|\bigr).
\]
Choi's theorem states that $\Phi$ is completely positive if and only if
$C_\Phi\succeq0$. Under the Choi--Jamio{\l}kowski correspondence,
Hermiticity-preserving maps correspond to Hermitian operators, and
$k$-positive maps correspond to $k$-block-positive operators.

\begin{lemma}{\rm\cite{kblockpositive}}\label{lemkbp}
Let $X=X^\dagger\in\mathcal L(\mathcal H_n\otimes\mathcal H_m)$. Then $X$
is $k$-block positive if and only if
\begin{equation}\label{kbpeq01}
  \operatorname{Tr}(X\rho)\geq 0
  \quad\text{for every density operator }\rho\in\mathcal S_k.
\end{equation}
\end{lemma}

Lemma~\ref{lemkbp} expresses the duality between $k$-block-positive operators
and states of Schmidt number at most $k$. The relevant support function for
positive operators is the $S(k)$-norm.

\begin{definition}[Vector $s(k)$-norm and operator $S(k)$-norm \cite{norm1}]
Let $|v\rangle\in\mathcal H_n\otimes\mathcal H_m$ be normalized and let
$s_1\geq\cdots\geq s_m\geq0$ be its Schmidt amplitudes. For
$1\leq k\leq m\leq n$, define
\[
\||v\rangle\|_{s(k)}
:=\sup_{|w\rangle\in\mathcal V_k}|\langle w|v\rangle|
=\left(\sum_{i=1}^k s_i^2\right)^{1/2}.
\]
For $X\in\mathcal L(\mathcal H_n\otimes\mathcal H_m)$, define
\[
\|X\|_{S(k)}
:=\sup_{|v\rangle,|w\rangle\in\mathcal V_k}
|\langle w|X|v\rangle|.
\]
\end{definition}

For a general Hermitian operator, the positive support $h_k(X)$ differs from
the absolute two-vector norm $\|X\|_{S(k)}$. Normality does not reduce the
two-vector optimization to quadratic forms. We use only the positive-operator
case below.

\begin{lemma}[Positive-operator support \cite{norm1}]\label{lem4}
If $X\succeq0$, then
\[
\|X\|_{S(k)}
=h_k(X)
=\max_{|v\rangle\in\mathcal V_k}\langle v|X|v\rangle
=\max_{\rho\in\mathcal S_k}\operatorname{Tr}(X\rho).
\]
For $\lambda\in\mathbb R$, this identity implies that the operator
$\lambda\mathbb I-X$ is $k$-block positive if and only if
$\lambda\geq\|X\|_{S(k)}$. Under this condition, it is a nontrivial
$(k+1)$-Schmidt-number witness precisely when
$\lambda<\lambda_{\max}(X)$.
\end{lemma}

For a positive reference operator, the exact $S(k)$-norm is its smallest
block-positive coefficient. A certified upper bound gives a valid, generally
nonoptimal, coefficient. Exact values and estimates are known for some
structured operators \cite{norm1,johnston2012norms}.

\subsection{Bounded-Schmidt-rank geometric measures}

For a normalized pure state $|\psi\rangle$ and $2\leq r\leq m$, define
\[
E_r(|\psi\rangle)
:=1-\max_{|\varphi\rangle\in\mathcal V_{r-1}}
|\langle\varphi|\psi\rangle|^2.
\]
The case $r=2$ is the conventional geometric measure of bipartite
entanglement, denoted by $E_G$. If $|\psi\rangle$ has ordered Schmidt
amplitudes $s_1\geq\cdots\geq s_m$, the best rank-$(r-1)$ approximation gives
\begin{equation}\label{RGMeq1}
  E_r(|\psi\rangle)
  =1-\sum_{i=1}^{r-1}s_i^2
  =\sum_{i=r}^{m}s_i^2.
\end{equation}
Equation~\eqref{RGMeq1} identifies $E_r$ as the $r$th Vidal tail monotone, with
the index convention used in
Refs.~\cite{VidalSingleCopy,VidalMonotones}. Its equivalent
bounded-Schmidt-rank geometric form was established in
Refs.~\cite{UyanikTurgut2010,Demianowicz2021}; see also the broader geometric
hierarchies in Refs.~\cite{Blasone2008,ZhuZhangZeng2023}. It vanishes exactly
when
$\operatorname{SR}(|\psi\rangle)\leq r-1$. The symbol $E_r$ exposes the rank
threshold; it does not denote a new monotone.

For a mixed bipartite state, we use the convex-roof extension
\begin{equation*}
  E_r(\rho)=
  \inf_{\{p_i,|\psi_i\rangle\}\in\mathfrak E(\rho)}
  \sum_i p_i E_r(|\psi_i\rangle).
\end{equation*}
In finite dimensions, continuity, compactness, and Carath\'eodory's theorem
give a finite minimizing ensemble, whose size may depend on dimension. Hence
\[
E_r(\rho)=0
\quad\Longleftrightarrow\quad
\operatorname{SN}(\rho)\leq r-1.
\]
An attaining zero-average ensemble contains only states satisfying
$E_r(|\psi_i\rangle)=0$, each of Schmidt rank at most $r-1$. The Schmidt-number
definition gives the converse.
The nested reference sets yield
\[
E_{2}(\rho)\ge E_{3}(\rho)\ge \cdots \ge E_{m}(\rho)\ge 0.
\]
The family $\{E_r\}_{r=2}^{m}$ resolves the nested sets
$\mathcal S_{r-1}$. Each pure-state member is an ordered spectrum tail, and its
mixed-state counterpart is a convex roof. Existing work evaluates these roofs
on symmetric families, including isotropic states
\cite{VollbrechtWerner2001,GirardGour2017}, and has used related geometric
quantities for entangled subspaces \cite{Demianowicz2021,ZhuZhangZeng2023}.
We seek lower bounds for arbitrary mixed states from a specified
reference-projector expectation.

\section{Single-scale baseline and multiscale Schmidt-spectrum bounds}

We first record the known single-tail characteristic curve and its index map to
Ref.~\cite{WuLiZhu2026}. This baseline is followed by the multistep boundary and
weighted optimization, which are the main results of this section. Ambient-space,
leakage, measurement, and confidence consequences are stated after the
multiscale analysis.

\subsection{Single-scale Vidal-tail characteristic curve: prior result and fixed-fidelity formulation}

Let $|\phi\rangle$ be a pure reference with ordered Schmidt probabilities
$\nu_1\geq\cdots\geq\nu_m\geq0$, and define
\[
A_r:=\sum_{i=1}^{r-1}\nu_i,
\qquad
B_r:=\sum_{i=r}^{m}\nu_i=1-A_r.
\]
Wu, Li, and Zhu use zero-based Schmidt indices and denote our $E_r$ by
$\mathcal E_{r-1}$ \cite{WuLiZhu2026}. With this index shift, their
Proposition~3 defines, for $0\leq t\leq1$,
\[
G_{r,\phi}(t):=
\max_{\rho:E_r(\rho)\leq t}F_\phi(\rho),
\]
and gives the characteristic curve
\begin{equation}
G_{r,\phi}(t)=
\begin{cases}
\bigl[\sqrt{B_rt}+\sqrt{A_r(1-t)}\bigr]^2,&t<B_r,\\[0.4em]
1,&t\geq B_r.
\end{cases}
\label{eq:wlz-characteristic-curve}
\end{equation}
The same maximum is attained by a pure state, and their Eq.~(63) gives the
proportional two-block optimizer. The positive fixed-fidelity branch below is
the algebraic inverse of Eq.~\eqref{eq:wlz-characteristic-curve}. We use this
inverse only as a single-scale baseline. The multistep theorem instead fixes a
vector of nested tails and determines its blockwise equality structure and
refinement behavior; the weighted theorem solves a separate relaxed
scalarization. Neither problem follows from reparametrizing the one-tail curve.

Let $\mathcal X_{r-1}$ contain the normalized pure states of Schmidt rank at
most $r-1$, so that $\mathcal S_{r-1}=\operatorname{conv}\mathcal X_{r-1}$.
The following established convex-roof identity is specialized to this hierarchy
\cite{StreltsovKampermannBruss2010,Regula2018}.

\begin{lemma}[Known convex-roof--fidelity identity]
\label{lem:convex-roof-fidelity-identity}
For every bipartite density operator,
\begin{equation}
 E_r(\rho)=1-\max_{\sigma\in\mathcal S_{r-1}}F(\rho,\sigma).
\label{eq:convex-roof-fidelity-identity}
\end{equation}
\end{lemma}

Indeed,
\begin{align*}
\max_{\sigma\in\mathcal S_{r-1}}F(\rho,\sigma)
&=\max_{\rho=\sum_jp_j|\psi_j\rangle\langle\psi_j|}
  \sum_jp_j\max_{\chi_j\in\mathcal X_{r-1}}
  |\langle\chi_j|\psi_j\rangle|^2\\
&=1-\min_{\rho=\sum_jp_j|\psi_j\rangle\langle\psi_j|}
  \sum_jp_jE_r(|\psi_j\rangle),
\end{align*}
which gives Eq.~\eqref{eq:convex-roof-fidelity-identity} without a separate
duality assumption.

For squared Uhlmann fidelity, define the Bures angle and distance to the free
set by
\[
\begin{aligned}
d_B(\rho,\sigma)&:=\arccos\sqrt{F(\rho,\sigma)},\\
d_B(\rho,\mathcal S_{r-1})
&:=\min_{\sigma\in\mathcal S_{r-1}}d_B(\rho,\sigma).
\end{aligned}
\]
Lemma~\ref{lem:convex-roof-fidelity-identity} implies
\begin{equation}
d_B(\rho,\mathcal S_{r-1})=\arcsin\sqrt{E_r(\rho)}.
\label{eq:bures-tail-identity}
\end{equation}
Since
$d_B(\phi,\mathcal S_{r-1})=\arccos\sqrt{A_r}$, the triangle inequality gives:

\begin{proposition}[General Bures-angle lower bound]
\label{prop:general-bures-bound}
For every density operator $\rho$,
\begin{equation}
E_r(\rho)\geq
\sin^2\!\left(
\left[\arccos\sqrt{A_r}
-\arccos\sqrt{F_\phi(\rho)}\right]_+
\right).
\label{eq:general-bures-bound}
\end{equation}
\end{proposition}

\begin{proof}
For every $\sigma\in\mathcal S_{r-1}$, the Bures-angle triangle inequality
gives
\[
d_B(\rho,\sigma)\geq d_B(\phi,\sigma)
-\arccos\sqrt{F_\phi(\rho)}.
\]
Minimizing over $\sigma$, taking the positive part, and using
Eq.~\eqref{eq:bures-tail-identity} proves Eq.~\eqref{eq:general-bures-bound}
because $\sin^2$ is increasing on $[0,\pi/2]$.
\end{proof}

Equation~\eqref{eq:general-bures-bound} is established geometry and reproduces
the inverse functional form of Eq.~\eqref{eq:wlz-characteristic-curve}. We use
it as an alternative derivation and as a convenient route to leakage and
confidence bounds. Neither the function nor its positive-branch attainability
is claimed as a result of this manuscript.

The Vidal--Jonathan--Nielsen problem instead fixes a source spectrum
$\boldsymbol\alpha$ and a target spectrum $\boldsymbol\nu$, then maximizes
target fidelity over output spectra subject to
all majorization constraints \cite{VidalJonathanNielsen2000}.
Majorization-lattice methods retain that fixed-source order structure
\cite{Bosyk2017}; our fidelity section fixes no source spectrum. General
incomplete-information methods bound a convex
measure through a Legendre conjugate
\cite{EisertBrandaoAudenaert2007,GuhneReimpellWerner2007,GuhneReimpellWerner2008};
we derive ours from the closed section value. The projector estimate of
Ref.~\cite{pra16zhang} replaces the reference spectrum by its largest value
$\nu_1$. These structural comparisons do not establish reducibility or
nonreducibility. Section~IV uses established support and witness identities
\cite{norm1,XiongSze2026} only for calibration.

\subsubsection{Known maximally entangled calibration in the present notation}

Let $|\psi\rangle\in\mathcal H_n\otimes\mathcal H_m$, with $m\leq n$, have
ordered Schmidt probabilities
$\mu_1\geq\cdots\geq\mu_m\geq0$ and $\sum_i\mu_i=1$. Equation~\eqref{RGMeq1}
gives, for $2\leq r\leq m$,
\begin{equation*}
    E_r(|\psi\rangle)
    =
    1-\sum_{i=1}^{r-1}\mu_i
    =
    \sum_{i=r}^{m}\mu_i. 
\end{equation*}
Define the scalar spectrum parameter
\begin{equation}\label{eq:lambda-general-definition}
    \lambda(\boldsymbol{\mu})
    =
    \frac{1}{m}
    \left(
        \sum_{i=1}^{m}\sqrt{\mu_i}
    \right)^2,
    \qquad
    \frac{1}{m}\le\lambda\le1.
\end{equation}
Product and maximally entangled states attain the endpoints. Equivalently,
\(\lambda\) is the normalized squared nuclear norm of the coefficient matrix of
\(|\psi\rangle\). Under the Sec.~V convention
\cite{Negativity02}, pure states satisfy
\[
\lambda(\psi)=\frac{\mathcal N_{\rm PT}(|\psi\rangle\langle\psi|)+1}{m}.
\]
Thus \(\lambda\) reparametrizes pure-state negativity but neither defines a
measure nor fixes the Schmidt spectrum. It fixes the lower envelopes below for
fixed $m$, ambient dimension, and zero-padding convention.

Set
\begin{equation*}
    p=r-1,
    \qquad
    q=m-r+1,
    \qquad
    p+q=m.  
\end{equation*}

For fixed \(\lambda\), minimize \(E_r\) over ordered Schmidt spectra:

\begin{widetext}
   \begin{equation}\label{eq:Fr-definition-main}
    \mathcal{F}_r(\lambda)
    =
    \min_{\substack{
        \boldsymbol{\mu}\in\mathbb R^m\\
        \mu_1\geq\mu_2\geq\cdots\geq\mu_m\geq0
    }}
    \left\{
        1-\sum_{i=1}^{r-1}\mu_i
        \;\middle|\;
        \sum_{i=1}^{m}\mu_i=1,
        \frac{1}{m}
        \left(
            \sum_{i=1}^{m}\sqrt{\mu_i}
        \right)^2
        =\lambda
    \right\}.
\end{equation} 
\end{widetext}

The function \(\mathcal{F}_r\) is the lower boundary of the pure-state
\((\lambda,E_r)\) region.

\begin{lemma}[Known isotropic calibration in the present notation]\label{thm:exact-Fr-main}
Fix finite local dimensions $m\leq n$. Pad each Schmidt spectrum to length $m$
and order it nonincreasingly. Let \(2\le r\le m\). For every
\(\lambda\in[1/m,1]\), the constrained minimum in
\eqref{eq:Fr-definition-main} equals
  \begin{widetext}
    \begin{equation}\label{eq:Fr-exact-main}
    \mathcal{F}_r(\lambda)
    =
    \begin{cases}
        0,
        &
        \displaystyle
        \frac{1}{m}
        \le \lambda
        \le
        \frac{r-1}{m},
        \\[1.2em]
        \displaystyle
        \frac{1}{m}
        \left[
            \sqrt{(m-r+1)\lambda}
            -
            \sqrt{(r-1)(1-\lambda)}
        \right]^2,
        &
        \displaystyle
        \frac{r-1}{m}
        <
        \lambda
        \le 1.
    \end{cases}
\end{equation}
\end{widetext}

When \(\lambda>(r-1)/m\), the minimum is attained by the
two-block Schmidt spectrum
\begin{equation}\label{eq:two-block-optimal-spectrum-main}
    \boldsymbol{\mu}_{\mathrm{opt}}
    =
    \left(
        \underbrace{
        \frac{1-t}{r-1},\ldots,
        \frac{1-t}{r-1}
        }_{r-1\ \mathrm{entries}},
        \underbrace{
        \frac{t}{m-r+1},\ldots,
        \frac{t}{m-r+1}
        }_{m-r+1\ \mathrm{entries}}
    \right),  
\end{equation}
where
\begin{equation*}
    t=\mathcal{F}_r(\lambda).
\end{equation*}
Among ordered spectra, Eq.~\eqref{eq:two-block-optimal-spectrum-main} uniquely
specifies the nonzero-branch optimizer. The pure state is unique only up to
local isometries and rotations within degenerate Schmidt subspaces.
\end{lemma}

The second branch of \eqref{eq:Fr-exact-main} can be written as
\begin{small}
\begin{equation}
    \mathcal{F}_r(\lambda)
    =
    1-
    \frac{1}{m}
    \left[
        \sqrt{(r-1)\lambda}
        +
        \sqrt{(m-r+1)(1-\lambda)}
    \right]^2.
    \label{eq:Fr-equivalent-main}
\end{equation}
\end{small}

Appendix A proves Lemma~\ref{thm:exact-Fr-main}. After an index shift, the same
function gives the known isotropic-state Vidal convex roofs
\cite{GirardGour2017}; we do not claim a new measure or isotropic-state
evaluation.

\begin{remark}
For \(r=2\), one has \(p=1\) and \(q=m-1\). The optimal spectrum reduces to
\begin{equation*}
    \boldsymbol{\mu}_{\mathrm{opt}}
    =
    \left(
        1-t,
        \frac{t}{m-1},
        \ldots,
        \frac{t}{m-1}
    \right),
\end{equation*}
It has one distinguished probability and a uniform residual sector, with
\begin{equation*}
    \mathcal{F}_2(\lambda)
    =
    1-
    \frac{1}{m}
    \left[
        \sqrt{\lambda}
        +
        \sqrt{(m-1)(1-\lambda)}
    \right]^2,
\end{equation*}
which recovers the geometric-entanglement boundary.
\end{remark}

The mixed-state bounds use this pure-state boundary without assuming
mixed-state attainability.

\subsubsection{Convexity and convex-roof extension}

The pure-state boundary needs no convexification.

\begin{proposition}\label{prop:Fr-convex-main}
The function \(\mathcal{F}_r(\lambda)\) is continuous, nondecreasing, and convex
on \([1/m,1]\). Its lower convex envelope is unchanged:
\begin{equation*}
    \operatorname{co}
    \big[
        \mathcal{F}_r(\lambda)
    \big]
    =
    \mathcal{F}_r(\lambda),
\end{equation*}
where \(\operatorname{co}\) denotes the lower convex envelope.
\end{proposition}

\begin{proof}
The function vanishes on \([1/m,p/m]\). On \((p/m,1)\),
\begin{equation*}
    \mathcal{F}_r(\lambda)
    =
    \frac{
        p+(q-p)\lambda
        -
        2\sqrt{pq\,\lambda(1-\lambda)}
    }{m}.
\end{equation*}
Its second derivative is
\begin{equation*}
    \mathcal{F}_r''(\lambda)
    =
    \frac{\sqrt{pq}}
    {2m[\lambda(1-\lambda)]^{3/2}}
    >0.
\end{equation*}
Since \(\mathcal F_r'((r-1)/m^+)=0\) and
\(\mathcal F_r''(\lambda)>0\) on the nonzero branch,
\(\mathcal F_r'(\lambda)\geq0\) there. The right derivative at
\(\lambda=p/m\) equals the zero derivative of the constant branch.
Continuity at the branch point completes the monotonicity and convexity proof.
\end{proof}

The boundary values are
\begin{equation*}
    \mathcal{F}_r\left(\frac{r-1}{m}\right)=0,
    \qquad
    \mathcal{F}_r(1)=\frac{m-r+1}{m}.
\end{equation*}
For fixed \(\lambda\), the hierarchy satisfies
\begin{equation*}
    \mathcal{F}_2(\lambda)
    \ge
    \mathcal{F}_3(\lambda)
    \ge
    \cdots
    \ge
\mathcal{F}_m(\lambda)
    \ge0.
\end{equation*}
The ordering follows because $E_r\geq E_{r+1}$ pointwise on the common
fixed-$\lambda$ feasible set. Taking the corresponding minima preserves it.

\subsubsection{Fixed-fidelity inverse and exact-fidelity refinements}

Choose local Schmidt isometries
$V_A:\mathbb C^m\to\mathcal H_n$ and
$V_B:\mathbb C^m\to\mathcal H_m$ such that
\[
|\phi\rangle
=(V_A\otimes V_B)\sum_{i=1}^{m}\sqrt{\nu_i}|ii\rangle.
\]
Write $F_\phi(\rho):=\langle\phi|\rho|\phi\rangle$.
The ordered reference spectrum implies
$A_r\geq(r-1)/m$ and $B_r\leq(m-r+1)/m$. For $0\leq F\leq1$, set
\begin{equation}\label{eq:reference-spectrum-boundary}
\mathcal R_{r,\phi}(F)
:=
\begin{cases}
0, & 0\leq F\leq A_r,\\[0.8em]
\bigl[\sqrt{B_rF}-\sqrt{A_r(1-F)}\bigr]^2,
& A_r<F\leq1.
\end{cases}
\end{equation}
When only the threshold matters, write $\mathcal R_A(F)$ for
Eq.~\eqref{eq:reference-spectrum-boundary} with $A_r=A$ and $B_r=1-A$.
Equivalently,
\[
\mathcal R_{r,\phi}(F)=
\sin^2\!\left(
\left[\arccos\sqrt{A_r}-\arccos\sqrt F\right]_+
\right),
\]
so Proposition~\ref{prop:general-bures-bound} is precisely
$E_r(\rho)\geq\mathcal R_{r,\phi}(F_\phi(\rho))$.

If $B_r=0$, or $\operatorname{SR}(|\phi\rangle)<r$, then $A_r=1$ and
$\mathcal R_{r,\phi}=0$. Formulas dividing by $B_r$ apply only to the nonzero
branch, where $F>A_r$ implies $B_r>0$.

\begin{proposition}[Inverse formulation of the established single-tail curve]
\label{thm:single-observable-Er-main}
Fix $2\leq r\leq m$ and the pure reference $|\phi\rangle$ above.

For every $F\in[0,1]$,
\[
\min_{\substack{|\psi\rangle\in\mathcal H_n\otimes\mathcal H_m\\
\langle\psi|\psi\rangle=1\\
|\langle\phi|\psi\rangle|^2=F}}
E_r(|\psi\rangle)
=\mathcal R_{r,\phi}(F).
\]
For $F>A_r$, this is the inverse of
Eq.~\eqref{eq:wlz-characteristic-curve}, and an aligned optimizer has the
ordered Schmidt probabilities
\begin{equation}\label{eq:reference-spectrum-optimizer}
\mu_i^*(F)=
\begin{cases}
\displaystyle\frac{1-t}{A_r}\nu_i, & 1\leq i\leq r-1,\\[0.9em]
\displaystyle\frac{t}{B_r}\nu_i, & r\leq i\leq m,
\end{cases}
\qquad
t=\mathcal R_{r,\phi}(F).
\end{equation}
This positive-branch optimizer is equivalent to Eq.~(63) of
Ref.~\cite{WuLiZhu2026}. For $0\leq F\leq A_r$, a pure state of Schmidt rank at
most $r-1$ attains the exact fidelity and the zero value. The function
$\mathcal R_{r,\phi}$ is continuous, nondecreasing, and convex on $[0,1]$.
Every
$\rho\in\mathcal D(\mathcal H_n\otimes\mathcal H_m)$ satisfies
\begin{equation}\label{eq:reference-spectrum-mixed-bound}
E_r(\rho)\geq
\mathcal R_{r,\phi}\bigl(F_\phi(\rho)\bigr).
\end{equation}
The right-hand side is positive if and only if $F_\phi(\rho)>A_r$.
\end{proposition}

\begin{proof}
Proposition~\ref{prop:general-bures-bound} gives the lower bound. For $F>A_r$,
the strict branch of $G_{r,\phi}$ in
Eq.~\eqref{eq:wlz-characteristic-curve} is invertible. Its pure optimizer has
$t=\mathcal R_{r,\phi}(F)$, which proves the positive-branch equality and
Eq.~\eqref{eq:reference-spectrum-optimizer}. An equivalent geometric
parametrization uses the normalized head and tail truncations
\[
|\chi_r\rangle=A_r^{-1/2}
\sum_{i=1}^{r-1}\sqrt{\nu_i}|ii\rangle,
\qquad
|\eta_r\rangle=B_r^{-1/2}
\sum_{i=r}^{m}\sqrt{\nu_i}|ii\rangle.
\]
Writing
\[
\alpha_r=\arccos\sqrt{A_r},
\qquad
\delta=\arccos\sqrt F,
\]
one has $0\leq\delta<\alpha_r$ and
$|\phi\rangle=\cos\alpha_r|\chi_r\rangle
+\sin\alpha_r|\eta_r\rangle$. The geodesic state
\begin{equation}
|\psi_F\rangle=
\cos(\alpha_r-\delta)|\chi_r\rangle
+\sin(\alpha_r-\delta)|\eta_r\rangle
\label{eq:positive-geodesic-state}
\end{equation}
satisfies
\[
\begin{aligned}
|\langle\phi|\psi_F\rangle|^2&=\cos^2\delta=F,\\
E_r(|\psi_F\rangle)&=\sin^2(\alpha_r-\delta)
=\mathcal R_{r,\phi}(F).
\end{aligned}
\]
This is the same proportional two-block optimizer in Bures-angle coordinates.
Since $t=\mathcal R_{r,\phi}(F)\leq B_r$, the inequalities
$(1-t)/A_r\geq t/B_r$ and $\nu_{r-1}\geq\nu_r$ preserve its ordering.

It remains to complete the exact-fidelity zero branch, which does not follow
from the maximum-fidelity statement alone. Define the normalized
rank-$(r-1)$ truncation
\[
|\chi\rangle=A_r^{-1/2}
\sum_{i=1}^{r-1}\sqrt{\nu_i}|ii\rangle.
\]
Its squared overlap with $|\phi\rangle$ is $A_r$. A fixed-point-free cyclic
shift $S$ on the second subsystem satisfies
$\langle\phi|(I\otimes S)|\chi\rangle=0$. Path connectedness of the unitary
group supplies a continuous path from $I$ to $S$, along which every state has
Schmidt rank at most $r-1$. Its squared overlap starts at $A_r$ and ends at $0$.
The intermediate value theorem covers all $F\in[0,A_r]$ and proves the zero
branch.

Direct differentiation on $A_r<F<1$ gives
\[
\mathcal R_{r,\phi}''(F)
=\frac{\sqrt{A_rB_r}}
{2[F(1-F)]^{3/2}}>0,
\]
The right derivative at $F=A_r$ is zero, making the extension
continuous, nondecreasing, and convex. Equation
\eqref{eq:reference-spectrum-mixed-bound} is Proposition
\ref{prop:general-bures-bound} in algebraic form.
\end{proof}

\paragraph*{Scope and provenance.}
The convex-roof--fidelity identity and Bures-angle function are established.
The positive characteristic boundary, pure optimizer, mixed-state sublevel
optimum, and finite-dimensional independence are due to Wu, Li, and Zhu
\cite{WuLiZhu2026}. Proposition~\ref{thm:single-observable-Er-main} records the
inverse fixed-fidelity form and adds the exact-fidelity pure completion of the
zero branch. It does not classify all equality states.

Equation~(63) of Ref.~\cite{WuLiZhu2026} exhibits a positive-branch optimizer
but does not state an if-and-only-if condition for its ordered spectrum. The
next proposition records that spectrum-level equality refinement.

\begin{proposition}[Uniqueness of the positive-branch ordered spectrum]
\label{prop:positive-spectrum-uniqueness}
Let $A_r<F\leq1$. Every pure state satisfying
$|\langle\phi|\psi\rangle|^2=F$ and
$E_r(|\psi\rangle)=\mathcal R_{r,\phi}(F)$ has the ordered Schmidt
probabilities in Eq.~\eqref{eq:reference-spectrum-optimizer}. This statement
does not assert uniqueness of the state vector or its Schmidt bases.
\end{proposition}

\begin{proof}
Let $\{\mu_i\}$ be the ordered Schmidt probabilities of a saturating state and
set $t=\sum_{i=r}^{m}\mu_i$. The von Neumann trace inequality
\cite{horn2012matrix} and Cauchy--Schwarz on the two blocks give
\[
\sqrt F\leq\sum_{i=1}^{m}\sqrt{\nu_i\mu_i}
\leq\sqrt{A_r(1-t)}+\sqrt{B_rt}=\sqrt F.
\]
The last equality uses $t=\mathcal R_{r,\phi}(F)$, so both blockwise
Cauchy--Schwarz inequalities must be equalities. The head vector
$(\sqrt{\mu_i})_{i<r}$ is proportional to
$(\sqrt{\nu_i})_{i<r}$, and the same is true for the two tail vectors.
Normalizing their squared norms to $1-t$ and $t$ gives
Eq.~\eqref{eq:reference-spectrum-optimizer}. Degenerate Schmidt values can still
permit rotations of the Schmidt bases, so no state-level uniqueness follows.
\end{proof}

\begin{corollary}[Fixed-fidelity value and Fenchel dual]
\label{cor:fixed-fidelity-value-dual}
For $F\in[0,1]$, define the primal value function
\begin{equation}\label{eq:fixed-fidelity-primal-value}
\mathcal B_{r,\phi}(F)
:=
\inf_{\substack{\rho\in\mathcal D(\mathcal H_n\otimes\mathcal H_m)\\
F_\phi(\rho)=F}}E_r(\rho).
\end{equation}
Then
\begin{equation}\label{eq:fixed-fidelity-exact-value}
\mathcal B_{r,\phi}(F)=\mathcal R_{r,\phi}(F),
\qquad 0\leq F\leq1,
\end{equation}
and the infimum is attained by a pure state for every $F$. Define the
one-observable Fenchel conjugate
\[
\widehat E_{r,\phi}(s)
:=
\sup_{\rho\in\mathcal D(\mathcal H_n\otimes\mathcal H_m)}
\bigl\{sF_\phi(\rho)-E_r(\rho)\bigr\}.
\]
With $A=A_r$, it is
\begin{equation}\label{eq:fixed-fidelity-fenchel-conjugate}
\widehat E_{r,\phi}(s)=
\begin{cases}
0, & s\leq0,\\[0.4em]
\displaystyle
\frac{s-1+\sqrt{(s-1)^2+4As}}{2}, & s>0.
\end{cases}
\end{equation}
It satisfies
\begin{equation}\label{eq:fixed-fidelity-biconjugate}
\mathcal R_{r,\phi}(F)
=\sup_{s\in\mathbb R}
\bigl\{sF-\widehat E_{r,\phi}(s)\bigr\}.
\end{equation}
\end{corollary}

\begin{proof}
Equation~\eqref{eq:reference-spectrum-mixed-bound} gives
$\mathcal B_{r,\phi}(F)\geq\mathcal R_{r,\phi}(F)$. The pure equality states from
Proposition~\ref{thm:single-observable-Er-main} are feasible in
Eq.~\eqref{eq:fixed-fidelity-primal-value}, proving the reverse inequality and
attainment.

Partition the supremum defining $\widehat E_{r,\phi}$ by fidelity and apply
Eq.~\eqref{eq:fixed-fidelity-exact-value}:
\[
\widehat E_{r,\phi}(s)
=\sup_{0\leq F\leq1}
\{sF-\mathcal R_{r,\phi}(F)\}.
\]
For $s\leq0$, the value is $0$ and is attained at $F=0$. When $A=1$,
$\mathcal R_{r,\phi}=0$, so the value for $s>0$ is $s$, in agreement with
Eq.~\eqref{eq:fixed-fidelity-fenchel-conjugate}. It remains to consider
$0<A<1$. Strict convexity of the positive branch gives a unique maximizer for
$s>0$. Solving
$\mathcal R_{r,\phi}'(F)=s$ yields
\[
F_s=\frac12\left[
1+\frac{s+2A-1}{\sqrt{(s-1)^2+4As}}
\right],
\]
and substitution gives Eq.~\eqref{eq:fixed-fidelity-fenchel-conjugate}.
Extend $\mathcal R_{r,\phi}$ to $\mathbb R$ by assigning $+\infty$ outside
$[0,1]$. The extension is proper, lower semicontinuous, and convex, so the
Fenchel--Moreau theorem gives Eq.~\eqref{eq:fixed-fidelity-biconjugate}. This
derivation does not assume a separate strong-duality statement.
\end{proof}

On the positive branch, Eq.~\eqref{eq:fixed-fidelity-exact-value} is implicit by
inversion of Proposition~3 in Ref.~\cite{WuLiZhu2026}; the exact-fidelity
zero-branch construction completes the statement on $[0,A_r]$. Thus the
all-state value is established; the additional content here is the
exact-fidelity zero-branch completion and the equality refinements stated above.
Classifying strictly mixed or full-rank equality states for nonmaximally
entangled references remains open. We call a mixed family saturating only when
its equality is verified.

\subsection{Multistep Schmidt-spectrum profiles}

A single $E_r$ resolves one Schmidt-rank threshold and falls within the
characteristic-curve baseline above. Several nested tails instead retain a
coarse-grained distribution of Schmidt weight. The results in this subsection
concern that multiscale information and are not obtained by inverting a
single-tail curve. One projector expectation does not supply these inputs.
Vidal monotones and majorization methods describe the complete ordered tail
hierarchy \cite{VidalMonotones,Bosyk2017,MarshallOlkinArnold2011}.
Theorem~\ref{thm:joint-multistep-main} gives an exact multistep boundary
that combines several selected tails with an if-and-only-if blockwise equality
condition and a partition-refinement hierarchy.

Let
\begin{equation*}
    \mathcal{P}
    =
    \{0=n_0<n_1<\cdots<n_L=m\}
\end{equation*}
be a partition of the ordered Schmidt spectrum into consecutive blocks
\[
    B_a=\{n_{a-1}+1,\ldots,n_a\},
    \qquad
    d_a=n_a-n_{a-1}.
\]
For each block, define the block weight
\begin{equation*}
    w_a=\sum_{i\in B_a}\mu_i,
    \qquad
    \sum_{a=1}^{L}w_a=1.
\end{equation*}
Ordered Schmidt probabilities impose
\[
\frac{w_a}{d_a}\geq\frac{w_{a+1}}{d_{a+1}},
\qquad a=1,\ldots,L-1.
\]
The blockwise-uniform spectrum $\mu_i=w_a/d_a$ realizes every such vector, so
\[
\mathcal W_{\mathcal P}^{\rm ord}
:=\left\{\boldsymbol w:\ w_a\geq0,\ \sum_aw_a=1,\
\frac{w_a}{d_a}\geq\frac{w_{a+1}}{d_{a+1}}\right\}.
\]
For each boundary \(n_j\), define the cumulative tail
\begin{equation*}
\begin{aligned}
    T_j
    &:=E_{n_j+1}(|\psi\rangle)
    =\sum_{i=n_j+1}^{m}\mu_i\\
    &=\sum_{a=j+1}^{L}w_a,\\
    &\hspace{-1.8em}T_0=1,\qquad T_L=0.
\end{aligned}
\end{equation*}
The differences recover the block weights:
\[
    w_a=T_{a-1}-T_a .
\]
The profile
\[
    \boldsymbol{T}_{\mathcal{P}}=(T_1,\ldots,T_{L-1})
\]
records the selected scales; distinct spectra may share it.
We call this profile admissible when the induced block masses
\(w_a=T_{a-1}-T_a\) belong to
\(\mathcal W_{\mathcal P}^{\rm ord}\). Equivalently, they are nonnegative,
sum to one, and have nonincreasing block means \(w_a/d_a\). These conditions
make the blockwise-uniform construction below a valid ordered Schmidt spectrum,
including profiles with zero block masses.

\begin{theorem}[Exact blockwise boundary for multistep profiles]
\label{thm:joint-multistep-main}
For every pure bipartite state and every partition \(\mathcal{P}\),
\begin{equation}
    \lambda
    \le
    \Lambda_{\mathcal{P}}
    (\boldsymbol{T}_{\mathcal{P}})
    :=
    \frac{1}{m}
    \left[
        \sum_{a=1}^{L}
        \sqrt{
            d_a
            \big(
                T_{a-1}-T_a
            \big)
        }
    \right]^2.
    \label{eq:partition-profile-bound-main}
\end{equation}
Conversely, every admissible profile is realized by the ordered spectrum
\(\mu_i=w_a/d_a\) for \(i\in B_a\), which attains the bound. For any state with
the prescribed profile, equality holds if and only if its Schmidt spectrum is
uniform within every block:
\begin{equation*}
    \mu_i=\frac{w_a}{d_a},
    \qquad
    i\in B_a.
\end{equation*}
\end{theorem}

\begin{proof}
For each block \(B_a\), the Cauchy--Schwarz inequality gives
\[
    \sum_{i\in B_a}\sqrt{\mu_i}
    \le
    \sqrt{
        d_a
        \sum_{i\in B_a}\mu_i
    }
    =
    \sqrt{d_aw_a}.
\]
Summing gives
\[
    \sum_{i=1}^{m}\sqrt{\mu_i}
    \le
    \sum_{a=1}^{L}\sqrt{d_aw_a}.
\]
Substitute \(w_a=T_{a-1}-T_a\) into \(\lambda\):
\[
    \lambda
    =
    \frac{1}{m}
    \left(
        \sum_{i=1}^{m}\sqrt{\mu_i}
    \right)^2
    \le
    \frac{1}{m}
    \left[
        \sum_{a=1}^{L}
        \sqrt{
            d_a(T_{a-1}-T_a)
        }
    \right]^2,
\]
This proves \eqref{eq:partition-profile-bound-main}. Equality in each
blockwise Cauchy--Schwarz inequality holds exactly when the entries in that
block are equal; for a zero-mass block they all vanish. Conversely,
admissibility gives \(w_a\geq0\), \(\sum_aw_a=1\), and nonincreasing means
\(w_a/d_a\). Hence \(\mu_i=w_a/d_a\) is an ordered Schmidt probability vector
with the prescribed tails and attains equality. This proves both sharpness for
every admissible profile and the stated equality characterization.
\end{proof}

\begin{proposition}[Partition-refinement principle]
\label{prop:partition-refinement}
Let $\mathcal P'$ refine $\mathcal P$, and evaluate both profiles on the same
ordered Schmidt spectrum. Then
\begin{equation}
    \lambda
    \le
    \Lambda_{\mathcal{P}'}
    \le
    \Lambda_{\mathcal{P}}.
    \label{eq:refinement-chain-main}
\end{equation}
Equality in the second inequality holds if and only if, within every block of
$\mathcal P$, all subblocks introduced by $\mathcal P'$ have the same mean
Schmidt probability. Hence any supplied refinement that separates subblocks
with unequal means gives a strict improvement. The singleton partition recovers
the full-spectrum value $\lambda$.
\end{proposition}

\begin{proof}
For one split, Cauchy--Schwarz gives
\[
    \sqrt{d_aw_a}+\sqrt{d_bw_b}
    \leq
    \sqrt{(d_a+d_b)(w_a+w_b)},
\]
with equality exactly when $w_a/d_a=w_b/d_b$. Repeated merging proves
Eq.~\eqref{eq:refinement-chain-main} and its equality condition. If all blocks
are singletons, the right-hand side of
Eq.~\eqref{eq:partition-profile-bound-main} equals the definition of $\lambda$.
\end{proof}

The proposition compares information sets: the refined bound is available only
when the additional tail data are supplied. It does not extract several tails
from the single reference fidelity used in the baseline proposition.

\subsubsection{Block-resolved observable input}

The projector in Proposition~\ref{thm:single-observable-Er-main} does not determine
the tails in Theorem~\ref{thm:joint-multistep-main}. For a certified pure source
with known ordered Schmidt modes $\{|a_i\rangle\}_{i=1}^m$, define at each cut
$n_j$ the nested local projectors
\[
P_j^A:=\sum_{i=1}^{n_j}|a_i\rangle\langle a_i|,
\qquad
Q_j:=(I_A-P_j^A)\otimes I_B.
\]
Then
\[
\langle\psi|Q_j|\psi\rangle
=\sum_{i=n_j+1}^m\mu_i=T_j.
\]
One mode-resolved measurement in the certified Schmidt basis evaluates these
nested commuting projectors. This idealized count excludes basis certification,
purity tests, and mode-dependent errors.

Let $\mathcal C_{1-\epsilon}$ be a simultaneous confidence set
for the tail vector obtained from block-resolved counts, and let
$\mathcal T_{\mathcal P}^{\rm ord}$ denote the physically admissible tail
vectors corresponding to $\mathcal W_{\mathcal P}^{\rm ord}$. On its coverage
event, Theorem~\ref{thm:joint-multistep-main} gives
\[
\lambda\leq
\max_{\boldsymbol T\in
\mathcal C_{1-\epsilon}\cap\mathcal T_{\mathcal P}^{\rm ord}}
\Lambda_{\mathcal P}(\boldsymbol T),
\]
and
\[
\min_{\boldsymbol T\in
\mathcal C_{1-\epsilon}\cap\mathcal T_{\mathcal P}^{\rm ord}}
\sum_j\omega_jT_j
\leq\mathcal E_{\boldsymbol\omega}(|\psi\rangle)
\leq
\max_{\boldsymbol T\in
\mathcal C_{1-\epsilon}\cap\mathcal T_{\mathcal P}^{\rm ord}}
\sum_j\omega_jT_j.
\]
Refinement requires finer outcomes or an independently certified source.

For a mixed state, $\operatorname{Tr}(Q_j\rho)$ gives fixed-subspace populations,
not the convex roofs $E_{n_j+1}(\rho)$ or the component tails of an optimal
decomposition. Their identification requires a pure-source, locally flagged,
or certified decomposition model.
Selected Vidal tails constrain stochastic pure-state conversion
\cite{VidalSingleCopy}; nonnegative weights scalarize these constraints for a
specified task, not a universal rate.

For \(w_a>0\), define the effective block dimension
\begin{equation*}
    d_a^{\mathrm{eff}}
    :=
    \frac{
        \left(
            \sum_{i\in B_a}\sqrt{\mu_i}
        \right)^2
    }{
        w_a
    },
    \qquad
    1\le d_a^{\mathrm{eff}}\le d_a.
\end{equation*}
For $w_a=0$, set $d_a^{\rm eff}=1$; the block contributes zero. Then
\begin{equation*}
    \lambda
    =
    \frac{1}{m}
    \left(
        \sum_{a=1}^{L}
        \sqrt{
            d_a^{\mathrm{eff}}w_a
        }
    \right)^2.
\end{equation*}
Equivalently, define the intrablock defect
\begin{equation*}
    \Delta_a
    :=
    d_aw_a
    -
    \left(
        \sum_{i\in B_a}\sqrt{\mu_i}
    \right)^2
    =
    \sum_{\substack{i,j\in B_a\\i<j}}
    \left(
        \sqrt{\mu_i}
        -
        \sqrt{\mu_j}
    \right)^2
    \ge0.
\end{equation*}
Then
\begin{equation*}
    d_a^{\mathrm{eff}}
    =
    d_a-\frac{\Delta_a}{w_a},
\end{equation*}
for \(w_a>0\), and
\begin{equation*}
    \lambda
    =
    \frac{1}{m}
    \left[
        \sum_{a=1}^{L}
        \sqrt{
            d_aw_a-\Delta_a
        }
    \right]^2.
\end{equation*}
The gap \(\Lambda_{\mathcal P}-\lambda\) depends jointly on
\(\{\Delta_a\}\); without lower block weights, near saturation need not imply
blockwise uniformity.

\subsection{Weighted multiscale profiles}

Nonnegative weights scalarize the multistep tail vector. The relaxation below
is a linear optimization over a Bhattacharyya--Hellinger superlevel set. General
optimization under Hellinger or other $\phi$-divergence constraints is
established \cite{BenTalEtAl2013}.
Theorem~\ref{thm:global-weighted-multiscale} gives a global analytic solution
for the resulting relaxed weighted Vidal-tail problem: it provides an explicit
global optimizer, proves uniqueness on the full-support branch, and reduces the
solution to one scalar equation.
Let
\[
    \boldsymbol{\omega}=(\omega_1,\ldots,\omega_{L-1}),
    \qquad
    \omega_j\ge0,
\]
and define the pure-state weighted profile
\begin{equation*}
    \mathcal{E}_{\boldsymbol{\omega}}(|\psi\rangle)
    =
    \sum_{j=1}^{L-1}
    \omega_j E_{n_j+1}(|\psi\rangle).
\end{equation*}
In block weights,
\begin{equation*}
    \mathcal{E}_{\boldsymbol{\omega}}(|\psi\rangle)
    =
    \sum_{a=1}^{L}c_aw_a,
    \qquad
    c_1=0,\quad
    c_a=\sum_{j=1}^{a-1}\omega_j.
\end{equation*}
Since \(\omega_j\ge0\), the coefficients satisfy
\[
    0=c_1\le c_2\le\cdots\le c_L.
\]
For prescribed effective-dimension upper bounds
\[
    1\le\widetilde d_a\le d_a,
    \qquad
    a=1,\ldots,L,
\]
define the relaxed block-mass optimization
\begin{widetext}
 \begin{equation}
    \mathcal G^{\rm rel}_{\boldsymbol{\omega},\widetilde{\boldsymbol d}}(\lambda)
    :=
    \min_{\boldsymbol w}
    \left\{
        \sum_{a=1}^{L}c_aw_a
        \ \middle|\
        w_a\geq0,\quad
        \sum_{a=1}^{L}w_a=1,\quad
        H_{\widetilde{\boldsymbol d}}(\boldsymbol w)\geq\lambda
    \right\},
    \label{eq:weighted-envelope-main}
\end{equation}   
\end{widetext}
where
\[
    H_{\widetilde{\boldsymbol d}}(\boldsymbol w)
    :=
    \frac{1}{m}
    \left(
        \sum_{a=1}^{L}
        \sqrt{\widetilde d_aw_a}
    \right)^2.
\]
The superscript ``rel'' denotes the full probability simplex, without the
physical constraints
\(w_a/d_a\geq w_{a+1}/d_{a+1}\). For the actual block dimensions, define
\[
\mathcal G^{\rm ord}_{\boldsymbol\omega,\boldsymbol d}(\lambda)
:=\min\left\{\sum_ac_aw_a:\
\boldsymbol w\in\mathcal W_{\mathcal P}^{\rm ord},\
H_{\boldsymbol d}(\boldsymbol w)\geq\lambda\right\}.
\]
Since \(\mathcal W_{\mathcal P}^{\rm ord}\) is a subset of the simplex,
\[
\mathcal G^{\rm ord}_{\boldsymbol\omega,\boldsymbol d}(\lambda)
\geq
\mathcal G^{\rm rel}_{\boldsymbol\omega,\boldsymbol d}(\lambda).
\]
Define the zero-cost indices and effective dimensions
\[
    \mathcal I_0
    :=
    \{a:c_a=0\},
    \qquad
    D_0
    :=
    \sum_{a\in\mathcal I_0}\widetilde d_a,
    \qquad
    D
    :=
    \sum_{a=1}^{L}\widetilde d_a.
\]
For $\operatorname{supp}(\boldsymbol w)\subseteq\mathcal I_0$,
Cauchy--Schwarz gives
\[
    H_{\widetilde{\boldsymbol d}}(\boldsymbol w)
    \leq
    \frac{D_0}{m}.
\]
Equality holds at
\[
    w_a
    =
    \frac{\widetilde d_a}{D_0},
    \quad a\in\mathcal I_0,
    \qquad
    w_a=0,
    \quad a\notin\mathcal I_0.
\]
Hence the relaxed envelope has the zero-cost branch
\[
    \mathcal G^{\rm rel}_{\boldsymbol{\omega},\widetilde{\boldsymbol d}}(\lambda)=0
    \qquad
    \text{for }
    0\leq\lambda\leq\frac{D_0}{m}.
\]
Values below \(1/m\) extend the envelope analytically. Since pure states satisfy
\(\lambda\geq1/m\), their zero-cost branch is \([1/m,D_0/m]\).

The maximum of $H_{\widetilde{\boldsymbol d}}$ over the simplex is
$D/m$, attained uniquely at
\[
    w_a=\frac{\widetilde d_a}{D}.
\]
Feasibility requires $\lambda\leq D/m$. The interval
\[
    \frac{D_0}{m}<\lambda<\frac{D}{m}
\]
forms the nontrivial branch.
\begin{theorem}[Global solution of the relaxed weighted optimization]
\label{thm:global-weighted-multiscale}
Assume that \(\widetilde d_a>0\) for every \(a\), that
\(0=c_1\leq\cdots\leq c_L\), and that the costs are not all equal. Then
\(\mathcal I_0=\{a:c_a=0\}\) is nonempty. For every
\(
    \frac{D_0}{m}<\lambda<\frac{D}{m},
\)
there exists a unique parameter $\alpha>0$ satisfying
\[
    \lambda
    =
    \frac{Z_1(\alpha)^2}
    {mZ_2(\alpha)},
\]
where
\[
    Z_1(\alpha)
    :=
    \sum_{a=1}^{L}
    \frac{\widetilde d_a}{c_a+\alpha},
    \qquad
    Z_2(\alpha)
    :=
    \sum_{a=1}^{L}
    \frac{\widetilde d_a}{(c_a+\alpha)^2}.
\]
The optimization defining
$\mathcal G^{\rm rel}_{\boldsymbol{\omega},\widetilde{\boldsymbol d}}(\lambda)$
has a unique global minimizer:
\[
    w_a^\star(\alpha)
    =
    \frac{
        \displaystyle
        \frac{\widetilde d_a}{(c_a+\alpha)^2}
    }{
        \displaystyle
        \sum_{b=1}^{L}
        \frac{\widetilde d_b}{(c_b+\alpha)^2}
    },
    \qquad
    a=1,\ldots,L.
\]
All components are positive:
\[
    w_a^\star(\alpha)>0
    \qquad
    \text{for every }a,
\]
Thus the minimizer has full support on the nontrivial branch, with value
\[
    \mathcal G^{\rm rel}_{\boldsymbol{\omega},\widetilde{\boldsymbol d}}(\lambda)
    =
    \frac{
        \displaystyle
        \sum_{a=1}^{L}
        \frac{c_a\widetilde d_a}{(c_a+\alpha)^2}
    }{
        Z_2(\alpha)
    }
    =
    \frac{Z_1(\alpha)}{Z_2(\alpha)}-\alpha.
\]
\end{theorem}

The qualifier ``relaxed'' is essential. The theorem solves the block-mass
problem in Eq.~\eqref{eq:weighted-envelope-main}; it does not evaluate the exact
fixed-fidelity convex roof of a general weighted profile. Appendix B proves
global optimality with a Cauchy--Schwarz certificate in addition to the
stationarity equations, so the displayed point is not merely a KKT candidate.

The uniqueness assertion is confined to
\(D_0/m<\lambda<D/m\). The zero-cost branch
\(0\leq\lambda\leq D_0/m\) and the endpoint \(\lambda=D/m\) lie outside the
scalar-root statement and are characterized separately. In particular, no
zero-cost or simplex-boundary minimizer is included in the theorem's
full-support claim.

If all costs vanish, every feasible set has zero objective; the theorem excludes
this case.
For the physical choice \(\widetilde{\boldsymbol d}=\boldsymbol d\), the relaxed
minimizer satisfies
\[
\frac{w_a^\star}{d_a}
\propto\frac{1}{(c_a+\alpha)^2},
\]
which is nonincreasing in \(a\). Thus the minimizer lies in
\(\mathcal W_{\mathcal P}^{\rm ord}\), and the relaxed and ordered values
coincide. This need not hold for arbitrary \(\widetilde{\boldsymbol d}\).

Appendix B proves global optimality and uniqueness.

At the upper endpoint $\lambda=D/m$, the feasible set reduces to
the unique maximizer of $H_{\widetilde{\boldsymbol d}}$,
\[
    w_a^\star
    =
    \frac{\widetilde d_a}{D},
\]
and hence
\[
    \mathcal G^{\rm rel}_{\boldsymbol{\omega},\widetilde{\boldsymbol d}}\left(\frac{D}{m}\right)
    =
    \frac{1}{D}
    \sum_{a=1}^{L}c_a\widetilde d_a.
\]
Together with the zero branch, this endpoint completes
$\mathcal G^{\rm rel}_{\boldsymbol{\omega},\widetilde{\boldsymbol d}}$.

The choice
\[
    L=2,
    \widetilde d_1=r-1,
    \widetilde d_2=m-r+1,
    c_1=0,
    c_2=1.
\]
gives $D_0=r-1$ and $D=m$. Its zero-cost branch is
$\mathcal G^{\rm rel}(\lambda)=0$ for
$1/m\leq\lambda\leq(r-1)/m$. On the nontrivial branch
$(r-1)/m<\lambda\leq1$, the solution reduces to
\[
    \mathcal{G}^{\rm rel}(\lambda)
    =
    \frac{1}{m}
    \left[
        \sqrt{(m-r+1)\lambda}
        -
        \sqrt{(r-1)(1-\lambda)}
    \right]^2,
\]
which is the positive branch of the two-block boundary.

\begin{proposition}[Convexity of the relaxed weighted envelope]
\label{prop:multistep-convex-main}
The function
\(
\mathcal{G}^{\rm rel}_{\boldsymbol{\omega},
\widetilde{\boldsymbol d}}(\lambda)
\)
is nondecreasing and convex on its domain.
\end{proposition}

\begin{proof}
Increasing \(\lambda\) shrinks the feasible set, proving monotonicity.
Compactness gives a minimizer \(\boldsymbol w^{(i)}\) at each \(\lambda_i\).
For
\(0\leq\theta\leq1\), the concavity of
\(H_{\widetilde{\boldsymbol d}}\) gives
\[
    H_{\widetilde{\boldsymbol d}}
    \big(
        \theta\boldsymbol{w}^{(1)}
        +(1-\theta)\boldsymbol{w}^{(2)}
    \big)
    \ge
    \theta\lambda_1+(1-\theta)\lambda_2.
\]
This convex combination is feasible at
\(\theta\lambda_1+(1-\theta)\lambda_2\); linearity of the objective gives
\[
\begin{aligned}
&\mathcal G^{\rm rel}_{\boldsymbol{\omega},\widetilde{\boldsymbol d}}
\bigl(\theta\lambda_1+(1-\theta)\lambda_2\bigr)\\
&\qquad\leq
\theta\mathcal G^{\rm rel}_{\boldsymbol{\omega},\widetilde{\boldsymbol d}}(\lambda_1)
+(1-\theta)\mathcal G^{\rm rel}_{\boldsymbol{\omega},\widetilde{\boldsymbol d}}(\lambda_2).
\end{aligned}
\]
\end{proof}

For a pure state, replacing \(d_a\) by the corresponding
\(d_a^{\rm eff}\leq d_a\) reduces the relaxed feasible set and yields
\begin{equation*}
    \mathcal{G}^{\rm rel}_{\boldsymbol{\omega},
    \boldsymbol d^{\mathrm{eff}}}(\lambda)
    \ge
    \mathcal{G}^{\rm rel}_{\boldsymbol{\omega},
    \boldsymbol d}(\lambda).
\end{equation*}
A mixed-state convex roof needs extra data for this state-dependent substitution.

For mixed states, define the convex-roof extension
\begin{equation*}
    \mathcal{E}_{\boldsymbol{\omega}}(\rho)
    =
    \inf_{\{p_j,|\psi_j\rangle\}\in\mathfrak E(\rho)}
    \sum_jp_j
    \mathcal{E}_{\boldsymbol{\omega}}
    (|\psi_j\rangle).
\end{equation*}
This convex roof may differ from the weighted sum of the separate roofs.

\begin{proposition}[Full-spectrum weighted baseline]
\label{prop:full-spectrum-weighted}
For every density operator and every nonnegative weight vector,
\begin{equation}
\mathcal E_{\boldsymbol\omega}(\rho)
\geq\sum_{j=1}^{L-1}\omega_jE_{n_j+1}(\rho)
\geq B_{\rm FS}\bigl(F_\phi(\rho)\bigr),
\label{eq:full-spectrum-weighted-chain}
\end{equation}
where
\begin{equation}
B_{\rm FS}(F):=
\sum_{j=1}^{L-1}\omega_j
\mathcal R_{n_j+1,\phi}(F).
\label{eq:full-spectrum-weighted-baseline}
\end{equation}
\end{proposition}

\begin{proof}
In any common pure-state decomposition of $\rho$, each average tail bounds its
separate convex roof $E_{n_j+1}(\rho)$. Multiply by $\omega_j\geq0$, sum, and
minimize to obtain the first inequality. Applying
Eq.~\eqref{eq:reference-spectrum-mixed-bound} gives the second.
\end{proof}

The first inequality is tight only if one decomposition minimizes every
positively weighted tail; the second requires simultaneous saturation of all
active fixed-fidelity bounds. One nonzero weight satisfies both conditions, but
a multiscale profile need not. We do not classify all equality cases.

For comparison, define
\begin{align}
\Lambda_\phi^{(1)}(F)
&:=\max\left\{\frac{F}{m\nu_1},\frac1m\right\},\nonumber\\
B_{\nu_1}(F)
&:=\mathcal G^{\rm rel}_{\boldsymbol\omega,\boldsymbol d}
\bigl(\Lambda_\phi^{(1)}(F)\bigr),
\label{eq:weighted-nu1-baseline}
\end{align}
and
\begin{equation}
B_{\rm best}(F):=\max\{B_{\rm FS}(F),B_{\nu_1}(F)\}.
\label{eq:weighted-best-baseline}
\end{equation}

\begin{proposition}[Auxiliary largest-coefficient bound for weighted profiles]
\label{thm:single-observable-multiscale-main}
For every nonnegative weight vector \(\boldsymbol{\omega}\) and every
$\rho\in\mathcal D(\mathcal H_n\otimes\mathcal H_m)$,
\begin{equation}
    \mathcal{E}_{\boldsymbol{\omega}}(\rho)
    \ge
    B_{\nu_1}\bigl(F_\phi(\rho)\bigr).
\label{eq:weighted-nu1-bound}
\end{equation}
More generally, $\boldsymbol d$ may be replaced by
$\bar{\boldsymbol d}$ with $1\leq\bar d_a\leq d_a$ if, for every
$\varepsilon>0$, there exists an $\varepsilon$-optimal decomposition for
$\mathcal E_{\boldsymbol\omega}(\rho)$ in which every nonzero pure component
satisfies $d^{\rm eff}_{a,j}\leq\bar d_a$ for all $a$.
\end{proposition}

\begin{proof}
Let $\boldsymbol d'$ denote $\boldsymbol d$ in the universal case and
$\bar{\boldsymbol d}$ in the conditional case. Choose an admissible
$\varepsilon$-optimal decomposition; every decomposition is admissible in the
universal case. Componentwise,
\(H_{\boldsymbol d'}(\boldsymbol w_j)\geq
H_{\boldsymbol d_j^{\rm eff}}(\boldsymbol w_j)=\lambda_j\).
Writing \(D':=\sum_a d'_a\), feasibility implies
\(\lambda_j\leq D'/m\). The largest-coefficient trace inequality gives
\[
F_\phi(\rho)=\sum_jp_jF_\phi(|\psi_j\rangle)
\leq m\nu_1\sum_jp_j\lambda_j.
\]
Since every $\lambda_j\geq1/m$,
\[
\Lambda_\phi^{(1)}(\rho)
\leq\sum_jp_j\lambda_j
\leq\frac{D'}{m},
\]
so every argument of
\(\mathcal G^{\rm rel}_{\boldsymbol\omega,\boldsymbol d'}\) lies in its domain.
Each component block vector is feasible and satisfies
\[
    \mathcal{E}_{\boldsymbol{\omega}}(|\psi_j\rangle)
    \ge
    \mathcal{G}^{\rm rel}_{\boldsymbol{\omega},\boldsymbol d'}(\lambda_j).
\]
Convexity from Proposition~\ref{prop:multistep-convex-main} gives
\[
    \sum_jp_j
    \mathcal{E}_{\boldsymbol{\omega}}(|\psi_j\rangle)
    \ge
    \mathcal{G}^{\rm rel}_{\boldsymbol{\omega},\boldsymbol d'}
    \left(
        \sum_jp_j\lambda_j
    \right).
\]
Because
\[
    \sum_jp_j\lambda_j\ge\Lambda_\phi^{(1)}(\rho)
\]
monotonicity then gives
\[
    \sum_jp_j
    \mathcal{E}_{\boldsymbol{\omega}}(|\psi_j\rangle)
    \ge
    \mathcal{G}^{\rm rel}_{\boldsymbol{\omega},\boldsymbol d'}
    \big(
        \Lambda_\phi^{(1)}(\rho)
    \big).
\]
The selected decomposition bounds the left side by
$\mathcal E_{\boldsymbol\omega}(\rho)+\varepsilon$. Let
$\varepsilon\downarrow0$.
\end{proof}

One fidelity does not certify $\bar d_a<d_a$. Any improvement of
$\boldsymbol d$ requires source information satisfying the quantified
$\varepsilon$-optimal-decomposition condition, not one chosen decomposition.

Propositions~\ref{prop:full-spectrum-weighted} and
\ref{thm:single-observable-multiscale-main} together give
\begin{equation}
\mathcal E_{\boldsymbol\omega}(\rho)
\geq B_{\rm best}\bigl(F_\phi(\rho)\bigr).
\label{eq:weighted-best-bound}
\end{equation}
If $j_*$ is the first index with $\omega_{j_*}>0$, then
\[
\begin{aligned}
B_{\rm FS}(F)>0&\iff F>A_{n_{j_*}+1},\\
B_{\nu_1}(F)>0&\iff F>n_{j_*}\nu_1.
\end{aligned}
\]
Because $A_{n_{j_*}+1}\leq n_{j_*}\nu_1$, the full-spectrum baseline activates
no later and is strictly stronger on
$A_{n_{j_*}+1}<F\leq n_{j_*}\nu_1$ whenever the thresholds differ. Once both
bounds are active, neither has a universal advantage. The quantity $B_{\rm FS}$
retains every cumulative reference weight, whereas $B_{\nu_1}$ compresses the
reference to $\nu_1$ before relaxing the weighted profile. We report
$B_{\rm FS}$, $B_{\nu_1}$, and $B_{\rm best}$ below.

Supplied tails constrain coarse-grained spectra and partition refinements. They
are additional inputs, not consequences of one reference expectation.
\emph{The exact fixed-fidelity value of a general weighted profile remains
open.}

\subsection{Operational consequences and incomplete-information inference}

Reference~\cite{WuLiZhu2026} explicitly notes that its characteristic value is
independent of additional finite local dimensions once the nonzero reference
spectrum is fixed. The following corollary packages that observation in the
fixed-fidelity, zero-padding notation and includes the completed zero branch.

\begin{corollary}[Fixed-fidelity form of finite-dimensional independence]
\label{prop:ambient-extension}
Let $|\phi\rangle$ have finite Schmidt rank $s$ and be embedded by local
isometries into a finite-dimensional space
$\mathcal K_A\otimes\mathcal K_B$. Set
$D=\min\{\dim\mathcal K_A,\dim\mathcal K_B\}$ and pad the reference Schmidt
spectrum by zeros to length $D$. For every $2\leq r\leq D$ and
$F\in[0,1]$,
\begin{equation}
\inf_{\substack{\rho\in\mathcal D(\mathcal K_A\otimes\mathcal K_B)\\
F_\phi(\rho)=F}}E_r(\rho)
=\mathcal R_{r,\phi}(F).
\label{eq:ambient-fixed-fidelity-value}
\end{equation}
For $r\leq s$, zero padding leaves $A_r$ and $B_r$ unchanged. For $r>s$,
$A_r=1$ and both sides of Eq.~\eqref{eq:ambient-fixed-fidelity-value} vanish.
No infinite-dimensional extension is asserted.
\end{corollary}

\begin{proof}
Zero padding preserves the fidelity identity and Bures-angle bound in the
larger finite-dimensional space. If $r\leq s$, the positive-branch state
in Eq.~\eqref{eq:positive-geodesic-state} and the zero-branch unitary path both
lie in the original reference support, so the lower bound remains attainable.
If $r>s$, the reference has Schmidt rank at most $r-1$. A fixed-point-free local
cyclic shift maps it to a state orthogonal to it, and a continuous unitary path
between the identity and that shift preserves its Schmidt rank. The squared
overlap therefore covers $[0,1]$, proving zero-valued attainability for every
$F$.
\end{proof}

\begin{corollary}[Leakage-aware local-filtering certificate]
\label{cor:leakage-filtering}
Let $\Pi_m=\Pi_A\otimes\Pi_B$ be a product projector satisfying
$P_\phi=\Pi_mP_\phi\Pi_m$, and define
\[
p_{\rm in}=\operatorname{Tr}(\Pi_m\rho),
\qquad
\rho_m=\frac{\Pi_m\rho\Pi_m}{p_{\rm in}},
\qquad
F_{\rm tot}=\operatorname{Tr}(P_\phi\rho).
\]
For $p_{\rm in}>0$,
\begin{equation}
E_r(\rho)
\geq p_{\rm in}E_r(\rho_m)
\geq p_{\rm in}\mathcal R_{r,\phi}
\left(\frac{F_{\rm tot}}{p_{\rm in}}\right)
\geq\mathcal R_{r,\phi}(F_{\rm tot}).
\label{eq:leakage-filtering-bound}
\end{equation}
\end{corollary}

\begin{proof}
The two local binary projections form a four-outcome LOCC instrument. Strong
monotonicity bounds the average Vidal tail. Dropping the three nonnegative terms
outside the target subspace leaves
$E_r(\rho)\geq p_{\rm in}E_r(\rho_m)$
\cite{VidalMonotones}. Because $P_\phi$ is supported in $\Pi_m$,
$F_\phi(\rho_m)=F_{\rm tot}/p_{\rm in}$, and the second inequality follows from
Proposition~\ref{thm:single-observable-Er-main}. Convexity and
$\mathcal R_{r,\phi}(0)=0$ imply
$\mathcal R_{r,\phi}(px)\leq p\mathcal R_{r,\phi}(x)$ for $0\leq p\leq1$.
\end{proof}

The perspective function
\[
G_{r,\phi}(p,F):=p\mathcal R_{r,\phi}(F/p),
\qquad 0\leq F\leq p\leq1,
\]
is nondecreasing in $F$ and nonincreasing in $p$. A joint coverage event with
$F_{\rm tot}\geq F_L$, $p_{\rm in}\leq p_U$, and
$A_r\leq\overline A_r$ yields
\begin{equation}
E_r(\rho)\geq
p_U\mathcal R_{\overline A_r}\!\left(\frac{F_L}{p_U}\right),
\qquad 0\leq F_L\leq p_U\leq1.
\label{eq:leakage-perspective-confidence}
\end{equation}
For a nonrectangular joint region, minimize $p\mathcal R_{A}(F/p)$ directly.
Combine marginal extrema only under simultaneous coverage.

\begin{corollary}[Self-testing-to-tail conversion]
\label{cor:self-testing-tail-conversion}
Let
$\rho_{\rm phys}\in\mathcal D(\mathcal K_A\otimes\mathcal K_B)$
be a finite-dimensional physical state, and let
\[
\Gamma_A:\mathcal L(\mathcal K_A)\to\mathcal L(\mathcal H_n),
\qquad
\Gamma_B:\mathcal L(\mathcal K_B)\to\mathcal L(\mathcal H_m)
\]
be local completely positive trace-preserving extraction maps. Set
\[
\sigma_{\rm ext}:=(\Gamma_A\otimes\Gamma_B)(\rho_{\rm phys}).
\]
Fix a pure reference $|\phi\rangle\in\mathcal H_n\otimes\mathcal H_m$ and
$2\leq r\leq m$. If a self-testing analysis certifies
\[
F_\phi(\sigma_{\rm ext})\geq F_L,
\]
then
\begin{equation}
E_r(\rho_{\rm phys})
\geq E_r(\sigma_{\rm ext})
\geq \mathcal R_{r,\phi}(F_L).
\label{eq:self-testing-tail-conversion}
\end{equation}
The condition $F_L>A_r$ certifies
$\operatorname{SN}(\rho_{\rm phys})\geq r$.
If the fidelity statement holds on an event of probability at least
$1-\epsilon$, the same qualification applies to Eq.~\eqref{eq:self-testing-tail-conversion}.
\end{corollary}

\begin{proof}
Convex-roof Vidal tails cannot increase under deterministic local operations
\cite{VidalMonotones}, which gives
$E_r(\rho_{\rm phys})\geq E_r(\sigma_{\rm ext})$.
Proposition~\ref{thm:single-observable-Er-main} and monotonicity of
$\mathcal R_{r,\phi}$ give the second inequality. If $F_L>A_r$, its right-hand
side is positive, so $\operatorname{SN}(\sigma_{\rm ext})\geq r$. Local channels
cannot increase Schmidt number, which proves the final claim.
\end{proof}

In swap-based self-testing, the extraction maps include local isometries and
discarded auxiliary systems. Corollary
\ref{cor:self-testing-tail-conversion} only post-processes a valid fidelity
lower bound; it does not change the sampling, fair-sampling, or device
assumptions.

\begin{corollary}[Confidence-qualified projector calibration]
\label{cor:confidence-qualified-projector-calibration}
Suppose that, on a joint coverage event of probability at least $1-\epsilon$,
the reference expectation and cumulative reference weight satisfy
\[
F_\phi(\rho)\geq F_L,
\qquad
A_r\leq\overline A_r.
\]
Then, on the same event,
\[
E_r(\rho)\geq\mathcal R_{\overline A_r}(F_L).
\]
If, in addition, an implemented pure projector $\widetilde\tau$ obeys
$\|\widetilde\tau-|\phi\rangle\langle\phi|\|_\infty\leq\varepsilon_\phi$
and the measured expectation has a one-sided lower confidence bound
$\operatorname{Tr}(\widetilde\tau\rho)\geq\widetilde F_L$, one may take
$F_L=\max\{0,\widetilde F_L-\varepsilon_\phi\}$.
\end{corollary}

\begin{proof}
The function $\mathcal R_A(F)$ is nondecreasing in $F$. On its positive branch,
the bracket in Eq.~\eqref{eq:reference-spectrum-boundary} is positive and
strictly decreases with $A$. The zero branch preserves this ordering, so
$\mathcal R_A(F)$ is nonincreasing in $A$. Applying
Proposition~\ref{thm:single-observable-Er-main} gives
\[
E_r(\rho)\geq\mathcal R_{A_r}(F_\phi(\rho))
\geq\mathcal R_{\overline A_r}(F_L).
\]
For the implemented projector, H\"older's inequality gives
$|\operatorname{Tr}[(\widetilde\tau-|\phi\rangle\langle\phi|)\rho]|
\leq\varepsilon_\phi$. No independence assumption is needed, provided the
reported event has the stated joint coverage.
\end{proof}

Corollary~\ref{cor:confidence-qualified-projector-calibration} is an inference
rule, not a prescription for measuring $P_\phi$. Any statistic $Z$ is admissible
if a proved analysis supplies a joint event of the
form $F_\phi(\rho)\geq f(Z)$ and $A_r\leq\overline A_r$. Substituting
$F_L=f(Z)$ preserves its coverage and assumptions. This includes bounds from
several local correlations but does not turn them into one local setting.

Near the threshold, if $0<A_r<1$ and
$F=A_r+\delta$ with $\delta\downarrow 0$, then
\[
\mathcal R_{r,\phi}(A_r+\delta)
=\frac{\delta^2}{4A_rB_r}+O(\delta^3).
\]
A certificate requires $F_L>\overline A_r$; sampling and calibration errors may
erase a small nominal violation.

The thresholds $\{A_r\}_{r=2}^{m}$ determine the reference spectrum:
$\nu_1=A_2$, $\nu_i=A_{i+1}-A_i$ for $2\leq i<m$, and
$\nu_m=1-A_m$. The same global projector expectation can consequently be
post-processed at every threshold. One bound uses one $A_r$; the hierarchy uses
the full spectrum. This does not make the projector a single local setting.

\subsubsection{Measurement resources, finite samples, and dimension leakage}

The protocol takes one expectation value as input; its cost depends on the local
decomposition and measurement architecture. In the reference Schmidt basis,
define
\[
\begin{aligned}
D_i&:=|i\rangle\langle i|,\\
X_{ij}&:=\frac{|i\rangle\langle j|+|j\rangle\langle i|}{\sqrt2},
&
Y_{ij}&:=\frac{|i\rangle\langle j|-|j\rangle\langle i|}{i\sqrt2}.
\end{aligned}
\]
The projector $P_\phi=|\phi\rangle\langle\phi|$ then has the local Hermitian
decomposition
\begin{equation}\label{eq:projector-local-decomposition}
P_\phi
=\sum_{i=1}^{m}\nu_iD_i\otimes D_i
+\sum_{1\leq i<j\leq m}\sqrt{\nu_i\nu_j}
\bigl(X_{ij}\otimes X_{ij}-Y_{ij}\otimes Y_{ij}\bigr),
\end{equation}
up to local Schmidt-basis isometries. It contains $m^2$ product correlators.
Separate measurements require at most $m^2$ product settings; grouping all
diagonal terms gives $m^2-m+1$. Commuting groups, collective measurements, or a
direct global projection change the platform-specific cost. An unrestricted
$m\times m$ state has order $m^4$ real tomographic parameters, but this count
alone proves no shot or setting advantage.

For a direct two-outcome measurement $\{P_\phi,I-P_\phi\}$, let
\[
\widehat F=\frac1N\sum_{i=1}^{N}X_i,
\qquad X_i\in\{0,1\}.
\]
Hoeffding's inequality \cite{Hoeffding1963} gives the one-sided confidence bound
\begin{equation}\label{eq:hoeffding-projector-lower}
F_\phi(\rho)\geq
F_L:=\max\left\{0,
\widehat F-\sqrt{\frac{\log(1/\epsilon)}{2N}}
\right\}
\end{equation}
with probability at least $1-\epsilon$. If instead
$P_\phi=\sum_{\ell}a_\ell A_\ell\otimes B_\ell$ is estimated from independent
local samples with normalized outcomes in $[-1,1]$ and $N_\ell$ shots per term,
a conservative Hoeffding radius is
\begin{equation}\label{eq:local-decomposition-hoeffding}
\Delta_\epsilon
=\sqrt{2\log(1/\epsilon)
\sum_\ell\frac{a_\ell^2}{N_\ell}}.
\end{equation}
Then $F_L=\max\{0,\widehat F-\Delta_\epsilon\}$. This radius exposes the shot
allocation; covariance and commuting-group estimators require separate
concentration analyses. Data-derived spectra, calibration errors, or multiple
thresholds require the joint coverage event in
Corollary~\ref{cor:confidence-qualified-projector-calibration}.

For a fixed total budget $N_{\rm tot}=\sum_\ell N_\ell$, minimizing the
independent-term radius in Eq.~\eqref{eq:local-decomposition-hoeffding} gives
$N_\ell\propto|a_\ell|$ (up to integer rounding), and hence
\[
\Delta_\epsilon^{\rm alloc}
=\sqrt{2\log(1/\epsilon)}
\frac{\sum_\ell|a_\ell|}{\sqrt{N_{\rm tot}}}.
\]
If pilot estimates of the setting variances $\sigma_\ell^2$ are available, the
variance objective instead gives $N_\ell\propto|a_\ell|\sigma_\ell$. These rules
allocate shots, not basis changes. In the elementary
grouping of Eq.~\eqref{eq:projector-local-decomposition}, all diagonal terms can
be estimated in one Schmidt-basis setting: for a joint outcome $(a,b)$ use
$Z_D=\nu_a\mathbf 1\{a=b\}$, whose range is $[0,\nu_1]$ and whose expectation is
$\sum_i\nu_i\Pr(i,i)$. Each off-diagonal setting has its own bounded outcome;
grouped settings must use their range and covariance rather than the
independent-term radius. A resource report should list correlators, groups,
shots, allocation, outcome ranges, confidence level, and switching overhead.

\paragraph{Forward design criterion.}
Suppose an experimental model supplies a nominal fidelity
$F_{\rm nom}(\phi)$ for each implementable reference
$\phi\in\Phi_{\rm impl}$ and a valid error radius
$\Delta_\phi(\{N_\ell\},\epsilon)$. For a fixed shot budget, reference and tail
selection can be posed before data collection as
\begin{equation}
\begin{aligned}
\max_{\phi,r,\{N_\ell\}}\;&
\Bigl[\arccos\sqrt{A_r(\phi)}\\[-1mm]
&\quad-\arccos\sqrt{[F_{\rm nom}(\phi)-\Delta_\phi]_+}\Bigr]_+ \\
\text{subject to}\;&
\phi\in\Phi_{\rm impl},\quad 2\leq r\leq m,\\
&N_\ell\in\mathbb N_0,\quad \sum_\ell N_\ell=N_{\rm tot}.
\end{aligned}
\label{eq:reference-shot-design}
\end{equation}
The objective is a certified Bures-angle margin; applying $\sin^2$ gives the
predicted tail certificate. Under one uncertainty model, a nonmaximally
entangled reference is selected only when its fidelity gain offsets its larger
cumulative threshold.
Equation~\eqref{eq:reference-shot-design} remains conditional on the source and
measurement model; it does not identify a universally optimal target.

\paragraph{Data-dependent references and joint one-sided coverage.}
Equation~\eqref{eq:hoeffding-projector-lower} does not cover a target selected
and tested on the same observations because it treats the projector as fixed.
Held-out data avoid this selection effect.

\begin{proposition}[Held-out certificate for a selected reference]
\label{prop:held-out-selected-reference}
Let $\mathcal C$ be an arbitrary calibration record and let
$|\widehat\phi_{\mathcal C}\rangle$ be the normalized reference selected from
that record. Conditional on $\mathcal C$, freeze this reference and measure
$\{P_{\widehat\phi_{\mathcal C}},I-P_{\widehat\phi_{\mathcal C}}\}$ on
$N_F$ fresh, independent and identically distributed test rounds. If $S_F$ is
the number of projector outcomes, define
\[
\widehat F=\frac{S_F}{N_F},
\qquad
F_L=\max\left\{0,\widehat F-
\sqrt{\frac{\log(1/\epsilon_F)}{2N_F}}\right\}.
\]
Write $\widehat A_r$ for the sum of the largest $r-1$ Schmidt probabilities of
$|\widehat\phi_{\mathcal C}\rangle$. Then, with probability at least
$1-\epsilon_F$, the inequalities
\[
E_r(\rho)\geq
\mathcal R_{\widehat A_r}(F_L),
\qquad 2\leq r\leq m,
\]
hold simultaneously for the state sampled in the test rounds.
\end{proposition}

\begin{proof}
Conditional on $\mathcal C$, the selected projector and every
$\widehat A_r$ are fixed. Hoeffding's inequality gives
$F_{\widehat\phi_{\mathcal C}}(\rho)\geq F_L$ except on an event of conditional
probability at most $\epsilon_F$. Corollary
\ref{cor:confidence-qualified-projector-calibration} then gives all thresholds
on that same event; no union over $r$ is needed. Averaging over
$\mathcal C$ preserves the coverage.
\end{proof}

This prospective protocol is not implemented in Sec.~VI. Calibration declares
the target; the proposition then certifies fidelity and Vidal tails relative to
it, not the unknown source spectrum. It assumes stationarity and conditional
sampling. Drifting or adversarial rounds require another concentration bound.

Raw mode counts support a second construction when the reference bases,
relative phases, and projector model are fixed but the Schmidt magnitudes are
uncertain. Assume a calibrated Schmidt-basis measurement yields
$\boldsymbol C=(C_1,\ldots,C_m)$ from
$N_A=\sum_iC_i$ independent multinomial trials with probabilities
$\boldsymbol\nu=(\nu_1,\ldots,\nu_m)$. This requires a validated detector-to-mode
map; intensity counts do not certify phases, basis alignment, or leakage. For a prespecified
nonempty set $\mathcal I\subseteq\{2,\ldots,m\}$, let
$q=|\mathcal I|$ and define
\[
\widehat A_r
=\max_{|S|=r-1}\frac{1}{N_A}\sum_{i\in S}C_i,
\]
and
\[
\overline A_r
=\min\left\{1,\widehat A_r+
\sqrt{\frac{\log\binom{m}{r-1}+\log(q/\epsilon_A)}{2N_A}}
\right\}.
\]

\begin{proposition}[Fixed-basis raw-count one-sided certificate]
\label{prop:raw-count-joint-certificate}
Under the multinomial calibration model above,
$A_r\leq\overline A_r$ holds simultaneously for all $r\in\mathcal I$ with
probability at least $1-\epsilon_A$. If a valid fidelity analysis also gives
$F_\phi(\rho)\geq F_L$ except with probability $\epsilon_F$, then, with
probability at least $1-\epsilon_A-\epsilon_F$,
\[
E_r(\rho)\geq\mathcal R_{\overline A_r}(F_L)
\quad\text{for every }r\in\mathcal I.
\]
The fidelity event must refer to the same fixed bases, phases, and underlying
reference spectrum as the calibration event. The two failure events need not be
independent. Any operator-calibration error
must additionally be included through Corollary
\ref{cor:confidence-qualified-projector-calibration}.
\end{proposition}

\begin{proof}
For a fixed subset $S$ of $r-1$ mode labels, the sum
$\sum_{i\in S}C_i/N_A$ is an empirical mean of Bernoulli variables with mean
$\sum_{i\in S}\nu_i$. Hoeffding's inequality and a union bound over the
$\binom{m}{r-1}$ subsets give failure probability at most
$\epsilon_A/q$ for the stated upper bound at threshold $r$. A second union
bound over $r\in\mathcal I$ proves simultaneous spectrum coverage. Combining
that event with the fidelity event and applying Corollary
\ref{cor:confidence-qualified-projector-calibration} proves the result.
\end{proof}

The proposition excludes data-dependent Schmidt bases, unknown phases, and
projectors selected from the test counts. These cases require the joint model
below. To guarantee radius $\Delta_A$ for $q$ prespecified thresholds, require
\begin{equation}
N_A\geq
\frac{\log\binom{m}{r-1}+\log(q/\epsilon_A)}{2\Delta_A^2}.
\label{eq:spectrum-calibration-sample-complexity}
\end{equation}
The confidence term has a plus sign. At central high-dimensional thresholds,
the combinatorial factor may make this guarantee impractical, motivating a
structured spectrum model or simultaneous likelihood region.

For illustration, take $q=1$ and $\epsilon_A=0.01$, and let $r-1$ be the
central cut. Equation~\eqref{eq:spectrum-calibration-sample-complexity} gives
\begin{table}[t]
\caption{Distribution-free calibration counts from
Eq.~\eqref{eq:spectrum-calibration-sample-complexity} for a central cut,
$q=1$, and $\epsilon_A=0.01$.}
\label{tab:spectrum-sample-complexity}
\centering
\begin{tabular}{cccc}
\toprule
$m$ & $r-1$ & $N_A(0.05)$ & $N_A(0.01)$ \\
\midrule
8 & 4 & 1771 & 44269 \\
32 & 16 & 4964 & 124098 \\
1021 & 510 & 141724 & 3543089 \\
\bottomrule
\end{tabular}
\end{table}
These counts apply to unstructured simultaneous subset control, not a justified
low-parameter spectral model.

The combinatorial correction is conservative at large $m$. A multinomial
likelihood region may replace it if it has simultaneous coverage. Marginal
plug-in intervals fail when target selection and fidelity estimation overlap.
Use either held-out rounds as in Proposition
\ref{prop:held-out-selected-reference}, or construct a selection-uniform
confidence set $\mathcal C_{1-\epsilon}(D)$ for the complete state,
reference, and measurement model. A suitable nuisance parameter includes
\[
\begin{aligned}
\vartheta=(&\boldsymbol\nu,U_A,U_B,\text{phases},P_{\rm impl},\rho,p_{\rm in},\\
&\text{sampling and detector parameters}).
\end{aligned}
\]
Such a region supports the direct profile minimum
\begin{equation}
L_r(D):=\inf_{\vartheta\in\mathcal C_{1-\epsilon}(D)}
\mathcal R_{r,\phi(\vartheta)}
\bigl(F_{\phi(\vartheta)}(\rho(\vartheta))\bigr).
\label{eq:joint-confidence-direct-profile}
\end{equation}
On the common coverage event, $E_r(\rho)\geq L_r(D)$. A simpler rectangular
reduction uses
\[
F_L(D)=\inf_{\vartheta\in\mathcal C_{1-\epsilon}(D)}F(\vartheta),
\qquad
\overline A_r(D)=
\sup_{\vartheta\in\mathcal C_{1-\epsilon}(D)}A_r(\vartheta).
\]
Both preserve joint coverage; incompatible target-specific marginals do not.
With leakage, minimize the perspective
$p\mathcal R_{r,\phi}(F/p)$ over the same confidence region. The rectangular
specialization is Eq.~\eqref{eq:leakage-perspective-confidence}. Proposition
\ref{prop:ambient-extension} shows that occupation outside the nominal support
preserves the fixed-fidelity theorem; a certified local-filter outcome in
Corollary~\ref{cor:leakage-filtering} strengthens the unconditional bound on
$E_r(\rho)$.

One event for a fixed state--reference pair covers all $r$ without another
union bound. Multiple states, selected references, or target families require a
prespecified primary comparison, joint region, or family-level allocation such
as Bonferroni or Holm. Individual $1-\epsilon$ entries do not give table-wide
$1-\epsilon$ coverage.

For comparison with the analytical one-global-projector construction of
Ref.~\cite{pra16zhang}, define the largest-coefficient relaxation
\[
\Lambda_\phi^{(1)}(\rho)
:=\max\left\{
\frac{F_\phi(\rho)}{m\nu_1},\frac1m
\right\}.
\]

\begin{corollary}[Comparison with the largest-coefficient relaxation]
\label{cor:dominance-nu1-relaxation}
For every $F\in[0,1]$,
\[
\mathcal R_{r,\phi}(F)
\geq
\mathcal F_r\left(
\max\left\{\frac{F}{m\nu_1},\frac1m\right\}
\right).
\]
The cumulative-spectrum bound is positive at $F>A_r$; the largest-coefficient
bound requires $F>(r-1)\nu_1$. Since
$A_r\leq(r-1)\nu_1$, the first activation threshold is never higher. If
$A_r<(r-1)\nu_1$, the cumulative-spectrum-adapted bound is strictly positive
while the largest-coefficient relaxation vanishes throughout
$A_r<F\leq(r-1)\nu_1$. For a maximally entangled reference, the two bounds
coincide. This comparison does not rank all incomplete-data methods.
\end{corollary}

\begin{proof}
For a pure state with overlap $F$, the trace inequality used in
Ref.~\cite{pra16zhang} gives
$\lambda\geq F/(m\nu_1)$. Lemma~\ref{thm:exact-Fr-main} and monotonicity of
$\mathcal F_r$ therefore make the right-hand side a universal pure-state lower
bound at fixed $F$. Proposition~\ref{thm:single-observable-Er-main} gives the exact
minimum at that same $F$, proving the pointwise inequality. The threshold
comparison follows from $A_r=\sum_{i=1}^{r-1}\nu_i\leq(r-1)\nu_1$.
\end{proof}

\begin{corollary}[Maximally entangled reference]
\label{cor:maximally-entangled-reference}
Let $|\Phi_m^+\rangle=m^{-1/2}\sum_{i=1}^{m}|ii\rangle$ and
$F_+(\rho)=\langle\Phi_m^+|\rho|\Phi_m^+\rangle$. Then
\[
E_r(\rho)\geq\mathcal F_r\bigl(F_+(\rho)\bigr).
\]
The explicit right-hand side is zero for
$F_+(\rho)\leq(r-1)/m$ and equals
\[
\frac1m\left[
\sqrt{(m-r+1)F_+(\rho)}
-\sqrt{(r-1)(1-F_+(\rho))}
\right]^2
\]
above that threshold. A fidelity satisfying
$F_+(\rho)>(r-1)/m$ therefore implies $\operatorname{SN}(\rho)\geq r$.
\end{corollary}

The optimizer in \eqref{eq:two-block-optimal-spectrum-main} is blockwise uniform:
\begin{equation*} 
    \underbrace{
        \frac{1-t}{r-1},
        \ldots,
        \frac{1-t}{r-1}
    }_{\text{retained sector}}
    \ge
    \underbrace{
        \frac{t}{m-r+1},
        \ldots,
        \frac{t}{m-r+1}
    }_{\text{Schmidt tail}}.
\end{equation*}
Blockwise Cauchy--Schwarz equality and ordering impose it.
Equivalently, it maximizes \(\lambda\) at fixed \(E_r\) and minimizes \(E_r\) at
fixed \(\lambda\).

\section{Quantitative Vidal-tail calibration of Schmidt-number witnesses}

This section asks how a certified Vidal-tail lower bound changes a qualitative
Schmidt-number witness into a quantitative statement. Geometric measures,
positive-operator $S(k)$-norms, and Schmidt-number witnesses probe the same
bounded-Schmidt-number sets through different optimizations. Their support
identities are established, and we use them only as a transfer mechanism.

A standard witness reports whether a state lies outside a bounded-Schmidt-number
set. A certified tail value supplies more: it measures the distance, in the
geometric hierarchy used here, from the corresponding zero set and therefore
calibrates how far the observed value lies beyond a rank threshold. This
distinction matters in high dimensions, where two states can have the same
Schmidt number but very different weight above the first $k$ Schmidt modes.
Section~III provides the resource lower bound; the present section converts it
into a valid witness coefficient and states the conditions under which the
resulting operator is nontrivial and detecting.

We reuse the support reduction on $\mathcal S_k$ from
Refs.~\cite{norm1,XiongSze2026} and the established rank-one identity
\[
\bigl\|\,|\psi\rangle\langle\psi|\,\bigr\|_{S(k)}
=\sum_{i=1}^{k}\mu_i
\]
\cite{norm1}. The complement of this support coefficient is the pure Vidal tail
$E_{k+1}$.

For a pure reference projector, Schmidt approximation gives
\[
\max_{\sigma\in\mathcal S_{r-1}}
\operatorname{Tr}(P_\phi\sigma)=A_r.
\]
Corollary~\ref{cor:projector-witness-calibration} converts threshold violation
into a quantitative bound. The operator
$(1-L_{k+1}(\tau))I-\tau$ is another reused, generally nonoptimal witness
template: a geometric lower bound supplies a valid coefficient but does not
replace exact $S(k)$ optimization.

\subsection{Comparisons within the geometric hierarchy}

The sequence \(E_r\) decreases as its reference set grows. The next comparison
gives the sharp coefficient for the mixed-state extension.

\begin{proposition}\label{pureRGM1}
Let $|\psi\rangle\in\mathcal{H}_n\otimes\mathcal{H}_m$, with $m\leq n$, have Schmidt amplitudes
\[
  s_1\geq s_2\geq\cdots\geq s_m\geq 0 .
\]
For $2\leq r\leq m$, the geometric measure relative to states of Schmidt rank
at most $r-1$ satisfies
\begin{equation}\label{GMeq1}
  E_r(|\psi\rangle)
  \leq
  \cdots
  \leq
  E_2(|\psi\rangle)
  =
  E_G(|\psi\rangle),
\end{equation}
For $3\leq r\leq m$, one has
\begin{footnotesize}
    \begin{equation}\label{GMeq2}
  E_r(|\psi\rangle)
  \leq
  \frac{m+1-r}{m+2-r}
  E_{r-1}(|\psi\rangle)
  \leq
  \frac{m+1-r}{m-1}
  E_2(|\psi\rangle).
\end{equation}
\end{footnotesize}
\end{proposition}

\begin{proof}
Fix \(3\leq r\leq m\). The claim is immediate when
\(E_{r-1}(\lvert\psi\rangle)=0\); assume it is positive.
By definition,
\[
  E_r(|\psi\rangle)
  =
  1-\sum_{j=1}^{r-1}s_j^2 .
\]
Using $\sum_{j=1}^{m}s_j^2=1$, we obtain
\begin{align*}
  \frac{E_r(|\psi\rangle)}{E_{r-1}(|\psi\rangle)}
  &=
  \frac{1-\sum_{j=1}^{r-1}s_j^2}
       {1-\sum_{j=1}^{r-2}s_j^2}  \\
  &=
  \frac{s_r^2+\cdots+s_m^2}
       {s_{r-1}^2+s_r^2+\cdots+s_m^2}.
\end{align*}
Since the Schmidt amplitudes are ordered nonincreasingly,
\[
  s_{r-1}^2\geq s_r^2,
  \qquad
  \sum_{j=r}^{m}s_j^2
  \leq
  (m-r+1)s_r^2 .
\]
These ordering estimates imply
\begin{align*}
  \frac{E_r(|\psi\rangle)}{E_{r-1}(|\psi\rangle)}
  &\leq
  \frac{(m+1-r)s_r^2}
       {s_{r-1}^2+(m+1-r)s_r^2} \\
  &\leq
  \frac{(m+1-r)s_r^2}
       {s_r^2+(m+1-r)s_r^2} \\
  &=
  \frac{m+1-r}{m+2-r}.
\end{align*}
Hence
\[
  E_r(|\psi\rangle)
  \leq
  \frac{m+1-r}{m+2-r}E_{r-1}(|\psi\rangle).
\]
Iteration gives
\[
  E_r(|\psi\rangle)
  \leq
  \frac{m+1-r}{m-1}E_2(|\psi\rangle).
\]
Finally, $E_2(|\psi\rangle)=E_G(|\psi\rangle)$.
\end{proof}

If \(E_{r-1}(|\psi\rangle)>0\), the one-step bound is tight exactly when
\[
  s_{r-1}=s_r=\cdots=s_m>0.
\]
If \(E_{r-1}=0\), both sides vanish. For \(E_2>0\), the iterated bound is an
equality exactly when \(s_2=\cdots=s_m>0\); product states give the zero
endpoint. Maximally entangled spectra saturate every step, so both coefficients
are optimal for pure states.

The convex roof preserves this coefficient.

\begin{proposition}[Mixed-state hierarchy]\label{thm:mixed-geometric-hierarchy}
Let
\(
  \rho\in\mathcal{D}(\mathcal H_n\otimes\mathcal H_m)
\)
be a bipartite mixed state. For $2\leq r\leq m$, the geometric measures satisfy
\begin{equation}\label{GMmixeq01}
  E_r(\rho)
  \leq
  \cdots
  \leq
  E_2(\rho)
  =
  E_G(\rho),
\end{equation}
For $3\leq r\leq m$, the sharper inequalities are
\begin{equation}\label{GMmixeq02}
  E_r(\rho)
  \leq
  \frac{m+1-r}{m+2-r}E_{r-1}(\rho)
  \leq
  \frac{m+1-r}{m-1}E_2(\rho).
\end{equation}
\end{proposition}

\begin{proof}
The pure-state ordering and convex roofs give Eq.~\eqref{GMmixeq01}. For the
one-step coefficient, fix $3\leq r\leq m$ and $\varepsilon>0$. By the definition
of the convex
roof, choose a pure-state decomposition
\[
  \rho=\sum_i p_i|\psi_i\rangle\langle\psi_i|
\]
such that
\[
  \sum_i p_i E_{r-1}(|\psi_i\rangle)
  \leq
  E_{r-1}(\rho)+\varepsilon.
\]
Proposition~\ref{pureRGM1} and this decomposition give
\begin{align*}
  E_r(\rho)
  &\leq
  \sum_i p_i E_r(|\psi_i\rangle) \\
  &\leq
  \frac{m+1-r}{m+2-r}
  \sum_i p_i E_{r-1}(|\psi_i\rangle) \\
  &\leq
  \frac{m+1-r}{m+2-r}
  \bigl(E_{r-1}(\rho)+\varepsilon\bigr).
\end{align*}
Letting $\varepsilon\downarrow0$ and iterating gives
\[
  E_r(\rho)
  \leq
  \frac{m+1-r}{m-1}E_2(\rho).
\]
The identity $E_2(\rho)=E_G(\rho)$ completes the proof.
\end{proof}

The maximally entangled pure state saturates the coefficient within the
mixed-state domain, so no smaller state-independent value works. This does not
classify mixed equality cases.

\subsection{Support-norm identities and witness construction}

The optimizations differ: \(E_{k+1}(\rho)\) is a convex-roof infimum, whereas
the positive-operator \(S(k)\)-norm is a support supremum over
\(\mathcal V_k\). Their pure-state identity yields only a one-sided mixed-state
relation.

We use the standard reduction from bounded-Schmidt-number states to
bounded-Schmidt-rank pure states.

\begin{lemma}[Standard support reduction \cite{norm1,XiongSze2026}]
\label{prop:witness-coefficient}
Assume that $2\leq k\leq m\leq n$, and let $M=M^\dagger$ act on
$\mathcal H_n\otimes\mathcal H_m$. Let $\mathcal S_{k-1}$ denote the
set of density operators with Schmidt number at most $k-1$, and define
\[
    c_{k-1}
    :=
    \max_{\sigma\in\mathcal S_{k-1}}
    \operatorname{Tr}(M\sigma),
    \qquad
    W_k:=c_{k-1}I-M.
\]
Then
\[
    c_{k-1}
    =
    \max_{\substack{\langle\psi|\psi\rangle=1\\
                     \operatorname{SR}(|\psi\rangle)\leq k-1}}
    \langle\psi|M|\psi\rangle.
\]
The operator $W_k$ is therefore $(k-1)$-block positive. It is a nontrivial
$k$-Schmidt-number witness if and only if
\[
    \lambda_{\max}(M)>c_{k-1}.
\]
Under this condition, $W_k$ detects every maximal-eigenvalue eigenvector of $M$,
whose Schmidt rank is at least $k$. The maximizer defining $c_{k-1}$ need not
have Schmidt rank exactly $k-1$.
\end{lemma}

\begin{proof}
Compactness gives a pure-state maximizer. Linearity bounds every convex
combination of projectors from $\mathcal V_{k-1}$ by this maximum, and each
projector lies in $\mathcal S_{k-1}$. Hence
$\operatorname{Tr}(W_k\sigma)\geq0$ on $\mathcal S_{k-1}$. Its smallest
eigenvalue is $\lambda_{\min}(W_k)=c_{k-1}-\lambda_{\max}(M)$, which gives both the
nontriviality criterion and the claim about a maximal-eigenvalue eigenvector.
\end{proof}

Section III calibrates the support of a pure projector.

\begin{corollary}[Quantitative calibration of a projector witness]
\label{cor:projector-witness-calibration}
Let $\tau_\phi=|\phi\rangle\langle\phi|$ have ordered Schmidt probabilities
$\nu_1\geq\cdots\geq\nu_m$, and let
$A_r=\sum_{i=1}^{r-1}\nu_i$ and $B_r=1-A_r$. Then
\[
W_r(\phi):=A_r I-\tau_\phi
\]
is nonnegative on $\mathcal S_{r-1}$. It is a nontrivial
$r$-Schmidt-number witness if and only if $B_r>0$, equivalently
$\operatorname{SR}(|\phi\rangle)\geq r$. If a state $\rho$ violates the witness
by
\[
\delta_r:=-\operatorname{Tr}[W_r(\phi)\rho]
=F_\phi(\rho)-A_r>0,
\]
then
\begin{equation}\label{eq:quantitative-witness-calibration}
E_r(\rho)\geq
\left[
\sqrt{B_r(A_r+\delta_r)}
-\sqrt{A_r(B_r-\delta_r)}
\right]^2.
\end{equation}
At each admissible violation, this is the optimal state-independent lower bound
over all density operators and is attained by a pure state. A specified mixed
state need not saturate it.
\end{corollary}

\begin{proof}
Lemma~\ref{prop:witness-coefficient} and the Schmidt approximation
formula give
\[
\max_{\sigma\in\mathcal S_{r-1}}
\operatorname{Tr}(\tau_\phi\sigma)=A_r.
\]
Since $\lambda_{\max}(\tau_\phi)=1$, nontriviality is equivalent to $A_r<1$.
The detection condition is $F_\phi(\rho)>A_r$. Substituting
$F_\phi(\rho)=A_r+\delta_r$ into
Proposition~\ref{thm:single-observable-Er-main} gives
Eq.~\eqref{eq:quantitative-witness-calibration}. Fixed-violation optimality
follows from Corollary~\ref{cor:fixed-fidelity-value-dual}, with the pure
saturating spectrum in Eq.~\eqref{eq:reference-spectrum-optimizer}. On the
positive branch this is a witness interpretation of the established single-tail
characteristic boundary, not a new optimization curve.
\end{proof}

The same structure links the geometric measure to the $S(k)$-norm.

\begin{proposition}[Convex-roof support relation]\label{normGM}
Assume that \(1\le k<m\), and let
\(
  \rho\in\mathcal{D}(\mathcal H_n\otimes\mathcal H_m)
\)
be a bipartite mixed state. Then
\begin{equation*}
    1-E_{k+1}(\rho)
    =
    \sup_{\{p_i,|\psi_i\rangle\}\in\mathfrak E(\rho)}
    \sum_i p_i
    \bigl\||\psi_i\rangle\langle\psi_i|\bigr\|_{S(k)},
\end{equation*}
and
\begin{equation}\label{SKGM}
  \|\rho\|_{S(k)}
  \leq
  1-E_{k+1}(\rho).
\end{equation}
\end{proposition}

\begin{proof}
For every pure state $|\psi_i\rangle$, the geometric measure is
\[
  E_{k+1}(|\psi_i\rangle)
  =
  1-
  \max_{|\phi\rangle\in\mathcal V_k}
  |\langle\phi|\psi_i\rangle|^2 .
\]
For ordered Schmidt probabilities
\(\mu_1^{(i)},\ldots,\mu_m^{(i)}\), the rank-one $S(k)$ identity gives
\[
\bigl\||\psi_i\rangle\langle\psi_i|\bigr\|_{S(k)}
=\sum_{a=1}^{k}\mu_a^{(i)}.
\]
It follows that
\[
  E_{k+1}(|\psi_i\rangle)
  =
  1-
  \bigl\|
  |\psi_i\rangle\langle\psi_i|
  \bigr\|_{S(k)} .
\]
Every ensemble then satisfies
\[
\sum_i p_i E_{k+1}(|\psi_i\rangle)
=
1-
\sum_i p_i
\bigl\||\psi_i\rangle\langle\psi_i|\bigr\|_{S(k)}.
\]
Minimizing the left side and using $\inf(1-x)=1-\sup x$ proves the identity.
Convexity of the $S(k)$-norm gives
\[
  \|\rho\|_{S(k)}
  \leq
  \sum_i p_i
  \bigl\|
  |\psi_i\rangle\langle\psi_i|
  \bigr\|_{S(k)}
\]
Taking the ensemble supremum gives
Eq.~\eqref{SKGM}.
\end{proof}

\begin{remark}
For $\rho=|\psi\rangle\langle\psi|$ and every \(1\leq k<m\), equality holds in
Eq.~\eqref{SKGM}. With the squared-fidelity convention fixed in Sec.~II, the
case \(k=1\) reduces to
\begin{align*}
  E_2(|\psi\rangle)
  &=
  1-
  \max_{|\phi\rangle\in\mathcal V_1}
  |\langle\phi|\psi\rangle|^2 \\
  &=
  1-
  F\!\left(
  \rho,
  |\phi_{\mathrm{opt}}\rangle
  \langle\phi_{\mathrm{opt}}|
  \right) \\
  &=
  1-\|\rho\|_{S(1)},
\end{align*}
where $|\phi_{\mathrm{opt}}\rangle$ is a closest product state. For this optimizer,
\[
  \|\rho\|_{S(1)}
  =
  1-E_2(\rho)
  =
  F\!\left(
  \rho,
  |\phi_{\mathrm{opt}}\rangle
  \langle\phi_{\mathrm{opt}}|
  \right),
\]
so the $S(1)$-norm is the maximal squared fidelity with a product state.
\end{remark}

Lemma~\ref{lem4} and Proposition~\ref{normGM} turn any certified geometric lower
bound into a witness coefficient. Here ``nontrivial'' means block positive on
the specified Schmidt-number set but not positive semidefinite.

\begin{proposition}[Certified lower-bound construction]
\label{prop:certified-witness}
Let $\tau\in\mathcal D(\mathcal H_n\otimes\mathcal H_m)$, let
$1\leq k<m$, and suppose that a rigorously justified number
$L_{k+1}(\tau)$ satisfies
\[
0\leq L_{k+1}(\tau)\leq E_{k+1}(\tau).
\]
Define
\[
\widetilde W_{k+1}(\tau)
:=\bigl(1-L_{k+1}(\tau)\bigr)I-\tau.
\]
Then $\widetilde W_{k+1}(\tau)$ is $k$-block positive. It is a
nontrivial $(k+1)$-Schmidt-number witness if and only if
\[
\lambda_{\max}(\tau)>1-L_{k+1}(\tau).
\]
For any test state $\sigma$, it detects $\sigma$ precisely when
\[
\operatorname{Tr}(\tau\sigma)>1-L_{k+1}(\tau).
\]
\end{proposition}

\begin{proof}
The assumptions give
\[
\|\tau\|_{S(k)}\leq1-E_{k+1}(\tau)
\leq1-L_{k+1}(\tau).
\]
Lemma~\ref{lem4} proves $k$-block positivity. The smallest eigenvalue of
$\widetilde W_{k+1}(\tau)$ is
$1-L_{k+1}(\tau)-\lambda_{\max}(\tau)$, which proves the nontriviality
criterion; direct evaluation gives the detection condition.
\end{proof}

Setting $L_{k+1}(\tau)=E_{k+1}(\tau)$ gives the formal specialization
\[
W_{k+1}(\tau)=\bigl(1-E_{k+1}(\tau)\bigr)I-\tau.
\]
This is circular unless $E_{k+1}(\tau)$ is known. Nontriviality differs
from self-detection: the specialization detects its defining state only when
\(
\operatorname{Tr}(\tau^2)>1-E_{k+1}(\tau)
\).

The construction is noncircular for a known convex roof. For example, let
$\tau=\rho_x^{\rm iso}=a(I-\Phi_m)+x\Phi_m$, where
$a=(1-x)/(m^2-1)$ and $x\geq1/m^2$. Then
$L_{k+1}(\tau)=\mathcal F_{k+1}(x)$ is exact by
Proposition~\ref{prop:isotropic-exact-Er}, while
\[
\|\tau\|_{S(k)}=a+(x-a)\frac{k}{m}.
\]
Proposition~\ref{prop:certified-witness} gives a mixed, nonprojector witness when
$x>1-\mathcal F_{k+1}(x)$. Section VI compares a four-dimensional instance with
the exact support coefficient to quantify the certificate's conservatism.

Estimates of $L_{k+1}$ and the test overlap require joint coverage, through
simultaneous one-sided bounds or a union-bound budget. Calibration and test data
must be separate or modeled jointly.

Here one expectation means the global observable $\tau$. A decomposition
$\tau=\sum_\ell a_\ell A_\ell\otimes B_\ell$ still requires local settings and
sampling. Access to $L_{k+1}$ alone establishes no experimental advantage.

\section{Auxiliary comparisons with negativity and concurrence}

Partial-transpose negativity, convex-roof extended negativity (CREN), and
$I$-concurrence admit one-sided comparisons with Vidal tails. Their mixed-state
optimizations differ, so these are neither equivalences nor projector-fidelity
calibrations. The purpose of this section is to identify what the multiscale
geometric hierarchy implies for established entanglement measures. In
particular, a higher tail distinguishes the amount of Schmidt weight above a
specified rank threshold, whereas ordinary negativity aggregates all pairwise
products of Schmidt amplitudes.

We use the negativity normalization of Ref.~\cite{pra16zhang}. Here $k$ matches
the Sec.~IV $S(k)$ notation and denotes the threshold called $r$ in Sec.~III.

For projector calibration, Ref.~\cite{pra16zhang} uses
\[
\Lambda_\phi^{(1)}
=\max\left\{\frac{F_\phi}{m\nu_1},\frac1m\right\}.
\]
Section~III replaces this relaxation by the exact $A_r$-adapted envelope.

For negativity, Ref.~\cite{pra16zhang} gives the $k=2$ specialization
$\mathcal N_c\leq g_{2,m}(E_2)-1$. General CREN constructions and negativity
bounds on Schmidt number are also established
\cite{CREN03,EltschkaSiewert2015,Regula2018}.
Theorem~\ref{thm:negativity-geometric} gives a $k\geq3$ Vidal-tail
envelope in this setting together with its analytic inverse and its
if-and-only-if pure-state equality spectrum. This statement is restricted to the
combined higher-tail $E_k$ relation; it does not concern CREN itself or general
Schmidt-number bounds. The implication remains directional: a lower bound on
$E_k$ does not yield a lower bound on negativity.

For concurrence, Ref.~\cite{pra16zhang} gives a nonlinear mixed-state relation
for $E_2$. Proposition~\ref{prop:concurrence-geometric} gives a linear
higher-tail relation, with no improvement at $k=2$ and a possibly nonoptimal
coefficient for $m>2$.

For pure states, all Vidal tails reconstruct the spectrum and negativity.
Separate mixed-state roofs need not share a spectrum.

One convex-roof lemma covers both mixed-state transfers.

\begin{lemma}[Convex-roof transfer]
\label{lem:convex-roof-transfer}
Let $e$ and $q$ be continuous functions on pure states, and let $E$ and $Q$
denote their convex roofs. If $q(|\psi\rangle)\leq g(e(|\psi\rangle))$ for a
continuous, nondecreasing, concave function $g$, then
$Q(\rho)\leq g(E(\rho))$. If instead
$q(|\psi\rangle)\geq c\,e(|\psi\rangle)$ for a constant $c\geq0$, then
$Q(\rho)\geq cE(\rho)$.
\end{lemma}

\begin{proof}
For the first claim, apply the pure-state inequality and Jensen's inequality to
an $\varepsilon$-optimal decomposition for $E$. For the second, apply the
pointwise inequality to one for $Q$ and use the infimum defining $E$. Let
$\varepsilon\downarrow0$; no common optimizer is needed.
\end{proof}

\subsection{Geometric measures and negativity}

Use the partial-transpose normalization of Ref.~\cite{Negativity02}:
\begin{equation}
    \mathcal{N}_{\mathrm{PT}}(\rho)
    :=
    \left\|\rho^{T_B}\right\|_1-1,
    \label{eq:pt-negativity}
\end{equation}
where $T_B$ is the partial transpose and $\|\cdot\|_1$ the trace norm. This is
twice the convention
$(\|\rho^{T_B}\|_1-1)/2$.

Following Ref.~\cite{CREN03}, define the convex-roof extended negativity (CREN)
by
\begin{equation*}
    \mathcal{N}_{c}(\rho)
    :=
    \inf_{\{p_j,\lvert\psi_j\rangle\}\in\mathfrak E(\rho)}
    \sum_j p_j
    \mathcal{N}_{\mathrm{PT}}
    \bigl(\lvert\psi_j\rangle\langle\psi_j\rvert\bigr),
\end{equation*}
with the Sec.~II ensemble set $\mathfrak E(\rho)$. The symbols
$\mathcal N_{\rm PT}$ and $\mathcal N_c$ remain distinct.

Let
\begin{equation}
    \lvert\psi\rangle
    =
    \sum_{i=1}^{m}
    s_i\lvert a_i b_i\rangle,
    s_1\geq \cdots\geq s_m\geq0,
    \sum_{i=1}^{m}s_i^2=1,
    \label{eq:schmidt-state-negativity}
\end{equation}
be a pure state in $\mathcal{H}_n\otimes\mathcal{H}_m$, with $m\leq n$. Then
\begin{equation*}
    \mathcal{N}_{\mathrm{PT}}
    \bigl(\lvert\psi\rangle\langle\psi\rvert\bigr)
    =
    \left(\sum_{i=1}^{m}s_i\right)^2-1
    =
    2\sum_{i<j}s_i s_j.
\end{equation*}
The Sec.~III parameter satisfies
\begin{equation*}
    \lambda(\psi)
    =
    \frac{1}{m}
    \left(\sum_{i=1}^{m}s_i\right)^2
    =
    \frac{
        \mathcal{N}_{\mathrm{PT}}
        (\lvert\psi\rangle\langle\psi\rvert)+1
    }{m}.
\end{equation*}
Thus the pure-state bound reparametrizes Lemma~\ref{thm:exact-Fr-main}. Its
mixed-state transfer distinguishes \(\mathcal N_{\rm PT}\) from
\(\mathcal N_c\).

Set \(p:=k-1\), \(q:=m-k+1\), and
\begin{equation*}
    g_{k,m}(x)
    :=
    \left[
        \sqrt{p(1-x)}
        +
        \sqrt{qx}
    \right]^2,
    \qquad
    0\leq x\leq\frac{q}{m}.
\end{equation*}
Ordering gives \(0\leq E_k(\lvert\psi\rangle)\leq q/m\); the same bound holds
for $E_k(\rho)$. The function $g_{k,m}$ is
nondecreasing on $[0,q/m]$ and maps this interval onto $[k-1,m]$.
With the normalization above, every bipartite state satisfies
\[
0\leq\mathcal N_{\rm PT}(\rho)
\leq\mathcal N_c(\rho)\leq m-1.
\]

\begin{theorem}[Vidal-tail upper bound for negativity]
\label{thm:negativity-geometric}
Assume that $2\leq k\leq m\leq n$, and 
let \(\rho\in\mathcal{D}(\mathcal{H}_n\otimes\mathcal{H}_m)\) be a bipartite state. Then
\begin{equation}
    \mathcal{N}_{\mathrm{PT}}(\rho)
    \leq
    \mathcal{N}_{c}(\rho)
    \leq
    g_{k,m}\bigl(E_k(\rho)\bigr)-1.
    \label{eq:negativity-geometric-bound}
\end{equation}
\end{theorem}

Appendix C proves Theorem~\ref{thm:negativity-geometric}.
At $k=2$, Eq.~\eqref{eq:negativity-geometric-bound} is the
geometric-measure--CREN relation of Ref.~\cite{pra16zhang}. The higher-tail
extension includes the inverse and equality conditions below.

Since $g_{k,m}$ is nondecreasing on $[0,q/m]$, invert
Theorem~\ref{thm:negativity-geometric}. For
$\mathcal{M}\in\{\mathcal{N}_{\mathrm{PT}},\mathcal{N}_{c}\}$ with
$0\leq\mathcal M(\rho)\leq m-1$,
\begin{equation*}
    E_k(\rho)
    \geq
    \frac{1}{m^2}
    \left[
        \sqrt{
            q\bigl(\mathcal{M}(\rho)+1\bigr)
        }
        -
        \sqrt{
            p\bigl(m-\mathcal{M}(\rho)-1\bigr)
        }
    \right]_{+}^{\,2},
\end{equation*}
where $[x]_{+}:=\max\{x,0\}$. 
The zero-set characterization of $E_k$ then gives
\begin{equation*}
    \mathcal{M}(\rho)>k-2
    \quad\Longrightarrow\quad
    E_k(\rho)>0.
\end{equation*}
Thus $\mathcal M>k-2$ excludes decompositions with component ranks at most
$k-1$. For a known density operator, $\mathcal N_{\rm PT}$ is computable;
CREN retains a convex roof.

For a pure state $\lvert\psi\rangle$ with Schmidt amplitudes
$s_1\geq\cdots\geq s_m$, set $t:=E_k(\lvert\psi\rangle)$. Then
\begin{align}
    \mathcal{N}_{\mathrm{PT}}
    \bigl(\lvert\psi\rangle\langle\psi\rvert\bigr)+1
    &=
    \left(
        \sum_{i=1}^{m}s_i
    \right)^2                                                   \notag\\
    &\leq
    \left[
        \sqrt{p(1-t)}
        +
        \sqrt{qt}
    \right]^2                                                   \notag\\
    &=
    g_{k,m}(t).
    \label{eq:pure-negativity-bound}
\end{align}
Equality holds if and only if the Schmidt amplitudes are uniform within
each of the two blocks, namely,
\begin{footnotesize}
    \begin{align*}
    s_1=\cdots=s_{k-1}
    &=
    \sqrt{\frac{1-E_k(\lvert\psi\rangle)}{k-1}},
    \\
    s_k=\cdots=s_m
    &=
    \sqrt{\frac{E_k(\lvert\psi\rangle)}{m-k+1}}.
\end{align*}
\end{footnotesize}

These are the two-block spectra of Lemma~\ref{thm:exact-Fr-main}, so the
pure-state inequality is tight. CREN equality requires one decomposition to
minimize both $E_k$ and $\mathcal N_c$, with two-block-uniform nonzero components
and equality in Jensen's inequality. Strict concavity of $g_{k,m}$ forces equal
component tails except at degenerate endpoints. Equality in
$\mathcal N_{\rm PT}\leq\mathcal N_c$ requires equality in trace-norm
convexity for the partial transposes. We do not classify mixed-state saturation.

For pure states, the full hierarchy reconstructs the spectrum:
\begin{footnotesize}
\begin{align*}
    s_1^2&=1-E_2(|\psi\rangle),\\
    s_i^2&=E_i(|\psi\rangle)-E_{i+1}(|\psi\rangle),
    \qquad 2\leq i\leq m-1,\\
    s_m^2&=E_m(|\psi\rangle).
\end{align*}
\end{footnotesize}
Substitution yields
\begin{footnotesize}
   \begin{align*}
    &\mathcal{N}_{\mathrm{PT}}
    \bigl(\lvert\psi\rangle\langle\psi\rvert\bigr)\\
    &\quad=
    \Bigg[
        \sqrt{1-E_2(\lvert\psi\rangle)}\\
    &\qquad
        +
        \sum_{i=2}^{m-1}
        \sqrt{E_i(\lvert\psi\rangle)-E_{i+1}(\lvert\psi\rangle)}\\
    &\qquad
        +
        \sqrt{E_m(\lvert\psi\rangle)}
    \Bigg]^2-1.
\end{align*} 
\end{footnotesize}
Separate mixed-state convex roofs do not determine one Schmidt spectrum.

\subsection{Geometric measures and concurrence}
Following Refs.~\cite{Concurrence00,Concurrence01}, define the pure-state
\(I\)-concurrence by
\begin{equation*}
    C(\lvert\psi\rangle)
    :=
    \sqrt{
        2\left(
            1-\operatorname{Tr}\rho_A^2
        \right)
    },
    \qquad
    \rho_A
    :=
    \operatorname{Tr}_B
    \lvert\psi\rangle\langle\psi\rvert.
\end{equation*}
For a mixed state $\rho$, define the convex roof
\begin{equation*}
    C(\rho)
    :=
    \inf_{\{p_j,\lvert\psi_j\rangle\}\in\mathfrak E(\rho)}
    \sum_jp_jC(\lvert\psi_j\rangle).
\end{equation*}
For the state in Eq.~\eqref{eq:schmidt-state-negativity},
\begin{equation*}
    C^2(\lvert\psi\rangle)
    =
    2\left(1-\sum_{i=1}^{m}s_i^4\right)
    =
    4\sum_{i<j}s_i^2s_j^2.
\end{equation*}

\begin{proposition}[Dimension-dependent concurrence bound]\label{prop:concurrence-geometric}
Assume that \(2\leq k\leq m\leq n\). Let
\(\rho\in\mathcal{D}(\mathcal{H}_n\otimes\mathcal{H}_m)\)
be a bipartite state.
Then
\begin{equation}\label{GMnegeq2}
E_{k}(\rho)\leq\frac{1}{2}\sqrt{\frac{m-k+1}{k-1}}C(\rho).
\end{equation}
\end{proposition}

\begin{proof}
Let \(|\psi_j\rangle\in \mathcal{H}_n\otimes\mathcal{H}_m\) have Schmidt
amplitudes
\[
s_1^j\geq s_2^j\geq\cdots\geq s_m^j\geq0.
\]
Then
\[
E_k(|\psi_j\rangle)=
\sum_{i=k}^{m}(s_i^j)^2.
\]
The concurrence satisfies
\begin{equation*}
C^2(|\psi_j\rangle)
=
4\sum_{i<t}(s_i^j)^2(s_t^j)^2 .
\end{equation*}
Keeping cross terms between the first \(k-1\) amplitudes and the tail gives
\begin{align*}
C^2(|\psi_j\rangle)
&\geq
4
\left(\sum_{i=1}^{k-1}(s_i^j)^2\right)
\left(\sum_{t=k}^{m}(s_t^j)^2\right)  \\
&=
4\bigl(1-E_k(|\psi_j\rangle)\bigr)
E_k(|\psi_j\rangle).
\end{align*}
Ordering gives
\[
E_k(|\psi_j\rangle)
=
\sum_{i=k}^{m}(s_i^j)^2
\leq
\frac{m-k+1}{m}.
\]
For $E_k(|\psi_j\rangle)>0$, this gives
\[
\frac{1-E_k(|\psi_j\rangle)}{E_k(|\psi_j\rangle)}
\geq
\frac{k-1}{m-k+1}.
\]
Together these give
\[
C^2(|\psi_j\rangle)
\geq
4\frac{k-1}{m-k+1}
E_k^2(|\psi_j\rangle),
\]
and hence
\begin{equation*}
C(|\psi_j\rangle)
\geq
2\sqrt{\frac{k-1}{m-k+1}}
E_k(|\psi_j\rangle).
\end{equation*}

Lemma~\ref{lem:convex-roof-transfer} with
$c=2\sqrt{(k-1)/(m-k+1)}$ yields
\[
C(\rho)
\geq
2\sqrt{\frac{k-1}{m-k+1}}
E_k(\rho).
\]
Equivalently,
\[
E_k(\rho)
\leq
\frac{1}{2}
\sqrt{\frac{m-k+1}{k-1}}
C(\rho).
\]
This proves Eq.~\eqref{GMnegeq2}; zero tail is immediate.
\end{proof}

For $k=2$, this linear bound is weaker than the nonlinear mixed-state bound in
Ref.~\cite{pra16zhang}. For $m>2$, discarding intrablock terms and linearizing
$E_k(1-E_k)$ may weaken its coefficient. For example, the
maximally entangled pure state has
\[
E_k=\frac{m-k+1}{m},
\qquad
C=\sqrt{\frac{2(m-1)}{m}},
\]
and the linear inequality is strict for every $m>2$.

\begin{remark}
For pure \(|\psi\rangle\), the argument retains
\[
C^2(|\psi\rangle)
\geq
4E_k(|\psi\rangle)\bigl(1-E_k(|\psi\rangle)\bigr),
\]
and
\[
C(|\psi\rangle)
\geq
2\sqrt{E_k(|\psi\rangle)\bigl(1-E_k(|\psi\rangle)\bigr)}.
\]
The convex roof does not transfer this nonlinear inequality without another
convexity argument. For mixed states we retain only
\[
E_{k}(\rho)
\leq
\frac{1}{2}
\sqrt{\frac{m-k+1}{k-1}}C(\rho).
\]
The coefficient may be nonoptimal for $m>2$, even for pure states.
\end{remark}

The case $m=k=2$ yields a sharp coefficient.

\begin{corollary}
Let
\(
\rho\in\mathcal D(\mathcal H_n\otimes\mathcal H_2)
\)
be a bipartite state with one two-dimensional subsystem. Then
\begin{equation*}
    E_{2}(\rho)\leq \frac{1}{2}C(\rho).
\end{equation*}
\end{corollary}

Pure maximally entangled states saturate the coefficient \(1/2\), and product
states saturate the zero endpoint. Mixed equality cases are not classified.

For a pure qubit--qudit state, the exact relation
\[
C^2(|\psi\rangle)=
4E_2(|\psi\rangle)\bigl(1-E_2(|\psi\rangle)\bigr)
\]
can be inverted, since \(0\leq E_2(|\psi\rangle)\leq 1/2\), to give
\[
E_2(|\psi\rangle)
=
\frac{1-\sqrt{1-C^2(|\psi\rangle)}}{2}.
\]
The linear bound
\[
E_2(|\psi\rangle)\leq \frac{1}{2}C(|\psi\rangle)
\]
is saturated precisely at
\[
E_2(|\psi\rangle)=0
\qquad\text{or}\qquad
E_2(|\psi\rangle)=\frac{1}{2},
\]
corresponding to product and maximally entangled states on the two-dimensional
Schmidt support.

The negativity inequality is sharp on the stated pure two-block spectra and
extends to CREN. The concurrence estimate may be nonoptimal in higher
dimensions. Neither scalar reconstructs the mixed-state hierarchy.
\section{Numerical illustrations and experimental interfaces}
\label{sec:numerical-illustration}

We first illustrate partition refinement and the two weighted baselines because
they correspond to the principal results of Sec.~III. Isotropic states then
provide a known mixed-state calibration, followed by a finite-sample threshold
example. Further calculations examine non-isotropic mixtures, structured
fidelity inputs, and witness conservatism. Published spatial-mode and
time--frequency values are treated as trusted summaries. No calculation below
uses raw data to test the method end to end or independently evaluates an
unknown mixed-state convex roof; the raw-count audit states what either task
would require.

\subsection{Multiscale refinement and weighted-bound crossings}

Define the four-level benchmark
\[
\begin{aligned}
|\psi_{\rm b}\rangle
&=\frac{1}{\sqrt{0.9983}}
\bigl(0.70|00\rangle+0.51|11\rangle\\
&\hspace{5em}+0.39|22\rangle+0.31|33\rangle\bigr),\\
\tau_{\rm b}&=|\psi_{\rm b}\rangle\langle\psi_{\rm b}|.
\end{aligned}
\]
The decimal amplitudes define the benchmark. Its normalized Schmidt
probabilities are
\[
\boldsymbol\nu
=(0.4908344,0.2605429,0.1523590,0.0962636),
\]
rounded after normalization. Direct evaluation gives
\[
\begin{aligned}
E_2&=0.509166,& E_3&=0.248623,& E_4&=0.096264,\\
\lambda&=0.913578,&
\mathcal N_{\rm PT}&=2.654312,& C&=1.147797.
\end{aligned}
\]
The known fixed-$\lambda$ boundaries are $0.463448$, $0.219014$, and
$0.049870$ for $r=2,3,4$, respectively. Their positive gaps from the exact tails
show the information lost under single-scalar compression; these substitutions
are illustrative rather than independent validation.

The scalar $\lambda$ loses the location of Schmidt weight. For the refinement in
Theorem~\ref{thm:joint-multistep-main}, use
\[
\begin{aligned}
\mathcal P_2&=\{0,1,4\},&
\mathcal P_3&=\{0,1,2,4\},\\
\mathcal P_4&=\{0,1,2,3,4\}.
\end{aligned}
\]
Each unoptimized step splits one residual block and gives
\[
\begin{aligned}
\Lambda_{\mathcal P_2}&=0.937523,&
\Lambda_{\mathcal P_3}&=0.917942,\\
\Lambda_{\mathcal P_4}&=\lambda=0.913578.
\end{aligned}
\]
This is the refinement chain of Proposition~\ref{prop:partition-refinement}.
Singleton blocks reproduce $\sum_i\sqrt{\mu_i}$, so $\mathcal P_4$ uses the full
spectrum; tighter values require more tails.

In the known Schmidt basis, the Sec.~III observables satisfy
\[
\begin{aligned}
\operatorname{Tr}(Q_1\tau_{\rm b})&=E_2,
&\operatorname{Tr}(Q_2\tau_{\rm b})&=E_3,\\
\operatorname{Tr}(Q_3\tau_{\rm b})&=E_4.
\end{aligned}
\]
They commute and derive from one mode-resolved measurement, assuming a pure
state and known basis. Mixed-state populations need not equal convex-roof tails.

The weighted comparison uses the two bounds in
Eqs.~\eqref{eq:full-spectrum-weighted-baseline} and
\eqref{eq:weighted-nu1-baseline} for singleton blocks in dimension four. The
three prespecified cases in Table~\ref{tab:weighted-baseline-comparison} show,
respectively, full-spectrum dominance,
post-activation $\nu_1$ dominance, and activation of only the full-spectrum
bound.

\begin{table*}[t]
\caption{Weighted fixed-fidelity lower bounds for singleton blocks. The last
column is the bound that should be reported. All displayed values are generated
directly from the two independently defined baselines, not from the exact
weighted convex roof.}
\label{tab:weighted-baseline-comparison}
\centering
\begin{tabular}{ccccccc}
\toprule
Case & $\boldsymbol\nu$ & $\boldsymbol\omega$ & $F$ & $B_{\rm FS}$ &
$B_{\nu_1}$ & $B_{\rm best}$ \\
\midrule
I & $(0.40,0.30,0.20,0.10)$ & $(1,1,1)/3$ & 0.9256
& 0.139772 & 0.063678 & 0.139772 \\
II & $(0.28,0.27,0.24,0.21)$ & $(1,1,1)/3$ & 0.4700
& 0.012870 & 0.018105 & 0.018105 \\
III & $(0.40,0.30,0.20,0.10)$ & $(0,1,1)/2$ & 0.7500
& 0.001569 & 0 & 0.001569 \\
\bottomrule
\end{tabular}
\end{table*}

For case I, a prespecified $20001$-point uniform scan on $[0,1]$ found
$B_{\rm FS}\geq B_{\nu_1}$ at every sampled point. The largest sampled
difference was $0.251031$ at $F=1$, and the trapezoidal integral of the positive
difference was $0.017036$. This grid does not prove analytical dominance. Case
II gives the opposite ordering. In case III the first positive
weight occurs at the $E_3$ cut, so
$A_3=0.70<F=0.75\leq2\nu_1=0.80$; it lies in the strict activation interval
$B_{\rm FS}>0=B_{\nu_1}$. Hence Eq.~\eqref{eq:weighted-best-baseline} reports
their maximum. Neither bound is shown tight, and the weights are mathematical,
not an experimental payoff.

\subsection{Exact mixed-state benchmark: isotropic states}

Let $\Phi_m=|\Phi_m^+\rangle\langle\Phi_m^+|$ and define the isotropic family
\[
\rho_F^{\rm iso}
:=\frac{1-F}{m^2-1}(I-\Phi_m)+F\Phi_m,
\qquad 0\leq F\leq1.
\]
It satisfies
$\operatorname{Tr}(\Phi_m\rho_F^{\rm iso})=F$ and is invariant under
$U\otimes\overline U$ twirling.

\begin{proposition}[Known isotropic convex roof recovered]
\label{prop:isotropic-exact-Er}
For every $2\leq r\leq m$,
\[
E_r(\rho_F^{\rm iso})=\mathcal F_r(F).
\]
The maximally entangled specialization of
Proposition~\ref{thm:single-observable-Er-main} is therefore tight on this mixed
family.
\end{proposition}

The equality is the established isotropic Vidal convex roof
\cite{VollbrechtWerner2001,GirardGour2017}. It also follows by twirling a
single-tail equality state over the $U\otimes\overline U$ orbit. We use it only
as a mixed-state calibration and do not repeat the known derivation.

Take $m=4$ and $F=0.82$. The mixed state $\rho_{0.82}^{\rm iso}$ exceeds every
threshold $F>(r-1)/4$, certifying Schmidt number four.
Table~\ref{tab:isotropic-mixed-benchmark} lists the exact roofs, equal to the
Sec.~III bounds.

\begin{table}[t]
\caption{Exact Vidal-tail convex roofs for the mixed isotropic state
$\rho_{0.82}^{\rm iso}$ in dimension four. The witness violation is
$\delta_r=0.82-(r-1)/4$.}
\label{tab:isotropic-mixed-benchmark}
\centering
\begin{tabular}{cccc}
\toprule
$r$ & $(r-1)/4$ & $\delta_r$ & $E_r(\rho_{0.82}^{\rm iso})$ \\
\midrule
2 & 0.25 & 0.57 & 0.327284 \\
3 & 0.50 & 0.32 & 0.115813 \\
4 & 0.75 & 0.07 & 0.007284 \\
\bottomrule
\end{tabular}
\end{table}

A synthetic direct-projection example quantifies finite-sample loss. Take
$S_F=410$ projector outcomes in $N=500$ trials, so that
$\widehat F=S_F/N=0.82$, and set the one-sided failure probability to
$\epsilon=0.01$.
Equation
\eqref{eq:hoeffding-projector-lower} gives
\[
F_L=0.752139>\frac34.
\]
With confidence at least $0.99$, the same data give
\[
E_4(\rho)
\geq\mathcal R_{4,\Phi_4}(F_L)
=6.12\times10^{-6}>0,
\]
which certifies Schmidt number four. The quadratic threshold opening reduces
the nominal $0.07$ violation to
$6.12\times10^{-6}$. This calculation assumes a direct Bernoulli model. For a
local implementation of
Eq.~\eqref{eq:projector-local-decomposition}, the shot allocation and confidence
radius must follow Eq.~\eqref{eq:local-decomposition-hoeffding} instead.

For the same Bernoulli data, exact one-sided binomial inversion gives the
Clopper--Pearson limit
\[
F_L^{\rm CP}=0.776524,
\qquad
\mathcal R_{4,\Phi_4}(F_L^{\rm CP})=9.74\times10^{-4}.
\]
At a true success probability $F_0=0.82$, the exact binomial probability that
the lower endpoint crosses $3/4$ is $0.925$ for this rule, compared with $0.574$
for the displayed Hoeffding rule; the corresponding false-negative
probabilities are $0.075$ and $0.426$. For failure probability $0.01$ and target
power $0.9$, a Hoeffding design requires approximately $1370$, $2685$, and
$16776$ shots when the true margins $F_0-A_4$ are $0.07$, $0.05$, and $0.02$,
respectively. Exact binomial limits apply only to the direct two-outcome
measurement; a weighted
local-correlator estimator requires a concentration or likelihood analysis for
its own outcome model.

\subsection{Screening illustration based on published tomography summaries}
\label{subsec:self-testing-application}

Zhang \textit{et al.} self-tested photonic bipartite states up to dimension four
under fair sampling \cite{Zhang2019SelfTesting}. For two qubits, a semidefinite
program converted tilted-Bell violations into fidelity lower bounds $F_S$. In
higher dimensions, four $2\times2$ blocks determined a pure target spectrum,
but the moment-relaxation cost prevented evaluation of a global self-testing
fidelity bound. Trusted tomography instead gave target fidelities $F_T$. We
lack the raw Bell, block, and tomography counts, a joint target-spectrum region,
and a high-dimensional self-testing bound. We therefore treat $F_T$ as a trusted
summary, not an end-to-end self-testing reanalysis.

Prospectively, if a Bell-correlation analysis supplies a numerical lower bound
$F_S\leq F_\phi(\sigma_{\rm ext})$, Corollary
\ref{cor:self-testing-tail-conversion} gives
\[
E_r(\rho_{\rm phys})
\geq\mathcal R_{r,\phi}(F_S),
\qquad
F_S>A_r\Longrightarrow
\operatorname{SN}(\rho_{\rm phys})\geq r.
\]
For the two-qubit targets
$|\phi(\theta)\rangle=\cos\theta|00\rangle+\sin\theta|11\rangle$,
the only threshold is
$A_2=\max\{\cos^2\theta,\sin^2\theta\}$. For a nearly product target, the tail
bound becomes positive only when the self-testing fidelity is close to one.
The publication plots but does not tabulate its two-qubit $F_S$ values; we do
not use digitized values as self-testing inputs.

The high-dimensional figure reports ten tomography fidelities and target
coefficients. For deterministic screening, we digitized the central bar heights
$h_i^{(j)}$ from Fig.~4 of Ref.~\cite{Zhang2019SelfTesting} at 300 dpi and set
\[
\widetilde c_i^{(j)}
=\frac{h_i^{(j)}}{\sqrt{\sum_l(h_l^{(j)})^2}},
\qquad
\widetilde\nu_i^{(j)}
=\operatorname{sort}_{\downarrow}\bigl[(\widetilde c_i^{(j)})^2\bigr].
\]
An uncertainty-qualified digitization would instead assign pixel-resolution
intervals $h_i^{(j)}\in[\underline h_i^{(j)},\overline h_i^{(j)}]$ and solve
\[
\begin{aligned}
\overline A_r^{(j)}=\sup\Biggl\{
&\sum_{i=1}^{r-1}\nu_i:
\nu_i=\frac{h_i^2}{\sum_lh_l^2},\\
&h_i\in[\underline h_i^{(j)},\overline h_i^{(j)}],\quad
\boldsymbol\nu=\operatorname{sort}_{\downarrow}\boldsymbol\nu
\Biggr\}.
\end{aligned}
\]
The plot lacks a calibrated pixel-to-amplitude error model, so the central
digitization is not a confidence region and is reported at limited precision.
For a lower-shifted screening input, define
$F_j^\star:=F_{T,j}-\sigma_{T,j}$, where $\sigma_{T,j}$ is the reported standard
deviation. Table~\ref{tab:self-testing-screen} evaluates
$\mathcal R_{r,\widetilde\phi_j}(F_j^\star)$. The digitized heights and all
unrounded values are recorded in the accompanying script.

\begin{table*}[t]
\caption{Retrospective screening of the high-dimensional photonic data in
Fig.~4 of Ref.~\cite{Zhang2019SelfTesting}. The ordered target probabilities
$\widetilde{\boldsymbol\nu}^{(j)}$ are obtained from digitized central bar
heights. We use $F_j^\star=F_{T,j}-\sigma_{T,j}$ only as a one-standard-deviation
screening input. The last column is the largest activated threshold, not a
confidence-qualified Schmidt-number statement. Probabilities and bounds are
rounded to three decimal places.}
\label{tab:self-testing-screen}
\centering
\scriptsize
\begin{tabular}{ccccccc}
\toprule
State & $F_T\pm\sigma_T$ & $F_j^\star$
& $\widetilde{\boldsymbol\nu}^{(j)}$
& $\mathcal R_2$ & $\mathcal R_3$ & $\mathcal R_4$ / largest $r$ \\
\midrule
$\psi_0$ & $0.974\pm0.003$ & $0.971$
    & $(0.273,0.265,0.240,0.222)$ & $0.564$ & $0.297$ & $0.099$ / $4$ \\
$\psi_1$ & $0.960\pm0.001$ & $0.959$
    & $(0.266,0.250,0.247,0.237)$ & $0.539$ & $0.287$ & $0.090$ / $4$ \\
$\psi_2$ & $0.969\pm0.001$ & $0.968$
    & $(0.270,0.247,0.247,0.237)$ & $0.560$ & $0.309$ & $0.104$ / $4$ \\
$\psi_3$ & $0.968\pm0.002$ & $0.966$
    & $(0.254,0.253,0.250,0.243)$ & $0.571$ & $0.312$ & $0.105$ / $4$ \\
$\psi_4$ & $0.951\pm0.004$ & $0.947$
    & $(0.259,0.257,0.248,0.237)$ & $0.519$ & $0.262$ & $0.074$ / $4$ \\
$\psi_5$ & $0.966\pm0.003$ & $0.963$
    & $(0.484,0.263,0.148,0.105)$ & $0.326$ & $0.107$ & $0.019$ / $4$ \\
$\psi_6$ & $0.957\pm0.004$ & $0.953$
    & $(0.499,0.286,0.182,0.033)$ & $0.290$ & $0.068$ & $0$ / $3$ \\
$\psi_7$ & $0.955\pm0.002$ & $0.953$
    & $(0.407,0.399,0.098,0.096)$ & $0.376$ & $0.055$ & $0.009$ / $4$ \\
$\psi_8$ & $0.950\pm0.006$ & $0.944$
    & $(0.711,0.145,0.143,\approx0)$ & $0.104$ & $0.022$ & $0$ / $3$ \\
$\psi_9$ & $0.981\pm0.004$ & $0.977$
    & $(0.997,0.002,0.001,0)$ & $0$ & $0$ & $0$ / -- \\
\bottomrule
\end{tabular}
\end{table*}

The nonuniform targets show the effect of cumulative spectral data. At the
largest activated threshold for $\psi_5$, $\psi_6$, $\psi_7$, and $\psi_8$, the values
in Table~\ref{tab:self-testing-screen} are, respectively,
$0.019$, $0.068$, $0.009$, and $0.022$, whereas the
largest-coefficient relaxation of Ref.~\cite{pra16zhang} is zero in all four
cases. The state $\psi_6$ shows that a fidelity above $0.95$ need not
activate the fourth tail when the smallest reference probability is only about
$0.033$. The nearly product target $\psi_9$ is a negative control: its high
fidelity does not cross even $A_2$.

The input $F_T$ comes from tomography, not a high-dimensional self-testing
bound. The shift $F_T-\sigma_T$ is not a one-sided confidence limit, target
uncertainties are unavailable, and digitization does not replace counts. A
confidence-qualified analysis requires the block and tomography counts,
setting and outcome labels, trial or exposure totals, the accidental-count
policy, and the map from the block data to the selected target. It would also
need either a documented calibration--test split or a selection-uniform joint
analysis as described after Proposition
\ref{prop:raw-count-joint-certificate}. Block-resolved probabilities could
constrain the Sec.~III multistep profiles only after conversion into valid
constraints on mixed-state decompositions. The fitted target spectrum alone is
insufficient.

For rank-$r$ design, the margin is $F_L-A_r$. At full rank it becomes
$F_L-(1-\nu_m)$, so weak occupation of the last Schmidt mode demands greater
precision. Balanced spectra may therefore favor high Schmidt-number
certification.
The four-dimensional experiment used 64 projective measurements for its
blockwise self-testing analysis and 256 for tomography
\cite{Zhang2019SelfTesting}. Equation
\eqref{eq:projector-local-decomposition} instead contains 16 product
correlators and at most 13 settings after elementary grouping. The counts are
not directly comparable: self-testing reduces device assumptions, whereas the
projector route presumes calibrated local bases. They show a model-dependent
trust--measurement tradeoff, not a universal advantage.

\subsection{Fidelity-certified photonic experiments across dimensions}
\label{subsec:experimental-fidelity-interfaces}

Proposition~\ref{thm:single-observable-Er-main} needs a valid lower bound on
$F_\phi(\rho)$, not a direct measurement of $P_\phi$. Let
$\mathcal D_{\rm exp}$ denote the event on which an
experimental analysis establishes
$F_\phi(\rho)\geq F_L$ and, when the target spectrum is uncertain,
$A_r\leq\overline A_r$. Corollary
\ref{cor:confidence-qualified-projector-calibration} gives the composable chain
\[
\begin{aligned}
\mathcal D_{\rm exp}
&\ \Longrightarrow\ 
F_\phi(\rho)\geq F_L,\quad A_r\leq\overline A_r,\\
&\ \Longrightarrow\ 
E_r(\rho)\geq\mathcal R_{\overline A_r}(F_L).
\end{aligned}
\]
The second implication is deterministic and changes neither measurements nor
the assumptions defining $\mathcal D_{\rm exp}$. The same boundary can process
projector data, structured-correlation bounds, or post-extraction self-testing
fidelity.

The two-basis protocol of Bavaresco \textit{et al.} provides a
structured-correlation interface \cite{bavaresco2018measurements}. Standard-basis
data determine a target
$|\phi\rangle=\sum_i\lambda_i|ii\rangle$, and correlations in a tilted basis
give a lower bound $\widetilde F\leq F_\phi(\rho)$. Their Schmidt-number
threshold is
$B_k(\phi)=\sum_{i=1}^{k}\lambda_i^2$. With
$\nu_i=\lambda_i^2$ and $k=r-1$, this threshold is precisely $A_r$. The
condition $\widetilde F>A_r$ recovers rank certification, whereas
Proposition~\ref{thm:single-observable-Er-main} retains the violation size
through $E_r(\rho)\geq\mathcal R_{r,\phi}(\widetilde F)$. The conversion also
applies to this architecture, including the pixel-entanglement experiment of
Ref.~\cite{Valencia2020Pixel}, if its reported fidelity is a valid lower bound.
If overlapping data determine both target and fidelity, the pair
$(F_L,\overline A_r)$ must have joint coverage; separate point estimates are
insufficient for a confidence statement.

A four-dimensional experiment gives an analytic nonuniform example. Guo
\textit{et al.} considered
\cite{Guo2018GMLE}
\[
|\xi_2\rangle
=\frac{\sqrt3}{2}|00\rangle
+\frac{|11\rangle+|22\rangle+|33\rangle}{2\sqrt3},
\]
whose ordered Schmidt probabilities are
$(3/4,1/12,1/12,1/12)$. They reported a target fidelity
$0.991\pm0.003$. We use the lower-shifted value $F_G^\star=0.988$ only for
screening, without assigning unreported confidence. Since
$A_2=3/4$, $A_3=5/6$, and $A_4=11/12$, the adapted boundary evaluated at this
screening input returns
\[
\begin{aligned}
\mathcal R_{2,\xi_2}(F_G^\star)&=0.161703,
&\mathcal R_{3,\xi_2}(F_G^\star)&=0.093508,\\
\mathcal R_{4,\xi_2}(F_G^\star)&=0.033145.
\end{aligned}
\]
At the same input, the largest-coefficient relaxation in Corollary
\ref{cor:dominance-nu1-relaxation} gives $0.007660$, $0$, and $0$,
respectively. The full spectrum leaves the adapted third- and fourth-tail
outputs positive after the
$\nu_1$-only relaxation has become trivial. They become confidence-qualified
lower bounds only if $F_G^\star$ is replaced by a valid $F_L$. The original
experiment tested genuine multilevel entanglement under a specified local
tensor-product structure, distinct from bipartite Schmidt number. Our
post-processing gives a possible Vidal-tail interpretation of the reported
fidelity but neither replaces the witness nor compares noise tolerance.

For maximally entangled targets, the adapted and $\nu_1$ envelopes coincide, but
the output still exceeds a rank threshold. Hu \textit{et al.} reported
$F_+=0.933\pm0.001$ for a $32$-dimensional spatial-mode source and certified
Schmidt number at least $30$ \cite{Hu2020Multipath}. Using
$F_{32}^\star=0.932$ as a lower-shifted screening value, Corollary
\ref{cor:maximally-entangled-reference} returns
\[
\begin{aligned}
\mathcal F_2(F_{32}^\star)&=0.817396,
&\mathcal F_{16}(F_{32}^\star)&=0.275746,\\
\mathcal F_{30}(F_{32}^\star)&=0.002242.
\end{aligned}
\]
The last value is positive because $0.932>29/32$, consistent with the reported
rank threshold. The first two show the lower-order information retained beyond
that binary statement. The source analysis establishes the rank threshold; the
displayed values are screening outputs.

A time--frequency experiment reported a fidelity lower-bound estimate
$0.654\pm0.004$ for a $1021$-dimensional maximally entangled target and a maximum
certified Schmidt number of $668$ \cite{Chang2026TimeFrequency}. At the central
value, the boundary gives
$\mathcal F_2(0.654)=0.623938$,
$\mathcal F_{512}(0.654)=0.024156$, and
$\mathcal F_{668}(0.654)=5.71\times10^{-7}$. The article displays the quoted
error bars at
three standard deviations. If their lower endpoint $F_L=0.650$ is used without
further reinterpretation, the corresponding values are
\[
\begin{aligned}
\mathcal F_2(0.650)&=0.619866,
&\mathcal F_{512}(0.650)&=0.022884,\\
\mathcal F_{664}(0.650)&=4.45\times10^{-7},
\end{aligned}
\]
whereas the $r=668$ bound vanishes. Lower-order tails remain sizable, while the
highest activated rank is sensitive to a small fidelity shift near
$A_r=(r-1)/1021$. These numbers inherit the experiment's three-standard-
deviation convention and are not presented as distribution-free confidence
bounds.

\subsection{Raw-count status and joint-confidence requirements}
\label{subsec:raw-count-status}

The examples above are not raw-count reanalyses. For
Ref.~\cite{Zhang2019SelfTesting}, our files contain reported tomography
fidelities and digitized target amplitudes, but no setting-resolved Bell, block,
or tomography counts. Only this example fits both target spectrum and fidelity
from experimental data. Proposition~\ref{prop:held-out-selected-reference} would
apply if the target-selection record were separated from fresh fidelity-test
rounds and the latter admitted a valid one-sided analysis. If records overlap,
the target-selection rule must be included in a uniform confidence region.
The source article states that supporting data appear in the article and
Supplement or are available on request. We report only the records obtained for
this analysis, not that other records do not exist.

The other targets are fixed by design. The nonuniform spectrum in
Ref.~\cite{Guo2018GMLE} is specified analytically, while the targets in
Refs.~\cite{Hu2020Multipath,Chang2026TimeFrequency} are maximally entangled.
Their $A_r$ values need no statistical estimate unless the implemented
projectors differ from the nominal targets. Without complete counts and
metadata, we cannot recompute one-sided fidelity bounds; central values and
error-bar endpoints remain screening inputs.

A raw-data record should contain, for each setting, its label, outcome counts,
trials or acquisition time,
randomization and stopping rule, background and accidental-count treatment,
detector-efficiency correction, and timestamps or batch labels. It should also
specify the measurement operators, local-mode ordering, target-selection
algorithm, calibration--test relation, confidence level, allocation of failure
probability, and any dimension or leakage test. These fields select the
Bernoulli, multinomial, independent-term, or martingale model. Without them, neither
Proposition~\ref{prop:raw-count-joint-certificate} nor a sharper likelihood
region can be evaluated reproducibly.

The reproducibility files prepared for release at
\url{https://github.com/Liangxiong920/normGM} supply input templates and an
implementation of Propositions
\ref{prop:held-out-selected-reference} and
\ref{prop:raw-count-joint-certificate}. The program accepts integer counts and
records the confidence budget and selection protocol. It rejects an overlapping
selection-and-test declaration unless the input already supplies a valid joint
region. No published count data are bundled, and the program produces no
experimental certificate from the screening values in this section.

The tail profile does not replace entropy or key-rate estimates. For a pure
state, Vidal tails enter the optimal stochastic
LOCC conversion probability \cite{VidalSingleCopy}; several $r$ values retain
information lost by one Schmidt-number label. For mixed states, our convex-roof
lower bounds imply neither distillable entanglement nor a key rate.

For a fixed reference, a valid $F_L$ yields the vector
$\{\mathcal R_{r,\phi}(F_L)\}_{r=2}^{m}$ without further measurements. For
design, the margin is $F_L-A_r$: balanced targets favor large-$r$ certification,
whereas an adapted target may increase fidelity for an uneven source. Report the
target spectrum, one-sided fidelity and coverage, joint uncertainty for a
data-dependent target, and the certified subspace. A correlation or steering
statistic enters only through a proved conversion to $F_L$.

\paragraph{Single-scale consistency check.}
A generic constrained SLSQP calculation, supplied with neither the proportional
two-block ansatz nor its objective value, reproduces six instances of the known
single-tail boundary to an absolute objective gap below $4.0\times10^{-13}$ and
a constraint residual below $8.3\times10^{-13}$. Seeds, starts, and unrounded
outputs are included in the reproducibility archive. This substitution check is
not evidence for theorem priority and is omitted from the main tables.

\subsection{Exact non-isotropic two-qubit benchmark}
\label{subsec:exact-two-qubit-benchmark}

Dimension two has an independent exact convex roof. For a pure two-qubit state,
$E_2=f(C)$ with
\[
f(C)=\frac{1-\sqrt{1-C^2}}{2}.
\]
The function $f$ is increasing and convex. Mixed-state concurrence gives
$E_2(\rho)\geq f(C(\rho))$; Wootters' equal-concurrence decomposition gives the
reverse inequality \cite{Wootters1998}. Hence
\begin{equation}
E_2(\rho)=\frac{1-\sqrt{1-C(\rho)^2}}{2}
\label{eq:two-qubit-exact-E2}
\end{equation}
for every two-qubit density operator.

Consider the rank-two family
\begin{align*}
|\phi_0\rangle
&=\sqrt{0.55}|00\rangle+\sqrt{0.45}|11\rangle,\\
|\xi\rangle
&=(I\otimes U_{0.5})
\bigl(\sqrt{0.65}|00\rangle+\sqrt{0.35}|11\rangle\bigr),\\
\rho_p&=p|\phi_0\rangle\langle\phi_0|
+(1-p)|\xi\rangle\langle\xi|,
\end{align*}
where $U_{0.5}$ rotates by $0.5$ radians. The components have different Schmidt
spectra and local bases, so the family is neither isotropic nor an affine
reparametrization of one projector. Table~\ref{tab:exact-two-qubit-benchmark}
compares the main bound, Eq.~\eqref{eq:two-qubit-exact-E2}, and the fixed
reference $|\Phi_2^+\rangle$.

\begin{table*}[t]
\caption{Exact non-isotropic mixed-state benchmark. The ratio is
$\mathcal R_{2,\phi_0}(F_{\phi_0})/E_2(\rho_p)$. The final column uses the fixed
maximally entangled reference in the displayed basis.}
\label{tab:exact-two-qubit-benchmark}
\centering
\begin{tabular}{ccccccc}
\toprule
$p$ & $F_{\phi_0}$ & $C(\rho_p)$ & Exact $E_2(\rho_p)$ &
$\mathcal R_{2,\phi_0}$ & Ratio & $\mathcal R_{2,\Phi_2}$ \\
\midrule
0.40 & 0.857275 & 0.860154 & 0.244982 & 0.116234 & 0.474 & 0.143366 \\
0.60 & 0.904850 & 0.869405 & 0.252950 & 0.167564 & 0.662 & 0.199284 \\
0.80 & 0.952425 & 0.916102 & 0.299527 & 0.242959 & 0.811 & 0.278940 \\
0.90 & 0.976213 & 0.951945 & 0.346865 & 0.300756 & 0.867 & 0.337875 \\
\bottomrule
\end{tabular}
\end{table*}

The analytic bound approaches but does not saturate the exact roof across these
four points. The maximally entangled reference is stronger throughout. This
negative result prevents target-favoring selection and gives a relative-tightness
check, not a universal witness advantage.

\subsection{Non-isotropic mixed-state bracketing with an ensemble upper bound}
\label{subsec:colored-noise-benchmark}

Away from the isotropic orbit, let
$|\Phi_d^+\rangle=d^{-1/2}\sum_{i=1}^{d}|ii\rangle$ and define
\begin{align}
\rho_{y,\kappa}^{(d)}
&=y|\Phi_d^+\rangle\langle\Phi_d^+|
+(1-y)\sum_{i=1}^{d}q_i(\kappa)|ii\rangle\langle ii|,
\label{eq:colored-noise-family}\\
q_i(\kappa)&=
\frac{e^{-\kappa(i-1)}}{\sum_{j=1}^{d}e^{-\kappa(j-1)}}.\nonumber
\end{align}
For $\kappa>0$, this state is not isotropically invariant. Use the
basis-mismatched reference
\[
\begin{aligned}
|\phi_{\eta,\theta}^{(d)}\rangle
&=\sum_{i=1}^{d}\sqrt{\nu_i(\eta)}
|i\rangle\otimes U_\theta|i\rangle,\\
\nu_i(\eta)
&=\frac{e^{-\eta(i-1)}}{\sum_{j=1}^{d}e^{-\eta(j-1)}}.
\end{aligned}
\]
where $U_\theta$ rotates the first two modes by $\theta$. The benchmark fixes
$d=4$, $r=3$, $\eta=0.2$, $\kappa=0.45$, and $\theta=0.25$ radians; these values
define the test and are not optimized against the reported gaps.

For each visibility, Table~\ref{tab:colored-noise-bracket} reports the analytic
bound $L_{\rm FS}=\mathcal R_{3,\phi}(F_\phi)$ and the established
largest-coefficient comparator
\[
L_{\nu_1}=\mathcal F_3\!\left(
\max\left\{\frac{F_\phi}{4\nu_1},\frac14\right\}
\right).
\]
We obtain the upper estimate $U_{\rm HJW}$ by diagonalizing $\rho$ and searching
rank-$\rho$ ensembles in the Hughston--Jozsa--Wootters
parametrization \cite{HughstonJozsaWootters1993}. Generic complex Givens
rotations minimize average $E_3$ without the fixed-fidelity envelope or KKT
equation. Each feasible ensemble gives an upper bound; the search does not
certify global optimality.

\begin{table*}[t]
\caption{Non-isotropic colored-noise bracket. $U_{\rm HJW}$ is a constructive
ensemble upper estimate, not the exact convex roof. The final two columns report
the absolute and relative unresolved gaps, with the latter defined as
$(U_{\rm HJW}-L_{\rm FS})/U_{\rm HJW}$.}
\label{tab:colored-noise-bracket}
\centering
\begin{tabular}{cccccccc}
\toprule
$y$ & $F_\phi$ & $L_{\rm FS}$ & $L_{\nu_1}$ & $U_{\rm HJW}$ &
$U_{\rm HJW}-L_{\rm FS}$ & Relative gap & Reconstruction residual \\
\midrule
0.60 & 0.677924 & 0.006802 & 0.000221 & 0.076871 & 0.070068 & 0.912 & $<1.9\times10^{-15}$ \\
0.70 & 0.746999 & 0.025069 & 0.004553 & 0.117839 & 0.092770 & 0.787 & $<1.9\times10^{-15}$ \\
0.80 & 0.816073 & 0.057812 & 0.014559 & 0.174833 & 0.117021 & 0.669 & $<1.9\times10^{-15}$ \\
0.90 & 0.885147 & 0.111410 & 0.030603 & 0.262282 & 0.150872 & 0.575 & $<1.9\times10^{-15}$ \\
0.95 & 0.919685 & 0.150730 & 0.041077 & 0.329008 & 0.178278 & 0.542 & $<1.9\times10^{-15}$ \\
\bottomrule
\end{tabular}
\end{table*}

Here $A_3=0.598688$, and the affine fidelity $F_\phi(y)$ activates the main
bound at $y\simeq0.48529$. The standard maximally entangled projector has
$F_{\Phi_4}(\rho_y)=y+(1-y)/4$ and crosses its Schmidt-number-three threshold
$1/2$ already at $y>1/3$. Thus the family shows no detection advantage over the
standard witness. For these parameters, the cumulative-spectrum bound exceeds
the $\nu_1$ relaxation but remains below the independent ensemble upper bound.
Establishing equality for mixed states would require an independent lower
relaxation or globally certified convex-roof optimization.

\subsection{Cross-referenced isotropic consistency check for a nonmaximal reference}

Separate projector and state. Define
\[
|\phi_{\rm c}\rangle
=\sqrt{0.4}|00\rangle+\sqrt{0.3}|11\rangle
+\sqrt{0.2}|22\rangle+\sqrt{0.1}|33\rangle
\]
and let $\Phi_4=|\Phi_4^+\rangle\langle\Phi_4^+|$. The reference spectrum is
$\boldsymbol\nu=(0.4,0.3,0.2,0.1)$, and
\[
\lambda_{\rm c}:=|\langle\Phi_4^+|\phi_{\rm c}\rangle|^2
=0.944414.
\]
Let $\rho_x^{\rm iso}$ be centered on $\Phi_4$, not
$|\phi_{\rm c}\rangle$. Its expectation of the nonmaximal projector is
\[
F_{\rm c}(x)
:=\langle\phi_{\rm c}|\rho_x^{\rm iso}|\phi_{\rm c}\rangle
=x\lambda_{\rm c}
+\frac{(1-x)(1-\lambda_{\rm c})}{15}.
\]
Proposition~\ref{prop:isotropic-exact-Er} gives its convex roofs:
\[
E_r(\rho_x^{\rm iso})=\mathcal F_r(x).
\]
Compare them with the cumulative-spectrum certificate
$\mathcal R_{r,\phi_{\rm c}}(F_{\rm c}(x))$ and the
largest-coefficient certificate
\[
\mathcal F_r\!\left(
\max\left\{\frac{F_{\rm c}(x)}{4\nu_1},\frac14\right\}
\right)
\]
from Ref.~\cite{pra16zhang}.

With $c=(1-\lambda_{\rm c})/15$, the activation points are
\[
x_r^{\rm cum}=\frac{A_r-c}{\lambda_{\rm c}-c},
\qquad
x_r^{(1)}=\frac{(r-1)\nu_1-c}{\lambda_{\rm c}-c}.
\]
If $x_r^{(1)}>1$, the largest-coefficient certificate remains zero on the
family. At $x=0.98$, $F_{\rm c}=0.925600$, giving
Table~\ref{tab:cross-isotropic}.

\begin{table*}[t]
\caption{Cross-referenced comparison at $x=0.98$. The exact column gives the
known convex roof of $\rho_{0.98}^{\rm iso}$; the last columns use the same
projector expectation. A dash marks an activation point outside
$0\leq x\leq1$.}
\label{tab:cross-isotropic}
\centering
\begin{tabular}{ccccccc}
\toprule
$r$ & $A_r$ & $x_r^{\rm cum}$ & $x_r^{(1)}$
& Exact $E_r$ & $A_r$-adapted & $\nu_1$ relaxation \\
\midrule
2 & 0.4 & 0.421272 & 0.421272 & 0.618756 & 0.328001 & 0.111607 \\
3 & 0.7 & 0.740181 & 0.846484 & 0.360000 & 0.089247 & 0.006201 \\
4 & 0.9 & 0.952786 & --       & 0.138756 & 0.002068 & 0 \\
\bottomrule
\end{tabular}
\end{table*}

\begin{figure*}[t]
\centering
\begin{tikzpicture}
\begin{groupplot}[
group style={group size=3 by 1,horizontal sep=1.0cm},
width=0.30\textwidth,
height=0.23\textwidth,
xmin=0,xmax=1,
xlabel={$x$},
grid=major,
tick label style={font=\scriptsize},
label style={font=\small},
title style={font=\small},
every axis plot/.append style={mark=none,thick}
]
\nextgroupplot[
title={$r=2$},ylabel={bound on $E_r$},ymin=0,ymax=0.68,
legend style={font=\scriptsize,at={(0.03,0.97)},anchor=north west,
draw=none,fill=none}]
\addplot[black] table[x=x,y=exact2]
{normGM-260721-v1-cross-isotropic.dat};
\addlegendentry{exact convex roof}
\addplot[blue!75!black] table[x=x,y=adapted2]
{normGM-260721-v1-cross-isotropic.dat};
\addlegendentry{$A_r$-adapted}
\addplot[red!75!black,dashed] table[x=x,y=nu12]
{normGM-260721-v1-cross-isotropic.dat};
\addlegendentry{$\nu_1$ relaxation}
\nextgroupplot[title={$r=3$},ymin=0,ymax=0.42]
\addplot[black] table[x=x,y=exact3]
{normGM-260721-v1-cross-isotropic.dat};
\addplot[blue!75!black] table[x=x,y=adapted3]
{normGM-260721-v1-cross-isotropic.dat};
\addplot[red!75!black,dashed] table[x=x,y=nu13]
{normGM-260721-v1-cross-isotropic.dat};
\nextgroupplot[title={$r=4$},ymin=0,ymax=0.18]
\addplot[black] table[x=x,y=exact4]
{normGM-260721-v1-cross-isotropic.dat};
\addplot[blue!75!black] table[x=x,y=adapted4]
{normGM-260721-v1-cross-isotropic.dat};
\addplot[red!75!black,dashed] table[x=x,y=nu14]
{normGM-260721-v1-cross-isotropic.dat};
\end{groupplot}
\end{tikzpicture}
\caption{Exact Vidal-tail convex roofs and two fixed-fidelity bounds for the
cross-referenced isotropic family. The state is centered on $\Phi_4$, while the
measured projector is $|\phi_{\rm c}\rangle\langle\phi_{\rm c}|$. The gap to the
exact curve shows that this mixed family does not saturate the
bound associated with the nonmaximally entangled reference.}
\label{fig:cross-isotropic}
\end{figure*}
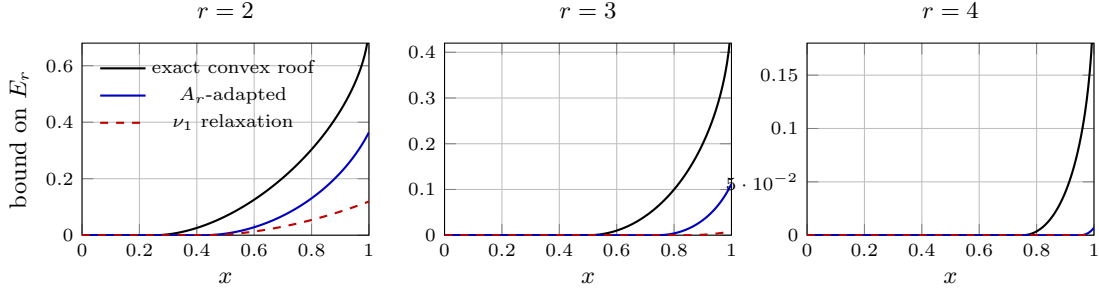

At $r=2,3$, the adapted bound exceeds the $\nu_1$ comparator; at $r=4$, only it
is nonzero. All lie below the known roofs. The pointwise gain agrees with
Corollary~\ref{cor:fixed-fidelity-value-dual} but neither saturates this family
nor constructs a strictly mixed saturator for the nonmaximal reference.

To scan reference-spectrum dependence, set
\[
\boldsymbol\nu(s)=\frac14(1+3s,1+s,1-s,1-3s),
\qquad 0\leq s\leq\frac13,
\]
for fixed $\rho_{0.98}^{\rm iso}$. This path joins the maximally entangled
spectrum to a rank-three boundary. Figure~\ref{fig:reference-scan} plots the
difference from the $\nu_1$ relaxation. It vanishes at $s=0$ and is positive on
parts of the sampled interior. Nonmonotonicity reflects lost cumulative-threshold
violations. The largest gaps on the $201$-point scan are
$0.256599$, $0.173351$, and $0.087335$
for $r=2,3,4$, respectively.

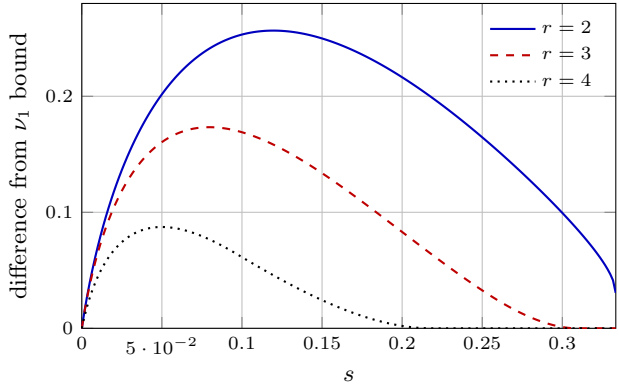
\begin{figure}[t]
\centering
\begin{tikzpicture}
\begin{axis}[
width=\columnwidth,
height=0.68\columnwidth,
xmin=0,xmax=0.3334,
ymin=0,ymax=0.28,
xlabel={$s$},
ylabel={difference from $\nu_1$ bound},
grid=major,
tick label style={font=\scriptsize},
label style={font=\small},
legend style={font=\scriptsize,draw=none,fill=none,at={(0.98,0.98)},anchor=north east},
every axis plot/.append style={mark=none,thick}
]
\addplot[blue!75!black] table[x=s,y=gap2]
{normGM-260721-v1-reference-scan.dat};
\addlegendentry{$r=2$}
\addplot[red!75!black,dashed] table[x=s,y=gap3]
{normGM-260721-v1-reference-scan.dat};
\addlegendentry{$r=3$}
\addplot[black,dotted] table[x=s,y=gap4]
{normGM-260721-v1-reference-scan.dat};
\addlegendentry{$r=4$}
\end{axis}
\end{tikzpicture}
\caption{Reference-spectrum scan at fixed
$\rho_{0.98}^{\rm iso}$. The ordinate is
$\mathcal R_{r,\phi(s)}(F_{\phi(s)})$ minus the
largest-coefficient lower bound. Nonmonotonicity reflects both retained
cumulative information and loss of threshold violation for highly nonuniform
references.}
\label{fig:reference-scan}
\end{figure}

\subsection{Consistency benchmark for a mixed calibration operator}

Take the mixed calibration operator
\[
\begin{aligned}
\tau=\rho_{0.99}^{\rm iso}
&=a(I-\Phi_4)+0.99\Phi_4\\
&=aI+(0.99-a)\Phi_4,
\qquad a=\frac{0.01}{15},
\end{aligned}
\]
with the known value
$L_{k+1}(\tau)=\mathcal F_{k+1}(0.99)$ in
Proposition~\ref{prop:certified-witness}. The operator
\[
\widetilde W_{k+1}(\tau)
=\bigl[1-\mathcal F_{k+1}(0.99)\bigr]I-\tau
\]
is $k$-block positive. Although written as a nonprojector, $\tau$ is affine in
$\Phi_4$, so $cI-\tau$ rescales its projector witness. The example checks the
coefficient, not detection or measurement advantage.

Use the full-rank white-noise family
\[
\sigma_y=y|\phi_{\rm c}\rangle\langle\phi_{\rm c}|+(1-y)\frac{I}{16}.
\]
At $y=0.98$,
\[
\begin{aligned}
F_{\phi_{\rm c}}(\sigma)&=0.981250,
&F_+(\sigma)&=0.926776,\\
\operatorname{Tr}(\tau\sigma)&=0.917557.
\end{aligned}
\]
Both $P_{\phi_{\rm c}}$ and $\Phi_4$ certify
$\operatorname{SN}(\sigma)=4$; the mixed-operator coefficient is weaker.

\begin{table*}[t]
\caption{Witness comparison on $\sigma_y$. Here $y_{\rm c}$ is the strict
detection threshold, $1-y_{\rm c}$ the white-noise tolerance, and the coefficient
gap is relative to the exact support of the same observable. The index $k$
denotes the support rank, so detection certifies Schmidt number at least $k+1$.
Each calibration uses 16 product correlators and at most 13 settings after
diagonal grouping; identity offsets add no setting.}
\label{tab:mixed-nonprojector-witness}
\centering
\scriptsize
\begin{tabular}{lcccccc}
\toprule
Calibration & $k$ & Expectation threshold & $y_{\rm c}$ & $1-y_{\rm c}$
& Coefficient gap & Terms/settings \\
\midrule
$P_{\phi_{\rm c}}$ (exact projector support) & 1 & $0.400000$ & $0.360000$ & $0.640000$ & $0$ & $16/\leq13$ \\
$P_{\phi_{\rm c}}$ (exact projector support) & 2 & $0.700000$ & $0.680000$ & $0.320000$ & $0$ & $16/\leq13$ \\
$P_{\phi_{\rm c}}$ (exact projector support) & 3 & $0.900000$ & $0.893333$ & $0.106667$ & $0$ & $16/\leq13$ \\
\midrule
$\Phi_4$ (maximally entangled projector) & 1 & $0.250000$ & $0.212606$ & $0.787394$ & $0$ & $16/\leq13$ \\
$\Phi_4$ (maximally entangled projector) & 2 & $0.500000$ & $0.496080$ & $0.503920$ & $0$ & $16/\leq13$ \\
$\Phi_4$ (maximally entangled projector) & 3 & $0.750000$ & $0.779554$ & $0.220446$ & $0$ & $16/\leq13$ \\
\midrule
$\tau$ with exact $S(k)$ support & 1 & $0.248000$ & $0.212606$ & $0.787394$ & $0$ & $16/\leq13$ \\
$\tau$ with exact $S(k)$ support & 2 & $0.495333$ & $0.496080$ & $0.503920$ & $0$ & $16/\leq13$ \\
$\tau$ with exact $S(k)$ support & 3 & $0.742667$ & $0.779554$ & $0.220446$ & $0$ & $16/\leq13$ \\
\midrule
$\tau$ with certified geometric coefficient & 1 & $0.341168$ & $0.319388$ & $0.680612$ & $0.093168$ & $16/\leq13$ \\
$\tau$ with certified geometric coefficient & 2 & $0.599499$ & $0.615466$ & $0.384534$ & $0.104165$ & $16/\leq13$ \\
$\tau$ with certified geometric coefficient & 3 & $0.831168$ & $0.880988$ & $0.119012$ & $0.088502$ & $16/\leq13$ \\
\bottomrule
\end{tabular}
\end{table*}

The $\Phi_4$ and exact $S(k)$ witnesses share $y_{\rm c}$ because
$\tau=aI+(0.99-a)\Phi_4$. The last three gaps measure the cost of replacing
$\|\tau\|_{S(k)}$ by $1-E_{k+1}(\tau)$. This checks
Proposition~\ref{prop:certified-witness} but gives no projector-witness advantage.
A better non-isotropic operator remains unknown.

\subsection{Computational cost, reproducibility, and limitations}

The script \path{normGM-260814-v1-numerics.py} and its JSON/CSV outputs
reproduce the weighted three-baseline comparison, direct
pure-state tests, the exact two-qubit family, the colored-noise brackets, and the
finite-sample calculations. It records the Python, NumPy, and SciPy versions and
uses seed 20260806. Pure-state rows use ten SLSQP starts on a general complex
matrix. Each HJW estimate uses eight starts and 2500 complex Givens proposals
per start; the largest reconstruction
residual is $1.84\times10^{-15}$. Neither the $20001$-point weighted scan nor the
scalar KKT solve uses fitted rows.

Separate files contain the published-summary calculations.
The script \path{normGM-260812-v1-self-testing.py} records digitized bar
heights, normalization, screening fidelities, adapted bounds, and
largest-coefficient comparators in Table~\ref{tab:self-testing-screen}.
The script \path{normGM-260812-v2-experimental-interfaces.py} records the exact
reference spectra and fidelity inputs used in
Sec.~\ref{subsec:experimental-fidelity-interfaces} and reproduces every
stated Vidal-tail value. An accompanying script implements the
held-out and multinomial constructions in Propositions
\ref{prop:held-out-selected-reference} and
\ref{prop:raw-count-joint-certificate}. Its JSON template contains no
experimental observations and rejects incomplete or statistically incompatible
inputs.

For a pure spectrum, tails and block profiles cost $O(m)$. The relaxed weighted
envelope adds an $O(L)$ scalar root solve per evaluation.
Proposition~\ref{prop:isotropic-exact-Er} gives exact $E_r$ for the isotropic
family, while Eq.~\eqref{eq:two-qubit-exact-E2} gives the independent exact
two-qubit benchmark. In the cross-reference calculation,
exact values belong to the isotropic state; the curves based on the nonmaximally
entangled reference are lower bounds. The code performs neither direct
convex-roof optimization for such a reference nor an SDP lower relaxation for
colored noise. Published-data scripts only post-process summaries; they do not
reconstruct Bell or local correlations, reanalyze raw counts, infer confidence limits, or propagate
target-spectrum errors beyond the stated inputs. The joint-confidence script
requires integer counts, a selection protocol, and a valid failure-probability
budget.

\section{Conclusions}

The main results resolve multiscale Schmidt-spectrum information. Several nested
Vidal tails fix the masses of consecutive blocks, and the normalized
nuclear-norm coordinate then has a sharp upper boundary. Equality holds exactly
for spectra that are uniform within each block. A supplied partition refinement
cannot weaken the boundary and is strict whenever the new subblocks have unequal
mean probabilities; singleton blocks recover the full-spectrum value. To our
knowledge, this is the first exact multistep boundary that includes both the
equality structure and the refinement criterion. Its gain comes from additional
tail data, not from one projector expectation.

Nonnegative weights turn the tail vector into a scalar profile. We obtained the
global solution of the relaxed block-mass problem. On the nontrivial branch, one
scalar equation determines a unique full-support optimizer and an explicit
value. General Hellinger-constrained optimization is established, so the claim
is specifically the weighted Vidal-tail formulation and solution. The theorem
does not evaluate the exact fixed-fidelity convex roof of a general weighted
profile. For mixed-state inference we therefore retain the certified baselines
$B_{\rm FS}$ and $B_{\nu_1}$ and use their maximum. The full-spectrum bound
activates no later, but neither bound dominates wherever both are positive.

The established single-tail fidelity--resource curve remains a supporting
baseline. The additional single-scale statements are an exact-fidelity pure
completion of the zero branch and uniqueness of the ordered positive-branch
spectrum. Spectrum uniqueness neither classifies state vectors nor excludes
basis rotations within degenerate Schmidt subspaces.

The resource bounds sharpen the interpretation of Schmidt-number witnesses.
Because $A_r$ is the projector support on $\mathcal S_{r-1}$, the same measured
violation that detects Schmidt number at least $r$ yields a quantitative lower
bound on $E_r$. This calibration distinguishes states that cross the same rank
threshold by different spectral margins. It uses established support and
$S(k)$ identities and may remain conservative for a general nonprojector
operator. The higher-tail CREN relation gives a complementary transfer: for
$k\geq3$, it bounds convex-roof extended negativity by $E_k$, admits an analytic
inverse, and identifies all pure equality spectra. To our knowledge, no previous
higher-tail theorem combines these three components; CREN, its $k=2$
specialization, and general negativity bounds on Schmidt number remain prior
work.
The concurrence comparison is only a nonoptimal analytical estimate.

Leakage follows from strong monotonicity, and the statistical interface
propagates reference, fidelity, and subspace uncertainty through joint
confidence regions. These statements require the stated calibration and
measurement assumptions. A global projector expectation is one
information-theoretic input, not necessarily one local setting. The numerical
section illustrates refinement, weighted-bound crossings, isotropic calibration,
and parameter-specific mixed-state brackets; published photonic summaries are
screening inputs rather than raw-count reanalyses. Open problems include the
exact weighted fixed-fidelity convex roof, complete equality-state
classification, strictly mixed or full-rank saturators for nonmaximal
references, an independent rigorous lower relaxation for colored noise, and an
end-to-end analysis based on setting-resolved experimental counts.

\begin{acknowledgments}
Sze was supported by the Hong Kong Research Grants Council (PolyU
15300121), The Hong Kong Polytechnic University (4-ZZRN). Jin was supported by the National Natural
Science Foundation of China 12301582 and GDSTA (SKXRC2025442).  
Wang was supported by the National Natural Science Foundation of China  12371458, the Guangdong Basic and Applied Basic Research Foundation (2024A1515030023), and  the Guangdong Association for Science and Technology Young Scientific and Technological Talent Cultivation Program (Grant No. SKXRC2026457).
GPT was used to assist with language editing, LaTeX organization, and
preparation of the reproducibility package. The authors retain responsibility
for the mathematical arguments, numerical results, citations, and conclusions.
\end{acknowledgments}

\section*{Competing Interests}

The authors declare no competing interests.

\section*{Data Availability}

We generated no experimental data. The data and software supporting this study
consist of numerical outputs and the scripts used to produce them. Version
\texttt{v1.0.0} is available under the MIT License from
\url{https://github.com/Liangxiong920/normGM/releases/tag/v1.0.0} and is
archived by Zenodo under DOI
\url{https://doi.org/10.5281/zenodo.22208160}
\cite{XiongSzeNormGM2026}. The numerical benchmark requires Python 3.10 or
later with NumPy and SciPy; the release documentation states the assumptions
of the published-summary and confidence-interface scripts. Its statistical
input templates are empty and contain no observations.

Section~VI also post-processes values reported in cited publications and central
bar heights digitized from a 300-dpi rendering of Fig.~4 in
Ref.~\cite{Zhang2019SelfTesting}. The integer pixel heights are provided for
reproducibility, but the source figure reports no calibrated
pixel-to-amplitude uncertainty. These values are screening inputs, not a
confidence region. We report the resulting probabilities and bounds to three
decimal places; full-precision machine outputs serve only as computational
records. We neither obtained nor redistribute the setting-resolved Bell, block,
or tomography counts required for joint one-sided confidence regions. The
release package contains generated numerical records and empty count schemas,
not raw observations. The screening calculations yield experimental
certificates only when combined with raw records and a statistical model matched
to the acquisition protocol.

\appendix

\section{Proof of Lemma \ref{thm:exact-Fr-main}}

\begin{proof}
Set
\begin{equation*}
    p=r-1,
    \qquad
    q=m-r+1,
    \qquad
    p+q=m.
\end{equation*}
Let
\begin{equation*}
    t
    =
    E_r(|\psi\rangle)
    =
    \sum_{i=r}^{m}\mu_i
\end{equation*}
be the last-\(q\) Schmidt weight. Then
\begin{equation*}
    \sum_{i=1}^{p}\mu_i=1-t,
    \qquad
    \sum_{i=p+1}^{m}\mu_i=t.
\end{equation*}

Blockwise Cauchy--Schwarz gives
\begin{align}
    \sum_{i=1}^{p}\sqrt{\mu_i}
    &\le
    \sqrt{p(1-t)},
    \label{eq:head-CS-main}
    \\
    \sum_{i=p+1}^{m}\sqrt{\mu_i}
    &\le
    \sqrt{qt}.
    \label{eq:tail-CS-main}
\end{align}
The constraint
\begin{equation*}
    \lambda
    =
    \frac{1}{m}
    \left(
        \sum_{i=1}^{m}\sqrt{\mu_i}
    \right)^2,
\end{equation*}
then implies
\begin{equation}
    \sqrt{m\lambda}
    \le
    \sqrt{p(1-t)}
    +
    \sqrt{qt}.
    \label{eq:lambda-t-inequality-main}
\end{equation}

Ordering implies that the mean of the first block is at least that of the
second block:
\begin{equation*}
    \frac{1-t}{p}
    \ge
    \frac{t}{q}.
\end{equation*}
Equivalently,
\begin{equation*}
    0\le t\le\frac{q}{m}.
\end{equation*}

Define
\begin{equation*}
    h_r(t)
    =
    \frac{1}{m}
    \left[
        \sqrt{p(1-t)}
        +
        \sqrt{qt}
    \right]^2,
    \qquad
    0\le t\le\frac{q}{m}.
\end{equation*}
The map
\[
    t\longmapsto
    \sqrt{p(1-t)}+\sqrt{qt}
\]
increases strictly on \([0,q/m]\), with zero derivative only at the right
endpoint:
\begin{equation*}
    \frac{\mathrm{d}}{\mathrm{d}t}
    \left[
        \sqrt{p(1-t)}+\sqrt{qt}
    \right]
    =
    -\frac{\sqrt p}{2\sqrt{1-t}}
    +
    \frac{\sqrt q}{2\sqrt t},
\end{equation*}
which is nonnegative exactly for \(t\le q/m\). The function \(h_r\) therefore
increases on its domain and
\begin{equation*}
    h_r(0)=\frac{p}{m},
    \qquad
    h_r\left(\frac{q}{m}\right)=1.
\end{equation*}

The argument separates at \(\lambda=p/m\).

\medskip
\noindent
\emph{Low-fidelity regime: \(1/m\le\lambda\le p/m\).}

Here a rank-at-most-\(p\) spectrum can realize the constraint with \(t=0\).
For \(p\ge2\), consider
\begin{equation*}
    \boldsymbol{\mu}(x)
    =
    \left(
        x,
        \underbrace{
        \frac{1-x}{p-1},\ldots,
        \frac{1-x}{p-1}
        }_{p-1\ \mathrm{entries}},
        \underbrace{0,\ldots,0}_{q\ \mathrm{entries}}
    \right),
    \quad
    \frac1p\le x\le1.
\end{equation*}
This spectrum is ordered. The continuous map
\begin{equation*}
    x\longmapsto
    \frac1m
    \left[
        \sqrt{x}
        +
        \sqrt{(p-1)(1-x)}
    \right]^2
\end{equation*}
takes the values \(p/m\) at \(x=1/p\) and \(1/m\) at
\(x=1\). It therefore covers the entire interval
\([1/m,p/m]\). When \(p=1\), this interval reduces to the
single point \(\lambda=1/m\), which is realized by the product
spectrum \((1,0,\ldots,0)\).

Every \(\lambda\in[1/m,p/m]\) can therefore be realized at \(t=0\), and
\begin{equation*}
    \mathcal{F}_r(\lambda)=0.
\end{equation*}

\medskip
\noindent
\emph{High-fidelity regime: \(p/m<\lambda\le1\).}

Here \(t>0\), and Eq.~\eqref{eq:lambda-t-inequality-main} gives
\begin{equation*}
    \lambda\le h_r(t).
\end{equation*}
Monotonicity of \(h_r\) on \([0,q/m]\) gives
\begin{equation*}
    t\ge h_r^{-1}(\lambda).
\end{equation*}

To invert $h_r$, impose equality in \eqref{eq:lambda-t-inequality-main}:
\begin{equation}
    \sqrt{m\lambda}
    =
    \sqrt{p(1-t)}
    +
    \sqrt{qt}.
    \label{eq:inverse-equation-main}
\end{equation}
Because \(t\le q/m\), \(\sqrt{q(1-t)}-\sqrt{pt}\geq0\). Squaring
\eqref{eq:inverse-equation-main} and using \(p+q=m\) gives
\begin{equation}
    \sqrt{m(1-\lambda)}
    =
    \sqrt{q(1-t)}
    -
    \sqrt{pt}.
    \label{eq:complementary-inverse-main}
\end{equation}
Equations \eqref{eq:inverse-equation-main} and
\eqref{eq:complementary-inverse-main} can be written as
\begin{equation*}
    \begin{pmatrix}
        \sqrt{\lambda}\\[0.2em]
        \sqrt{1-\lambda}
    \end{pmatrix}
    =
    \frac1{\sqrt m}
    \begin{pmatrix}
        \sqrt p & \sqrt q\\
        \sqrt q & -\sqrt p
    \end{pmatrix}
    \begin{pmatrix}
        \sqrt{1-t}\\[0.2em]
        \sqrt t
    \end{pmatrix}.
\end{equation*}
The matrix is orthogonal and symmetric; applying it again gives
\begin{equation*}
    \sqrt{t}
    =
    \frac1{\sqrt m}
    \left[
        \sqrt{q\lambda}
        -
        \sqrt{p(1-\lambda)}
    \right].
\end{equation*}
Squaring gives
\begin{equation}
    h_r^{-1}(\lambda)
    =
    \frac1m
    \left[
        \sqrt{q\lambda}
        -
        \sqrt{p(1-\lambda)}
    \right]^2.
    \label{eq:inverse-solution-main}
\end{equation}
The expression in square brackets is positive precisely when
\(\lambda>p/m\).

For attainability, take \(t\) from \eqref{eq:inverse-solution-main} and set
\begin{equation*}
    \mu_1=\cdots=\mu_p=\frac{1-t}{p},
    \qquad
    \mu_{p+1}=\cdots=\mu_m=\frac{t}{q}.
\end{equation*}
Because \(t\le q/m\), this spectrum is ordered. It also saturates
\eqref{eq:head-CS-main} and \eqref{eq:tail-CS-main}, so
\begin{equation*}
    \frac1m
    \left(
        \sum_{i=1}^{m}\sqrt{\mu_i}
    \right)^2
    =
    h_r(t)
    =
    \lambda.
\end{equation*}
Its tail weight is \(t\), proving
\begin{equation*}
    \mathcal{F}_r(\lambda)
    =
    \frac1m
    \left[
        \sqrt{q\lambda}
        -
        \sqrt{p(1-\lambda)}
    \right]^2.
\end{equation*}
Substituting \(p=r-1\) and \(q=m-r+1\) gives
\eqref{eq:Fr-exact-main} and the optimizer
\eqref{eq:two-block-optimal-spectrum-main}.

Expanding the square and using \(p+q=m\) gives the equivalent form
\begin{align*}
    \frac1m
    \left[
        \sqrt{q\lambda}
        -
        \sqrt{p(1-\lambda)}
    \right]^2
    &=
    1-
    \frac1m
    \left[
        \sqrt{p\lambda}
        +
        \sqrt{q(1-\lambda)}
    \right]^2,
\end{align*}
which proves \eqref{eq:Fr-equivalent-main}.

Equality in both Cauchy--Schwarz steps requires blockwise uniformity. The
ordered nonzero-branch minimizer is therefore uniquely
\eqref{eq:two-block-optimal-spectrum-main}.
\end{proof}

\section{Proof of Theorem \ref{thm:global-weighted-multiscale}}
\begin{proof}
The proof has four steps.

\emph{Step 1: Convexity of the optimization problem.}
Expanding $H_{\widetilde{\boldsymbol d}}$ gives
\[
    H_{\widetilde{\boldsymbol d}}(\boldsymbol w)
    =
    \frac{1}{m}
    \left[
        \sum_{a=1}^{L}\widetilde d_aw_a
        +
        2\sum_{a<b}
        \sqrt{\widetilde d_a\widetilde d_b}
        \sqrt{w_aw_b}
    \right].
\]
The first term is linear, and each $\sqrt{w_aw_b}$ is concave on the
nonnegative orthant. It follows that $H_{\widetilde{\boldsymbol d}}$ is concave and
\[
    \left\{
        \boldsymbol w:
        H_{\widetilde{\boldsymbol d}}(\boldsymbol w)\geq\lambda
    \right\}
\]
is convex. The linear objective makes the optimization convex.

For $\lambda<D/m$, the vector
\[
    \widehat w_a
    =
    \frac{\widetilde d_a}{D}
\]
is strictly feasible because
\[
    H_{\widetilde{\boldsymbol d}}(\widehat{\boldsymbol w})
    =
    \frac{D}{m}
    >
    \lambda.
\]
This strictly feasible point establishes Slater's condition, so the KKT
conditions characterize global optimality.

\emph{Step 2: KKT equations.}
Introduce the Lagrangian
\begin{widetext}
\[
    \mathcal L
    (\boldsymbol w,\eta,\beta,\boldsymbol\gamma)
    =
    \sum_{a=1}^{L}c_aw_a
    +
    \eta
    \left(
        \sum_{a=1}^{L}w_a-1
    \right)
    +
    \beta
    \left[
        \lambda-H_{\widetilde{\boldsymbol d}}(\boldsymbol w)
    \right]
    -
    \sum_{a=1}^{L}\gamma_aw_a,
\]
\end{widetext}
where
\[
    \beta\geq0,
    \qquad
    \gamma_a\geq0.
\]
Besides primal and dual feasibility, the KKT conditions require
\[
    \beta
    \left[
        \lambda-H_{\widetilde{\boldsymbol d}}(\boldsymbol w)
    \right]
    =0,
    \qquad
    \gamma_aw_a=0,
\]
and stationarity.

For an interior feasible point, set
\[
    S(\boldsymbol w)
    :=
    \sum_{b=1}^{L}
    \sqrt{\widetilde d_bw_b}.
\]
Then
\[
    \frac{\partial H_{\widetilde{\boldsymbol d}}}
    {\partial w_a}
    =
    \frac{S(\boldsymbol w)}{m}
    \frac{\sqrt{\widetilde d_a}}{\sqrt{w_a}}.
\]
For an interior candidate, $\gamma_a=0$, and stationarity becomes
\[
    c_a+\eta
    =
    \beta
    \frac{S(\boldsymbol w)}{m}
    \frac{\sqrt{\widetilde d_a}}{\sqrt{w_a}}.
\]
Writing $\alpha:=\eta$, this relation implies
\[
    w_a
    \propto
    \frac{\widetilde d_a}{(c_a+\alpha)^2}.
\]
Step 3 gives a unique \(\alpha>0\) for each
\(D_0/m<\lambda<D/m\); below we verify the KKT conditions with
\(\eta=\alpha\) and \(\beta=m/Z_1(\alpha)>0\).
Normalizing this proportionality gives
\[
    w_a^\star(\alpha)
    =
    \frac{
        \widetilde d_a/(c_a+\alpha)^2
    }{
        Z_2(\alpha)
    }.
\]

For this vector,
\[
    \sum_{a=1}^{L}
    \sqrt{\widetilde d_aw_a^\star(\alpha)}
    =
    \frac{Z_1(\alpha)}
    {\sqrt{Z_2(\alpha)}}.
\]
The corresponding constraint value is
\[
    H_{\widetilde{\boldsymbol d}}
    \bigl(\boldsymbol w^\star(\alpha)\bigr)
    =
    \frac{Z_1(\alpha)^2}
    {mZ_2(\alpha)}.
\]
Choose $\alpha$ so this expression equals $\lambda$.

Stationarity holds for
\[
    \eta=\alpha,
    \qquad
    \beta=\frac{m}{Z_1(\alpha)},
    \qquad
    \gamma_a=0.
\]
At $\boldsymbol w^\star(\alpha)$,
\[
    \frac{\partial H_{\widetilde{\boldsymbol d}}}
    {\partial w_a}
    =
    \frac{Z_1(\alpha)}{m}
    \bigl(c_a+\alpha\bigr),
\]
and hence
\[
    c_a+\eta
    -
    \beta
    \frac{\partial H_{\widetilde{\boldsymbol d}}}
    {\partial w_a}
    =
    0.
\]
All KKT conditions are therefore satisfied.

\emph{Step 3: Existence and uniqueness of $\alpha$.}
Define
\[
    Z_3(\alpha)
    :=
    \sum_{a=1}^{L}
    \frac{\widetilde d_a}{(c_a+\alpha)^3}.
\]
Using
\[
    Z_1'(\alpha)=-Z_2(\alpha),
    \qquad
    Z_2'(\alpha)=-2Z_3(\alpha),
\]
we obtain
\[
    \frac{d}{d\alpha}
    \log
    \left[
        \frac{Z_1(\alpha)^2}{mZ_2(\alpha)}
    \right]
    =
    \frac{
        2\bigl[
            Z_1(\alpha)Z_3(\alpha)-Z_2(\alpha)^2
        \bigr]
    }{
        Z_1(\alpha)Z_2(\alpha)
    }.
\]
The Cauchy--Schwarz inequality gives
\[
    Z_2(\alpha)^2
    \leq
    Z_1(\alpha)Z_3(\alpha).
\]
Equality would require all costs $c_a$ to coincide, which is excluded. The map
\[
    \alpha
    \longmapsto
    \frac{Z_1(\alpha)^2}{mZ_2(\alpha)}
\]
is therefore strictly increasing.

At the endpoints,
\[
    \lim_{\alpha\rightarrow0^+}
    \frac{Z_1(\alpha)^2}{mZ_2(\alpha)}
    =
    \frac{D_0}{m},
\]
because the zero-cost blocks dominate both $Z_1$ and $Z_2$,
whereas
\[
    \lim_{\alpha\rightarrow\infty}
    \frac{Z_1(\alpha)^2}{mZ_2(\alpha)}
    =
    \frac{D}{m}.
\]
The endpoint limits consequently give, for every
\[
    \frac{D_0}{m}<\lambda<\frac{D}{m},
\]
there is a unique $\alpha>0$ satisfying the constraint.

\emph{Step 4: Explicit global-optimality certificate.}
KKT gives global optimality; the following certificate proves uniqueness.
For any $\alpha>0$ and any feasible
$\boldsymbol w$, the Cauchy--Schwarz inequality gives
\begin{align*}
    \left(
        \sum_{a=1}^{L}
        \sqrt{\widetilde d_aw_a}
    \right)^2
    &=
    \left[
        \sum_{a=1}^{L}
        \sqrt{
            \frac{\widetilde d_a}{c_a+\alpha}
        }
        \sqrt{
            (c_a+\alpha)w_a
        }
    \right]^2                                                   \\
    &\leq
    Z_1(\alpha)
    \left(
        \sum_{a=1}^{L}c_aw_a+\alpha
    \right).
\end{align*}
Feasibility implies
\[
    \left(
        \sum_{a=1}^{L}
        \sqrt{\widetilde d_aw_a}
    \right)^2
    \geq
    m\lambda.
\]
Combining these inequalities gives, for every feasible vector,
\[
    \sum_{a=1}^{L}c_aw_a
    \geq
    \frac{m\lambda}{Z_1(\alpha)}-\alpha.
\]

For the unique parameter $\alpha$ determined above,
\[
    m\lambda
    =
    \frac{Z_1(\alpha)^2}{Z_2(\alpha)},
\]
so the lower bound becomes
\[
    \sum_{a=1}^{L}c_aw_a
    \geq
    \frac{Z_1(\alpha)}{Z_2(\alpha)}-\alpha.
\]
The candidate attains it:
\begin{align*}
    \sum_{a=1}^{L}
    c_aw_a^\star(\alpha)
    &=
    \frac{
        \displaystyle
        \sum_{a=1}^{L}
        \frac{c_a\widetilde d_a}{(c_a+\alpha)^2}
    }{
        Z_2(\alpha)
    }                                                         \\
    &=
    \frac{
        Z_1(\alpha)-\alpha Z_2(\alpha)
    }{
        Z_2(\alpha)
    }                                                         \\
    &=
    \frac{Z_1(\alpha)}{Z_2(\alpha)}-\alpha.
\end{align*}
The vector $\boldsymbol w^\star(\alpha)$ is therefore globally optimal.

Equality in the Cauchy--Schwarz inequality used above requires
\[
    w_a
    \propto
    \frac{\widetilde d_a}{(c_a+\alpha)^2}
\]
for every $a$. Since $\widetilde d_a,\alpha>0$, every component is positive.
Thus no simplex-boundary vector can attain the minimum on the nontrivial branch;
vectors supported only on zero-cost blocks are also infeasible because
\(\lambda>D_0/m\). Uniqueness of $\alpha$ then proves a unique full-support
minimizer under the hypotheses of Theorem~\ref{thm:global-weighted-multiscale}.
\end{proof}

\section{Proof of Theorem \ref{thm:negativity-geometric}}

\begin{proof}
For a pure state, Eq.~\eqref{eq:pure-negativity-bound} gives
\[
\mathcal N_{\rm PT}(|\psi\rangle\langle\psi|)
\leq g_{k,m}(E_k(|\psi\rangle))-1.
\]
The function $g_{k,m}-1$ is continuous, nondecreasing, and concave on
$[0,q/m]$. The first part of Lemma~\ref{lem:convex-roof-transfer} therefore
implies
\[
\mathcal N_c(\rho)\leq g_{k,m}(E_k(\rho))-1.
\]
For every pure-state decomposition of $\rho$, convexity of the trace norm gives
\[
\mathcal N_{\rm PT}(\rho)
\leq\sum_jp_j\mathcal N_{\rm PT}(|\psi_j\rangle\langle\psi_j|).
\]
Taking the infimum over decompositions yields
$\mathcal N_{\rm PT}(\rho)\leq\mathcal N_c(\rho)$ and completes the proof.
\end{proof}

\clearpage

\bibliography{references}

\end{document}